\documentclass[preprint,11pt]{article}
\usepackage{amsfonts}
\usepackage{color}
\usepackage{graphicx}
\usepackage[dvips]{epsfig}
\usepackage{graphics} 
\usepackage{times} 
\usepackage[cmex10]{amsmath} 
\usepackage{cite}
\usepackage{multirow}
\usepackage[tight,footnotesize]{subfigure}
\usepackage{algorithm}
\usepackage{amsthm}

\def\norm #1{\left\|#1\right\|}

\def\twon #1{\left\|#1\right\|_2}
\def\onen #1{\left\|#1\right\|_1}
\def\zeron #1{\left\|#1\right\|_0}
\def\frobn #1{\left\|#1\right\|_{\text{F}}}
\def\twonen #1{\left\|#1\right\|_{2,1}}

\def\atomn #1{\left\|#1\right\|_{\cA}}

\def\vect #1{\text{vec}#1}
\def\cov #1{\text{Cov}#1}
\def\abs #1{\left|#1\right|}
\def\inp #1{\left\langle#1\right\rangle}

\def\st{\text{subject to }}

\def\bC{\mathbb{C}}

\def\bR{\mathbb{R}}

\def\bT{\mathbb{T}}
\def\bE{\mathbb{E}}

\def\m #1{\boldsymbol{#1}}

\def\cA{\mathcal{A}}

\def\cD{\mathcal{D}}

\def\cK{\mathcal{K}}
\def\cL{\mathcal{L}}
\def\cM{\mathcal{M}}

\def\cT{\mathcal{T}}
\def\cQ{\mathcal{Q}}

\def\cS{\mathcal{S}}
\def\cT{\mathcal{T}}

\def\bee{\begin{equation}}
\def\ene{\end{equation}}

\def\beq{\begin{eqnarray}}
\def\enq{\end{eqnarray}}
\def\lentwo{\setlength\arraycolsep{2pt}}

\newtheorem{rem}{Remark}[section]
\newtheorem{cor}{Corollary}[section]
\newtheorem{thm}{Theorem}[section]

\newtheorem{defi}{Definition}[section]

\def\equ #1{\begin{equation}#1\end{equation}}
\def\equa #1{\begin{eqnarray}#1\end{eqnarray}}
\def\sbra #1{\left(#1\right)}
\def\mbra #1{\left[#1\right]}
\def\lbra #1{\left\{#1\right\}}
\def\diag #1{\text{diag}#1}
\def\tr #1{\text{Tr}#1}

\def\rank #1{\text{rank}#1}
\def\spark #1{\text{spark}#1}
\def\st {\text{ subject to }}

\title{Sparse Methods for Direction-of-Arrival Estimation}
\author{Zai Yang\thanks{School of Automation, Nanjing University of Science and Technology, Nanjing 210094, China} \thanks{School of Electrical and Electronic Engineering, Nanyang Technological University, Singapore 639798} ,\; Jian Li\thanks{Department of Electrical and Computer Engineering, University of Florida, Gainesville, FL 32611, USA} ,\; Petre Stoica\thanks{Department of Information Technology, Uppsala University, Uppsala, SE 75105, Sweden} ,\; and Lihua Xie$^\dag$}

\begin{document}
\maketitle

\tableofcontents
\newpage


\section{Introduction}

Direction-of-arrival (DOA) estimation refers to the process of retrieving the direction information of several electromagnetic waves/sources from the outputs of a number of receiving antennas that form a sensor array. DOA estimation is a major problem in array signal processing and has wide applications in radar, sonar, wireless communications, etc.

The study of DOA estimation methods has a long history. For example, the conventional (Bartlett) beamformer, which dates back to the World War II, simply uses Fourier-based spectral analysis of the spatially sampled data. Capon's beamformer was later proposed to improve the estimation performance of closely spaced sources \cite{capon1969high}. Since the 1970s when Pisarenko found that the DOAs can be retrieved from data second order statistics \cite{pisarenko1973retrieval}, a prominent class of methods designated as subspace-based methods have been developed, e.g., the multiple signal classification
(MUSIC) and the estimation of parameters by rotational invariant techniques (ESPRIT) along with their variants \cite{schmidt1981signal,schmidt1986multiple,paulraj1986subspace,roy1989esprit, barabell1983improving}. Another common approach is the nonlinear least squares (NLS) method that is also known as the (deterministic) maximum likelihood estimation. For a complete review of these methods, readers are referred to \cite{krim1996two,stoica2005spectral, zoubir2014academic}. Note that these methods suffer from certain well-known limitations. For example, the subspace-based methods and the NLS need {\em a priori} knowledge on the source number that may be difficult to obtain. Additionally, Capon's beamformer, MUSIC and ESPRIT are covariance-based and require a sufficient number of data snapshots to accurately estimate the data covariance matrix. Moreover, they can be sensitive to source correlations that tend to cause a rank deficiency in the sample data covariance matrix. Also, a very accurate initialization is required for the NLS since its objective function has a complicated multimodal shape with a sharp global minimum.

The purpose of this article is to provide an overview of the recent work on sparse DOA estimation methods. These new methods are motivated by techniques in sparse representation and compressed sensing methodology \cite{chen2001atomic,donoho2003optimally, candes2006robust,donoho2006compressed,baraniuk2007compressive}, and most of them have been proposed during the last decade. The sparse estimation (or optimization) methods can be applied in several demanding scenarios, including cases with no knowledge of the source number, limited number of snapshots (even a single snapshot), and highly or completely correlated sources. Due to these attractive properties they have been extensively studied and their popularity is reflected by the large number of publications about them.

It is important to note that there is a key difference between the common sparse representation framework and DOA estimation. To be specific, the studies of sparse representation have been focused on {\em discrete linear} systems. In contrast to this, the DOA parameters are {\em continuous} valued and the observed data are {\em nonlinear} in the DOAs. Depending on the model adopted, we can classify the sparse methods for DOA estimation into three categories, namely, {\em on-grid}, {\em off-grid} and {\em gridless}, which also corresponds to the chronological order in which they have been developed. For on-grid sparse methods, the data model is obtained by assuming that the true DOAs lie on a set of fixed grid points in order to straightforwardly apply the existing sparse representation techniques. While a grid is still required by off-grid sparse methods, the DOAs are not restricted to be on the grid. Finally, the recent gridless sparse methods do not need a grid, as their name suggests, and they operate directly in the continuous domain.


The organization of this article is as follows. The data model for DOA estimation is introduced in Section \ref{sec:model} for far-field, narrowband sources that are the focus of this article. Its dependence on the array geometry and the parameter identifiability problem are discussed. In Section \ref{sec:SSR} the concepts of sparse representation and compressed sensing are introduced and several sparse estimation techniques are discussed. Moreover, we discuss the feasibility of using sparse representation techniques for DOA estimation and highlight the key differences between sparse representation and DOA estimation. The on-grid sparse methods for DOA estimation are introduced in Section \ref{sec:ongrid}. Since they are straightforward to obtain in the case of a single data snapshot, we focus on showing how the temporal redundancy of multiple snapshots can be utilized to improve the DOA estimation performance. Then, the off-grid sparse methods are presented in Section \ref{sec:offgrid}. Section \ref{sec:gridless} is the main highlight of this article in which the recently developed gridless sparse methods are presented. These methods are of great interest since they operate directly in the continuous domain and have strong theoretical guarantees. Some future research directions will be discussed in Section \ref{sec:futurework} and conclusions will be drawn in Section \ref{sec:conclusion}.


Notations used in this article are as follows. $\bR$ and $\bC$ denote the sets of real and complex numbers respectively. Boldface letters are reserved for vectors and matrices. $\abs{\cdot}$ denotes the amplitude of a scalar or the cardinality of a set. $\onen{\cdot}$, $\twon{\cdot}$ and $\frobn{\cdot}$ denote the
$\ell_1$, $\ell_2$ and Frobenius norms respectively. $\m{A}^T$, $\m{A}^*$ and $\m{A}^H$ are the matrix transpose, complex conjugate and conjugate transpose of $\m{A}$ respectively. $x_j$ is the $j$th entry of a vector $\m{x}$, and $\m{A}_j$ denotes the $j$th row of a matrix $\m{A}$. Unless otherwise stated, $\m{x}_{\Omega}$ and $\m{A}_{\Omega}$ are the subvector and submatrix of $\m{x}$ and $\m{A}$ obtained by retaining the entries of $\m{x}$ and the rows of $\m{A}$ indexed by the set $\Omega$. For a vector $\m{x}$, $\diag\sbra{\m{x}}$ is a diagonal matrix with $\m{x}$ on the diagonal. $\m{x}\succeq\m{0}$ means $x_j\geq0$ for all $j$. $\rank\sbra{\m{A}}$ denotes the rank of a matrix $\m{A}$ and $\tr\sbra{\m{A}}$ denotes the trace. For positive semidefinite matrices $\m{A}$ and $\m{B}$, $\m{A}\geq\m{B}$ means that $\m{A}-\m{B}$ is positive semidefinite. Finally, $\bE\mbra{\cdot}$ denotes the expectation of a random variable, and for notational simplicity a random variable and its numerical value will not be distinguished.

\section{Data Model} \label{sec:model}

\subsection{Data Model}
In this section, the DOA estimation problem is stated. Consider $K$ narrowband far-field source signals $s_k$, $k=1,\dots,K$, impinging on an array of omnidirectional sensors from directions $\theta_k$, $k=1,\dots,K$. According to \cite{krim1996two,stoica2005spectral,chung2014doa}, the time delays at different sensors can be represented by simple phase shifts, resulting in the following data model:
\equ{\m{y}(t) = \sum_{k=1}^K\m{a}\sbra{\theta_k}s_k(t) +\m{e}(t)= \m{A}(\m{\theta})\m{s}(t)+\m{e}(t),\quad t=1,\dots,L,\label{formu:observation_model1}}
where $t$ indexes the snapshot and $L$ is the number of snapshots, $\m{y}(t)\in\bC^M$, $\m{s}(t)\in\bC^K$ and $\m{e}(t)\in\bC^M$ denote the array output, the vector of source signals and the vector of measurement noise at snapshot $t$, respectively, where $M$ is the number of sensors. $\m{a}\sbra{\theta_k}$ is the so-called steering vector of the $k$th source that is determined by the geometry of the sensor array and will be given later. The steering vectors compose the array manifold matrix $\m{A}(\m{\theta})= \mbra{\m{a}\sbra{\theta_1},\dots,\m{a}\sbra{\theta_K}}$. More compactly, (\ref{formu:observation_model1}) can be written as
\equ{\m{Y} = \m{A}(\m{\theta})\m{S}+\m{E},\label{formu:observation_model}}
where $\m{Y}=\mbra{\m{y}(1),\dots,\m{y}(L)}$, and $\m{S}$ and $\m{E}$ are similarly defined. Given the data matrix $\m{Y}$ and the mapping $\theta\rightarrow \m{a}(\theta)$, the objective is to estimate the parameters $\theta_k$, $k=1,\dots,K$ that are referred to as the DOAs. It is worth noting that the source number $K$ is usually unknown in practice; typically, $K$ is assumed to be smaller than $M$, as otherwise the DOAs cannot be uniquely identified from the data (see details in Subsection \ref{sec:ident}).

\subsection{The Role of Array Geometry} \label{sec:model_geometry}

\begin{figure}
\centering
  \includegraphics[width=3.5in]{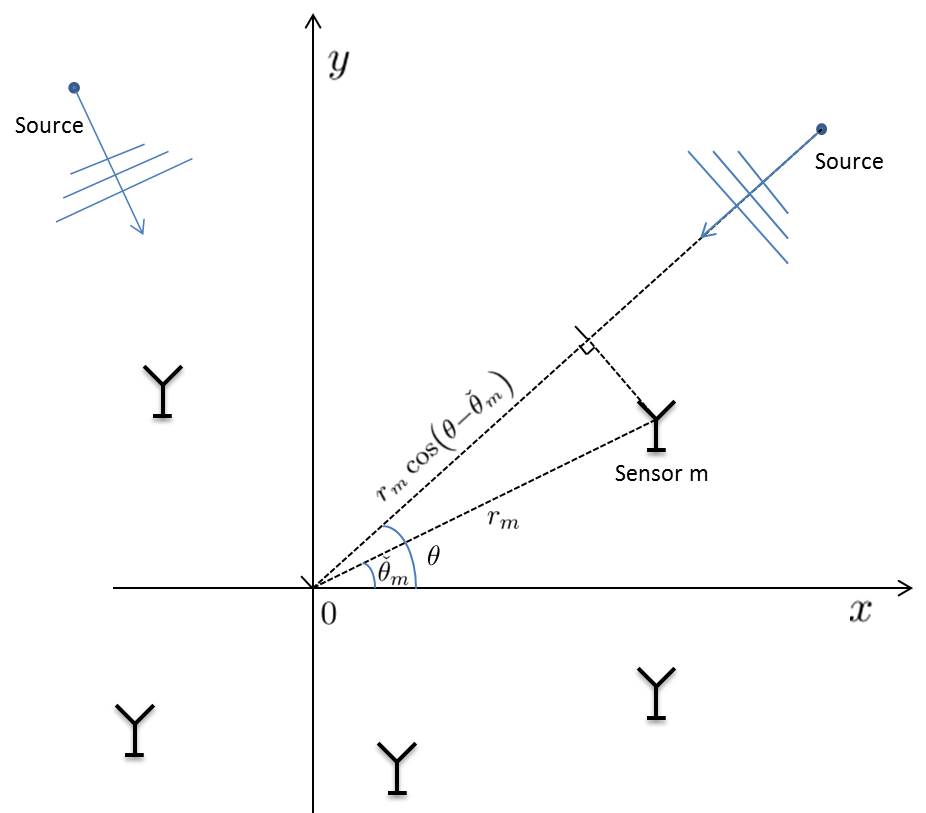}
\centering
\caption{The setup of the DOA estimation problem with a general 2-D array.} \label{Fig:model}
\end{figure}

We now discuss how the mapping $\theta\rightarrow \m{a}(\theta)$ is determined by the array geometry. We first consider a general 2-D array with the $M$ sensors located at points $\sbra{r_m, \check{\theta}_m}$, expressed in polar coordinates. For convenience, the unit of distance is taken as half the wavelength of the waves. Then $\m{a}(\theta)$ will be given by
\equ{a_m(\theta) = e^{i\pi r_m \cos\sbra{\theta-\check{\theta}_m} }, \quad m=1,\dots,M.}

In the particularly interesting case of a linear array, assuming that the sensors are located on the nonnegative $x$-axis, we have that $\check{\theta}_m = 0^\circ$, $m=1,\dots,M$. Therefore, $\m{a}(\theta)$ will be given by
\equ{a_m(\theta) = e^{i\pi r_m \cos\theta}, \quad m=1,\dots,M.}
We can replace the variable $\theta$ by $f = \frac{1}{2}\cos\theta$ and define without ambiguity $\m{a}(f)=\m{a}(\arccos\sbra{2f})=\m{a}(\theta)$. Then, the mapping $\m{a}(f)$ is given by
\equ{a_m(f) = e^{i2\pi r_m f}, \quad m=1,\dots,M.}
In the case of a single snapshot, obviously, the {\em spatial} DOA estimation problem becomes a {\em temporal} frequency estimation problem (a.k.a. line spectral estimation) given the samples $y_m$, $m=1,\dots,M$ measured at time instants $r_m$, $m=1,\dots,M$.

\begin{figure}
\centering
  \includegraphics[width=4in]{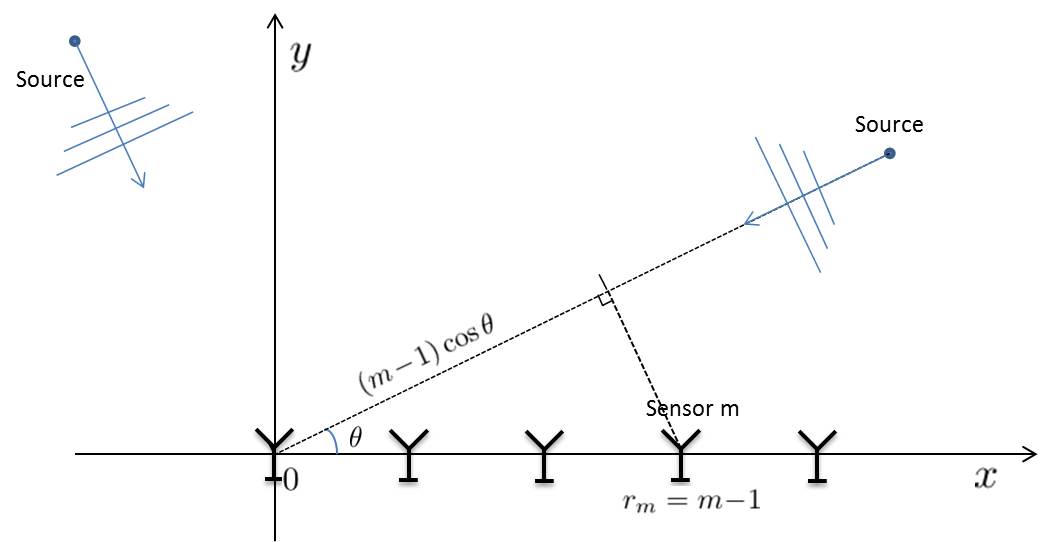}
\centering
\caption{The setup of the DOA estimation problem with a ULA.} \label{Fig:model}
\end{figure}

If we further assume that the sensors of the linear array are equally spaced, then the array is known as a uniform linear array (ULA). We consider the case when two adjacent antennas of the array are spaced by a unit distance (half the wavelength). Then, we have that $r_m = m-1$ and
\equ{\m{a}(f) = \mbra{1, e^{i2\pi f},\dots, e^{i2\pi (M-1) f}}^T.}
If a linear array is obtained from a ULA by retaining only a part of the sensors, then it is known as a sparse linear array (SLA).

It is worth noting that for a 2-D array, it is possible to estimate the DOAs in the entire $360^\circ$ range, while for a linear array we can only resolve the DOAs in a $180^\circ$ range: $\theta_k\in\left[0^{\circ},180^{\circ}\right)$. In the latter case, correspondingly, the ``frequencies'' are: $f_k=\frac{1}{2}\cos\theta_k \in\left(-\frac{1}{2},\frac{1}{2}\right]$. Throughout this article, we let $\cD_{\theta}$ denote the domain of the DOAs that can be $\left[0^{\circ},360^{\circ}\right)$ or $\left[0^{\circ},180^{\circ}\right)$, depending on the context. Also, let $\bT=\left(-\frac{1}{2},\frac{1}{2}\right]$ be the frequency interval for linear arrays.
Finally, we close this section by noting that the grid-based (on-grid and off-grid) sparse methods can be applied to arbitrary sensor arrays, while the existing gridless sparse methods are typically limited to ULAs or SLAs.

\subsection{Parameter Identifiability} \label{sec:ident}
The DOAs $\lbra{\theta_k}_{k=1}^K$ can be uniquely identified from $\m{Y}=\m{A}\sbra{\m{\theta}}\m{S}$ if there do not exist $\lbra{\theta_k'}_{k=1}^K \neq \lbra{\theta_k}_{k=1}^K$ and $\m{S}'$ such that $\m{Y}=\m{A}\sbra{\m{\theta}}\m{S} = \m{A}\sbra{\m{\theta}'}\m{S}'$. Guaranteeing that the parameters can be uniquely identified in the noiseless case is usually a prerequisite for their accurate estimation. The parameter identifiability problem for DOA estimation was studied in \cite{bresler1986number} for ULAs and in \cite{wax1989unique,nehorai1991direction} for general arrays. The results in \cite{chen2006theoretical,davies2012rank,yang2016exact} are also closely related to this problem. For a general array, define the set
\equ{\cA_{\theta} = \lbra{\m{a}\sbra{\theta}: \; \theta\in\cD_{\theta}}, \label{eq:cAtheta}}
and define the spark of $\cA_{\theta}$, denoted by $\text{spark}\sbra{\cA_{\theta}}$, as
the smallest number of elements in $\cA_{\theta}$ that are linearly dependent \cite{kruskal1977three}. For any $M$-element array, it holds that
\equ{\text{spark}\sbra{\cA_{\theta}}\leq M+1. \label{eq:sparkAtheta}}
Note that it is generally difficult to compute $\spark\sbra{\cA_{\theta}}$, except in the ULA case in which $\spark\sbra{\cA_{\theta}}=M+1$ by the fact that any $M$ steering vectors in $\cA_{\theta}$ are linearly independent.

The paper \cite{wax1989unique} showed that any $K$ sources can be uniquely identified from $\m{Y}=\m{A}\sbra{\m{\theta}}\m{S}$ provided that
\equ{K < \frac{\text{spark}\sbra{\cA_{\theta}} - 1 + \rank\sbra{\m{S}}}{2}. \label{eq:cond1}}
Note that the above condition cannot be easily used in practice since it requires knowledge on $\m{S}$. To resolve this problem, it was shown in \cite{davies2012rank} that the condition in \eqref{eq:cond1} is equivalent to
\equ{K < \frac{\text{spark}\sbra{\cA_{\theta}} - 1 + \rank\sbra{\m{Y}}}{2}. \label{eq:cond2}}
Moreover, the condition in \eqref{eq:cond1} or \eqref{eq:cond2} is also necessary \cite{davies2012rank}. Combining these results, we have the following theorem.

\begin{thm} Any $K$ sources can be uniquely identified from $\m{Y}=\m{A}\sbra{\m{\theta}}\m{S}$ if and only if the condition in \eqref{eq:cond2} holds. \label{thm:uniformident}
\end{thm}

Theorem \ref{thm:uniformident} provides a necessary and sufficient condition for unique identifiability of the parameters. In the single snapshot case, the condition in \eqref{eq:cond2} reduces to
\equ{K < \frac{\text{spark}\sbra{\cA_{\theta}}}{2}. \label{eq:cond2_1}}
Therefore, Theorem \ref{thm:uniformident} implies that more sources can be determined if more snapshots are collected, except in the trivial case when the source signals are identical up to scaling factors.
In the ULA case, the condition in \eqref{eq:cond2} can be simplified as
\equ{K< \frac{M+\rank \sbra{\m{Y}}}{2}. \label{Kbound_ULA}}
Using the inequality $\rank\sbra{\m{S}}\leq K$ and \eqref{eq:sparkAtheta}, the condition in \eqref{eq:cond1} or \eqref{eq:cond2} implies that
\equ{K < \text{spark}\sbra{\cA_{\theta}}-1 \leq M.}

Theorem \ref{thm:uniformident} specifies the condition required to guarantee unique identifiability for \emph{any} $K$ source signals. It was shown in \cite{wax1989unique} that if $\lbra{\theta_k}$ are fixed and $\m{S}$ is randomly drawn from some absolutely continuous distribution, then the $K$ sources can be uniquely identified with probability one, provided that
\equ{K < \frac{2\rank\sbra{\m{S}}}{2\rank\sbra{\m{S}}+1} \sbra{\text{spark}\sbra{\cA_{\theta}}-1}. \label{eq:cond3}}
Moreover, the following condition, which is slightly different from that in \eqref{eq:cond3}, is necessary:
\equ{K \leq \frac{2\rank\sbra{\m{S}}}{2\rank\sbra{\m{S}}+1} \sbra{\text{spark}\sbra{\cA_{\theta}}-1}. \label{eq:cond4}}
The condition in \eqref{eq:cond3} is weaker than that in \eqref{eq:cond1} or \eqref{eq:cond2}. As an example, in the single snapshot case, the upper bounds on $K$ in \eqref{eq:cond2} and \eqref{eq:cond3} are approximately $\frac{1}{2}\text{spark}\sbra{\cA_{\theta}}$ and $\frac{2}{3}\text{spark}\sbra{\cA_{\theta}}$, respectively. However, the paper \cite{nehorai1991direction} pointed out that the condition in \eqref{eq:cond3} has a relatively limited practical relevance in finite-SNR applications, since under \eqref{eq:cond3}, with a strictly positive probability, false DOA estimates far from the true DOAs may occur.

\section{Sparse Representation and DOA estimation} \label{sec:SSR}
In this section we will introduce the basics of sparse representation that has been an active research topic especially in the past decade. More importantly, we will discuss its connections to and the key differences from DOA estimation.

\subsection{Sparse Representation and Compressed Sensing}

\subsubsection{Problem Formulation}
We first introduce the topic of sparse representation and the closely related concept of compressed sensing that have found broad applications in image, audio and signal processing, communications, medical imaging, and computational biology, to name just a few (see, e.g., the various special issues in several journals \cite{baraniuk2008compressive,chartrand2010introduction, baraniuk2010applications,starck2011introduction}). Let $\m{y}\in\bC^M$ be the signal that we observe. We want to sparsely represent $\m{y}$ via the following model:
\equ{\m{y} = \m{A} \m{x} + \m{e}, \label{eq:sparserep}}
where $\m{A}\in\bC^{M\times \overline{N}}$ is a {\em given} matrix, with $M\ll \overline{N}$, that is referred as a dictionary and whose columns are called atoms, $\m{x}\in\bC^{\overline{N}}$ is a sparse coefficient vector (note that the notation $N$ is reserved for later use), and $\m{e}$ accounts for the representation error. By sparsity we mean that only a few entries, say $K\ll \overline{N}$, of $\m{x}$ are nonzero and the rest are zero. This together with \eqref{eq:sparserep} implies that $\m{y}$ can be well approximated by a linear combination of $K$ atoms in $\m{A}$. The underlying motivation for the sparse representation is that even though the observed data $\m{y}$ lies in a high-dimensional space, it can actually be well approximated in some lower-dimensional subspace ($K<M$). Given $\m{y}$ and $\m{A}$, the problem of sparse representation is to find the sparse vector $\m{x}$ subject to data consistency.

The concept of sparse representation was later extended within the framework of compressed sensing \cite{candes2006robust,donoho2006compressed,baraniuk2007compressive}. In compressed sensing, a sparse signal, represented by the sparse vector $\m{x}$, is recovered from undersampled linear measurements $\m{y}$, i.e., the system model \eqref{eq:sparserep} applies with $M \ll \overline{N}$. In this context, $\m{y}$ is referred to as the compressive data, $\m{A}$ is the sensing matrix, and $\m{e}$ denotes the measurement noise. Note that a data model similar to \eqref{eq:sparserep} applies if the signal of interest is sparse in some domain. Given $\m{y}$ and $\m{A}$, the problem of sparse signal recovery in compressed sensing is also to solve for the sparse vector $\m{x}$ subject to data consistency. With no rise for ambiguity we will not distinguish between the terminologies used for sparse representation and compressed sensing, as these two problems are very much alike.

To solve for the sparse signal, intuitively, we should find the sparsest solution. In the absence of noise, therefore, we should solve the following optimization problem:
\equ{\min_{\m{x}} \zeron{\m{x}} \st \m{y}=\m{A}\m{x}, \label{eq:ell0min}}
where $\zeron{\m{x}}= \#\lbra{j:\; x_j\neq 0}$ counts the nonzero entries of $\m{x}$ and is referred to as the $\ell_0$ (pseudo-)norm or the sparsity of $\m{x}$. View $\m{A}$ as the set of its column vectors, and define its spark, denoted by $\text{spark}\sbra{\m{A}}$ as in Subsection \ref{sec:ident}. It can be shown that the true sparse signal $\m{x}$ can be uniquely determined by \eqref{eq:ell0min} if $\m{x}$ has a sparsity of
\equ{K<\frac{\text{spark}\sbra{\m{A}}}{2}. \label{eq:Kcond_l0}}
To see this, suppose there exists $\m{x}'$ of sparsity $K'\leq K$ satisfying $\m{y}=\m{A}\m{x}'$ as well. Then it holds that $\m{A}\sbra{\m{x}-\m{x}'}=\m{0}$. Since $\m{x}-\m{x}'$ has a sparsity of at most $K+K'\leq 2K<\text{spark}\sbra{\m{A}}$, which holds following \eqref{eq:Kcond_l0}, it can be concluded that $\m{x}-\m{x}'=\m{0}$ and thus $\m{x}=\m{x'}$ since any $\text{spark}\sbra{\m{A}}-1$ columns of $\m{A}$ are linearly independent.
It is interesting to note that the condition in \eqref{eq:Kcond_l0} is very similar to that in \eqref{eq:cond2_1} required to guarantee identifiability for DOA estimation in the single snapshot case.

Unfortunately, the $\ell_0$ optimization problem in \eqref{eq:ell0min} is NP hard to solve. Therefore, more efficient approaches are needed. We note that many methods and algorithms have been proposed for sparse signal recovery, e.g., convex relaxation or $\ell_1$ optimization \cite{chen2001atomic,donoho2003optimally}, $\ell_q$, $0<q<1$ (pseudo-)norm optimization \cite{gorodnitsky1997sparse,rao1999affine,rao2003subset,foucart2009sparsest, chartrand2007nonconvex,chartrand2007exact,tan2011sparse}, greedy algorithms such as orthogonal matching pursuit (OMP), compressive sampling matching pursuit (CoSaMP) and subspace pursuit (SP) \cite{pati1993orthogonal,donoho2012sparse,tropp2007signal,needell2009cosamp, dai2009subspace,davenport2010analysis}, iterative hard thresholding (IHT) \cite{blumensath2009iterative}, maximum likelihood estimation (MLE), etc. Readers can consult \cite{tropp2010computational} for a review. Here we introduce convex relaxation, $\ell_q$ optimization and MLE in the ensuing subsections.

\subsubsection{Convex Relaxation} \label{sec:l1}
The first practical approach to sparse signal recovery that we will introduce is based on the convex relaxation, which replaces the $\ell_0$ norm by its tightest convex relaxation---the $\ell_1$ norm. Therefore, we solve the following optimization problem in lieu of \eqref{eq:ell0min}:
\equ{\min \onen{\m{x}}, \st \m{y}=\m{A}\m{x}, \label{formu:l1}}
which is sometimes referred to as basis pursuit (BP) \cite{chen2001atomic}. Since the $\ell_1$ norm is convex, \eqref{formu:l1} can be solved in a polynomial time. In fact, the use of $\ell_1$ optimization for obtaining a sparse solution dates back to the paper \cite{claerbout1973robust} about seismic data recovery. While the BP was empirically observed to give good performance, a rigorous analysis had not been provided for decades.

To introduce the existing theoretical guarantees for the BP in \eqref{formu:l1}, we define a metric of the matrix $\m{A}$ called mutual coherence that quantifies the correlations between the atoms in $\m{A}$ \cite{donoho2003optimally}.

\begin{defi} The mutual coherence of a matrix $\m{A}$, $\mu\sbra{\m{A}}$, is the largest absolute correlation between any two columns of $\m{A}$, i.e.,
\equ{\mu\sbra{\m{A}}=\max_{i\neq j}\frac{\abs{\inp{\m{a}_i,\,\m{a}_j}}}{\twon{\m{a}_i}\twon{\m{a}_j}},}
where $\inp{\cdot,\cdot}$ denotes the inner product.
\end{defi}

Intuitively, if two atoms in $\m{A}$ are highly correlated, then it will be difficult to distinguish their contributions to the measurements $\m{y}$. In the extreme case when two atoms are completely coherent, it will be impossible to distinguish their contributions and thus impossible to recover the sparse signal $\m{x}$. Therefore, to guarantee successful signal recovery, the mutual coherence $\mu\sbra{\m{A}}$ should be small. This is true, according to the following theorem.

\begin{thm}[\cite{donoho2003optimally}] Assume that $\zeron{\m{x}}\leq K$ for the true signal $\m{x}$ and $\mu<\frac{1}{2K-1}$. Then, $\m{x}$ is the unique solution of the $\ell_0$ optimization and the BP problem. \label{thm:BP_coherence}
\end{thm}

Another theoretical guarantee is based on the restricted isometry property (RIP) that quantifies the correlations of the atoms in $\m{A}$ in a different manner and has been popular in the development of compressed sensing.

\begin{defi}[\cite{candes2006compressive}] The $K$-restricted isometry constant (RIC) of a matrix $\m{A}$, $\delta_K\sbra{\m{A}}$, is the smallest number such that the inequality
\equ{\sbra{1-\delta_K\sbra{\m{A}}}\twon{\m{v}}^2\leq\twon{\m{A}\m{v}}^2\leq \sbra{1+\delta_K\sbra{\m{A}}}\twon{\m{v}}^2\notag}
holds for all $K$-sparse vectors $\m{v}$. $\m{A}$ is said to satisfy the $K$-RIP with constant $\delta_K\sbra{\m{A}}$ if $\delta_K\sbra{\m{A}}<1$. \label{defi_RIP}
\end{defi}

By definition, matrices that have small RICs perform approximately orthogonal/unitary transformations when applied to sparse vectors. The following theoretical guarantee is provided in \cite{candes2008restricted}.

\begin{thm}[\cite{candes2008restricted}] Assume that $\zeron{\m{x}}\leq K$ for the true signal $\m{x}$  and $\delta_{2K}<\sqrt{2}-1$. Then $\m{x}$ is the unique solution of the $\ell_0$ optimization and the BP problem.\label{Thm:BP_RIP}
\end{thm}

After the work \cite{candes2008restricted}, the RIP condition has been improved, e.g., to $\delta_{2K}<\frac{3}{4+\sqrt{6}}$ \cite{foucart2012sparse}. Other types of RIP conditions are also available, e.g., $\delta_K<0.307$ in \cite{cai2010new}. It is known that stronger results can be provided by using RIP as compared to the mutual coherence. But it is worth noting that, unlike the mutual coherence that can be easily computed given the matrix $\m{A}$, the complexity of computing the RIC of $\m{A}$ may increase dramatically with the sparsity $K$.

In the presence of noise we can solve the following regularized optimization problem, usually known as the least absolute shrinkage and selection operator (LASSO) \cite{tibshirani1996regression}:
\equ{\min_{\m{x}} \lambda\onen{\m{x}}+\frac{1}{2}\twon{\m{A}\m{x}-\m{y}}^2, \label{formu:Lasso}}
where $\lambda>0$ is a regularization parameter, to be specified, or the basis pursuit denoising (BPDN) problem:
\equ{\min_{\m{x}} \onen{\m{x}}, \st \twon{\m{A}\m{x}-\m{y}}\leq\eta, \label{formu:BPDN}}
where $\eta\geq \twon{\m{e}}$ is an upper bound on the noise energy. Note that \eqref{formu:Lasso} and \eqref{formu:BPDN} are equivalent for appropriate choices of $\lambda$ and $\eta$, and that both degenerate to BP in the noiseless case by letting $\eta,\lambda\rightarrow0$. Under RIP conditions similar to the above ones, it has been shown that the sparse signal $\m{x}$ can be stably reconstructed with the reconstruction error being proportional to the noise level \cite{candes2008restricted}.

Besides \eqref{formu:Lasso} and \eqref{formu:BPDN}, another $\ell_1$ optimization method for sparse recovery is the so-called square-root LASSO \cite{belloni2011square}:
\equ{\min_{\m{x}} \tau\onen{\m{x}} + \twon{\m{y}-\m{A}\m{x}}, \label{eq:SR_lasso}}
where $\tau>0$ is a regularization parameter. Compared to the LASSO, for which the noise is usually assumed to be Gaussian and the regularization parameter $\lambda$ is chosen proportional to the standard deviation of the noise, SR-LASSO requires a weaker assumption on the noise distribution and $\tau$ can be chosen as a constant that is independent of the noise level \cite{belloni2011square}.

The $\ell_1$ optimization problems in \eqref{formu:l1}, \eqref{formu:Lasso}, \eqref{formu:BPDN} and \eqref{eq:SR_lasso} are convex and are guaranteed to be solvable in a polynomial time; however, it is not easy to efficiently solve them in the case when the problem dimension is high since the $\ell_1$ norm is not a smooth function. Significant progress has been made over the past decade to accelerate the computation. Examples include $\ell_1$-magic \cite{candes2005l1}, interior-point method \cite{kim2008interior}, conjugate gradient method \cite{lustig2007sparse}, fixed-point continuation \cite{hale2007fixed}, Nesterov's smoothing technique with continuation (NESTA) \cite{nesterov2005smooth,becker2009nesta}, ONE-L1 algorithms \cite{yang2011orthonormal}, alternating direction method of multipliers (ADMM) \cite{boyd2011distributed,yang2011alternating} and so on.

\subsubsection{$\ell_q$ Optimization} \label{sec:lq}
For a vector $\m{x}$, the $\ell_q$, $0<q<1$ (pseudo-)norm is defined as:
\equ{\norm{\m{x}}_q = \sbra{\sum_{n} \abs{x_n}^q}^{\frac{1}{q}}, \label{eq:lq}}
which is a nonconvex relaxation of the $\ell_0$ norm. Compared to \eqref{eq:ell0min} and \eqref{formu:l1}, in the noiseless case, the $\ell_q$ optimization problem is given by:
\equ{\min_{\m{x}} \norm{\m{x}}_q^q, \st \m{y} =\m{A}\m{x}, \label{eq:lqmin}}
where $\norm{\m{x}}_q^q$, instead of $\norm{\m{x}}_q$, is used for the convenience of algorithm development. Since the $\ell_q$ norm is a closer approximation to the $\ell_0$ norm, compared to the $\ell_1$ norm, it is expected that the $\ell_q$ optimization in \eqref{eq:lqmin} results in better performance than the BP. This is true, according to \cite{chartrand2007exact,foucart2009sparsest}. Indeed, $\ell_q$, $0<q<1$ optimization can exactly determine the true sparse signal under weaker RIP conditions than that for the BP. Note that the results above are applicable to the globally optimal solution to \eqref{eq:lqmin}, whereas we can only guarantee convergence to a locally optimal solution in practice.

A well-known algorithm for $\ell_q$ optimization is the {\em foc}al {\em u}nderdetermined
{\em s}ystem {\em s}olver (FOCUSS) \cite{gorodnitsky1997sparse,rao1999affine}. FOCUSS is an iterative reweighted least squares method. In each iteration, FOCUSS solves the following weighted least squares problem:
\equ{\min_{\m{x}} \sum_{n} w_n \abs{x_n}^2, \st \m{A}\m{x} = \m{y}, \label{eq:FOCUSS_iter}}
where the weight coefficients $w_n = \abs{x_n}^{q-2}$ are updated using the latest solution $\m{x}$. Note that \eqref{eq:FOCUSS_iter} can be solved in closed form and hence an iterative algorithm can be implemented with a proper initialization. This algorithm can be interpreted as a majorization-minimization (MM) algorithm that is guaranteed to converge to a local minimum.

In the presence of noise, the following regularized problem is considered in lieu of \eqref{formu:Lasso}:
\equ{\min_{\m{x}} \lambda\norm{\m{x}}_q^q+\frac{1}{2}\twon{\m{A}\m{x}-\m{y}}^2, \label{eq:lqminnoisy}}
where $\lambda>0$ is a regularization parameter. A regularized FOCUSS algorithm for \eqref{eq:lqminnoisy} was developed in \cite{rao2003subset} by using the same main idea as in FOCUSS. A difficult problem regarding \eqref{eq:lqminnoisy} is the choice of the parameter $\lambda$. Although several heuristic methods for tuning this parameter were introduced in \cite{rao2003subset}, to the best of our knowledge there have been no theoretical results on this aspect.

To bypass the parameter tuning problem, a maximum {\em a posterior} (MAP) estimation approach called SLIM (sparse learning via iterative minimization) was proposed in \cite{tan2011sparse}. Assuming i.i.d. Gaussian noise with variance $\eta$ and the following prior distribution for $\m{x}$:
\equ{f(\m{x})\propto \prod_n e^{-\frac{2}{q}\sbra{\abs{x_n}^q-1}},}
SLIM computes the MAP estimate by solving the following $\ell_q$ optimization problem:
\equ{\min_{\m{x}} M\log \eta+\eta^{-1}\twon{\m{A}\m{x}-\m{y}}^2 + \frac{2}{q}\norm{\m{x}}_q^q. \label{eq:slim}}
To locally solve \eqref{eq:slim}, SLIM iteratively updates $\m{x}$, as the regularized FOCUSS does. However, unlike FOCUSS, SLIM also iteratively updates the parameter $\eta$ based on the latest solution $\m{x}$. Once $q$ is given, SLIM is hyper-parameter free. Note that \eqref{eq:slim} reduces to \eqref{eq:lqminnoisy} for fixed $\eta$.

\subsubsection{Maximum Likelihood Estimation (MLE)}
MLE is another common approach to sparse estimation. In contrast to the convex relaxation and OMP, one advantage of MLE is that it does not require knowledge of the noise level or the sparsity level (the latter being often needed to choose $\lambda$ in \eqref{formu:Lasso} properly). To derive it, assume that $\m{x}$ follows a multivariate Gaussian distribution with mean zero and covariance $\m{P}=\diag\sbra{\m{p}}$, where $p_n\geq0$, $p=1,\dots,\overline{N}$ (this can be viewed as a prior distribution that does not necessarily have to hold in practice). Also, assume i.i.d.~Gaussian noise with variance $\sigma$. It follows from the data model in \eqref{eq:sparserep} that $\m{y}$ follows a Gaussian distribution with mean zero and covariance $\m{R} = \m{A}\m{P}\m{A}^H +\sigma\m{I}$. Consequently, the negative log-likelihood function associated with $\m{y}$ is given by
\equ{\cL\sbra{\m{p},\sigma} = \log\abs{\m{R}} + \m{y}^H\m{R}^{-1}\m{y}.}
The parameters $\m{p}$ and $\sigma$ can be estimated by minimizing $\cL$:
\equ{\min_{\m{p},\sigma}\log\abs{\m{R}} + \m{y}^H\m{R}^{-1}\m{y}. \label{eq:negloglike}}
Once $\m{p}$ and $\sigma$ are solved for, the posterior distribution of the sparse signal $\m{x}$ can be obtained: it is a Gaussian distribution with mean and covariance given, respectively, by
{\lentwo\equa{\m{\mu}
&=& \m{\Sigma}\m{A}^H\m{y}, \\ \m{\Sigma}
&=& \sbra{\m{A}^H\m{A} + \sigma\m{P}^{-1}}^{-1}.
}}The vector $\m{x}$ can be estimated as its posterior mean $\m{\mu}$. In this process, the sparsity of $\m{x}$ is achieved by the fact that most of the entries of $\m{p}$ approach zero in practice. Theoretically, it can be shown that in the limiting noiseless case, the global optimizer to \eqref{eq:negloglike} coincides with that of $\ell_0$ optimization \cite{wipf2004sparse}.

The main difficulty of MLE is solving \eqref{eq:negloglike} in which the first term of the objective function, viz.~$\log\abs{\m{R}}$, is a nonconvex (in fact, concave) function of $\sbra{\m{p},\sigma}$. Different approaches have been proposed, e.g., reweighted optimization \cite{stoica2012spice} and sparse Bayesian learning (SBL) \cite{tipping2001sparse,wipf2004sparse,ji2008bayesian}. In \cite{stoica2012spice}, a majorization-minimization approach is adopted to linearize $\log\abs{\m{R}}$ at each iteration by its tangent plane $\tr\sbra{\m{R}_j^{-1}\m{R}}+const$ given the latest estimate $\m{R}_j$. The resulting problem at each iteration is convex and solved using an algorithm called {\em sp}arse {\em i}terative {\em c}ovariance-based {\em e}stimation (SPICE) \cite{stoica2011new,stoica2011spice,stoica2012spice, stoica2014weighted} that will be introduced in Subsection \ref{sec:SPICE}.

The MLE has been interpreted from a different perspective within the framework of SBL or Bayesian compressed sensing. In particular, to achieve sparsity, a prior distribution is assumed for $\m{x}$ that promotes sparsity and is usually referred to as a sparse prior. In \cite{tipping2001sparse}, for example, a Student's $t$-distribution is assumed for $\m{x}$ that is constructed in a hierarchical manner: specifically, a Gaussian distribution as above at the first stage followed by a Gamma distribution for the inverse of the powers, $p_n^{-1}$, $n=1,\dots,{\overline{N}}$ at the second stage. Interestingly, despite different approaches, the same objective function is obtained as in \eqref{eq:negloglike}. To optimize \eqref{eq:negloglike}, an expectation-maximization (EM) algorithm is adopted \cite{mclachlan1997algorithm}. In the E-step, the posterior distribution of $\m{x}$ is computed, as mentioned previously, while in the M-step, $\m{p}$ and $\sigma$ are updated as functions of the latest statistics of $\m{x}$, viz.~$\m{\mu}$ and $\m{\Sigma}$. The process is repeated and it guarantees a monotonic decrease of $\cL$. Finally, we note that with other sparse priors for $\m{x}$ that may possess different sparsity promoting properties, the obtained objective function of SBL can be slightly different from that of the MLE in \eqref{eq:negloglike} (see, e.g., \cite{babacan2010bayesian}).

\subsection{Sparse Representation and DOA Estimation: the Link and the Gap}
In this subsection we discuss the link and the gap between sparse representation and DOA estimation. By doing so, we can see the possibility and the main challenges of using the sparse representation techniques for DOA estimation. It has been mentioned that the underlying motivation of sparse representation is that the observed data $\m{y}$ can be well approximated in a lower-dimensional subspace. In fact, this is exactly the case in DOA estimation where the data snapshot $\m{y}(t)$ is a linear combination of the steering vectors of the sources and the sparsity arises from the fact that there are less sources than sensors (note that for some special arrays and methods more sources than the sensors can be detected). By comparing the models in \eqref{formu:observation_model1} and \eqref{eq:sparserep}, it can be seen that the process of DOA estimation boils down to a sparse representation of the data snapshot with each DOA $\theta$ corresponding to one atom given by $\m{a}\sbra{\theta}$. Therefore, it is possible to use sparse representation techniques in DOA estimation.

However, there exist major differences between the common sparse representation framework and DOA estimation. First, and most importantly, the dictionary in sparse representation usually contains a finite number of atoms while in the DOA estimation problem the parameters are continuously valued, which leads to infinitely many atoms. More concretely, the atoms in sparse representation are given by the columns of a matrix. But in DOA estimation each atom $\m{a}\sbra{\theta}$ is parameterized by a continuous parameter $\theta$.

Second, there are usually multiple snapshots in DOA estimation problems, in contrast to the single snapshot case in sparse representation. It is then crucial to exploit the temporal redundancy of the snapshots in DOA estimation since the number of antennas can be limited due to physical and other constraints. Typically, the number of antennas $M$ is about $10\sim100$, while the number of snapshots $L$ can be much larger.

Last but not least, the existing theoretical guarantees of the sparse representation techniques are usually derived using what is known as incoherence analysis, e.g., those based on the mutual coherence and RIP, in the sense that they are applicable only in the case of incoherent dictionaries. This means that such guarantees can hardly be applied to DOA estimation problems, in which the atoms are completely coherent. But this does not necessarily mean that satisfactory performance cannot be achieved in DOA estimation problems. Indeed, note that the success of sparse signal recovery is measured by the size of the reconstruction error of the sparse signal $\m{x}$, and that a slight error in the support usually results in a large estimation error. But this is not true for DOA estimation where the estimation error is actually measured by the error of the support (and, therefore, a small estimation error of the support is acceptable).

The next three sections describe three different possibilities for dealing with the first gap---discrete versus continuous atoms---when applying sparse representation to DOA estimation. In each section, we will also discuss how the signal sparsity and the temporal redundancy of the multiple snapshots are exploited and what theoretical guarantees can be obtained.


\section{On-Grid Sparse Methods} \label{sec:ongrid}
In this section we introduce the first class of sparse methods for DOA estimation, termed as on-grid sparse methods. These methods are developed by directly applying sparse representation and compressed sensing techniques to DOA estimation. To do so, the DOAs are assumed to lie on a prescribed grid so that the problem can be solved within the common framework of sparse representation. The main challenge then is how to exploit the temporal redundancy of multiple snapshots.

In the following we first introduce the data model that we will use throughout this section. Then, we present several formulations and algorithms for DOA estimation within the on-grid framework, including $\ell_{2,q}$ optimization methods with $0\leq q\leq 1$, SBL and SPICE. Guidelines for grid selection will also be provided.

\subsection{Data Model}
To fill the gap between continuous DOA estimation and discrete sparse representation, it is simply assumed by the on-grid sparse methods that the continuous DOA domain $\cD_{\theta}$ can be replaced by a {\em given} set of grid points
\equ{\overline{\m{\theta}}=\lbra{\overline{\theta}_1,\dots,\overline{\theta}_{\overline{N}}},}
where $\overline{N}\gg M$ is the grid size. This means that the candidate DOAs can only take values in $\overline{\m{\theta}}$, which results in the following $M\times \overline{N}$ dictionary matrix
\equ{\m{A}= \m{A}\sbra{\overline{\m{\theta}}}= \mbra{\m{a}\sbra{\overline{\theta}_1}, \dots, \m{a}\sbra{\overline{\theta}_{\overline{N}}}}.}
It follows that the data model in \eqref{formu:observation_model} for DOA estimation can be equivalently written as:
\equ{\m{Y} = \m{A} \m{X} + \m{E}, \label{eq:datamodel_D}}
where $\m{X}=\mbra{\m{x}(1),\dots,\m{x}(L)}$ is an $\overline{N}\times L$ matrix in which each column $\m{x}(t)$ is an augmented version of the source signal $\m{s}(t)$ and is defined by:
\equ{x_n(t) = \left\{\begin{array}{ll} s_k(t), & \text{ if } \overline{\theta}_{n} =\theta_k; \\ 0, & \text{ otherwise, } \end{array}\right. \quad n=1,\dots,\overline{N},\; t=1,\dots,L.}
It can be seen that for each $t$, $\m{x}(t)$ contains only $K$ nonzero entries, whose locations correspond to the $K$ DOAs, and therefore it is a sparse vector as $K\ll \overline{N}$. Moreover, $\m{x}(t)$, $t=1,\dots,L$ are jointly sparse in the sense that they share the same support. Alternatively, we can say that $\m{X}$ is row-sparse in the sense that it contains only a few nonzero rows.

By means of the data model in \eqref{eq:datamodel_D}, the DOA estimation problem is transformed into a sparse signal recovery problem. The DOAs are encoded in the support of the sparse vectors $\m{x}(t)$, $t=1,\dots,L$ and therefore, we only need to recover this support from which the estimated DOAs can be retrieved.

The key and only difference between \eqref{eq:datamodel_D} and \eqref{eq:sparserep} is that the former contains multiple data snapshots that are also referred to as multiple measurement vectors (MMVs). In the case of a single snapshot with $L=1$ (i.e., single measurement vector (SMV)), the sparse representation techniques can be readily applied to DOA estimation. In the case of multiple snapshots, the main difficulty consists in exploiting the temporal redundancy of the snapshots---the joint sparsity of the columns of $\m{X}$---for possibly improved performance. Since the MMV data model in \eqref{eq:datamodel_D} is quite general, extensive studies have been performed for the joint sparse signal recovery problem (see, e.g., \cite{cotter2005sparse,malioutov2005sparse,chen2006theoretical, fornasier2008recovery,mishali2008reduce, gribonval2008atoms,kowalski2009sparse, ji2009multi,eldar2009robust, eldar2010average, hyder2010direction, van2010theoretical, kim2012compressive,lee2012subspace,davies2012rank, stoica2011spice}). We only discuss some of them in the ensuing subsections.

Before proceeding to the on-grid sparse methods, we make some comments on the data model in \eqref{eq:datamodel_D}. Note that the set of grid points $\overline{\m{\theta}}$ needs to be fixed {\em a priori} so that the dictionary $\m{A}$ is known, which is required in the sparse signal recovery process. Consequently, there is no guarantee that the true DOAs lie on the grid $\overline{\m{\theta}}$; in fact, this fails with probability one, resulting in the grid mismatch problem \cite{chi2011sensitivity,chae2010effects}. To ensure at least that the true DOAs are {\em close} to the grid points, in practice the grid needs to be dense enough (with $\overline{N}\gg M$). Therefore, \eqref{eq:datamodel_D} can be viewed as a zeroth-order approximation of the true data model in \eqref{formu:observation_model} and the noise term $\m{E}$ in \eqref{eq:datamodel_D} may also comprise the approximation error besides the true noise in \eqref{formu:observation_model}.

\subsection{$\ell_{2,0}$ Optimization}
We first discuss how the joint sparsity can be exploited for sparse recovery. We start with the definition of sparsity for the row-sparse matrix $\m{X}$. Since each row of $\m{X}$ corresponds to one potential source, it is natural to define the sparsity as the number of nonzero rows of $\m{X}$, which is usually expressed as the following $\ell_{2,0}$ norm (see, e.g., \cite{chen2006theoretical,hyder2010direction}):
\equ{\norm{\m{X}}_{2,0}= \#\lbra{n:\; \twon{\m{X}_n}>0} = \#\lbra{n:\; \m{X}_n\neq\m{0}}, \label{eq:l20norm}}
where $\m{X}_n$ denotes the $n$th row of $\m{X}$. Note that in \eqref{eq:l20norm} the $\ell_2$ norm can in fact be replaced by any other norm. Following from the $\ell_0$ optimization in the single snapshot case, the following $\ell_{2,0}$ optimization can be proposed in the absence of noise:
\equ{\min_{\m{X}} \norm{\m{X}}_{2,0}, \st \m{Y} = \m{A}\m{X}. \label{eq:l20min}}
Suppose the optimal solution, denoted by $\widehat{\m{X}}$, can be obtained. Then, the DOAs can be retrieved from the row-support of $\widehat{\m{X}}$.

To realize the potential advantage of using the joint sparsity of the snapshots, consider the following result.

\begin{thm}[\cite{chen2006theoretical}] The true matrix $\m{X}$ is the unique solution to \eqref{eq:l20min} if
\equ{\norm{\m{X}}_{2,0}< \frac{\text{spark}\sbra{\m{A}}-1+\rank\sbra{\m{Y}}}{2}. \label{eq:Kcond_l20}} \label{thm:l20min}
\end{thm}

Note that the condition in \eqref{eq:Kcond_l20} is very similar to that in \eqref{eq:cond2} required to guarantee parameter identifiability for DOA estimation. By Theorem \ref{thm:l20min}, the number of recoverable DOAs can be increased in general by collecting more snapshots since then $\rank\sbra{\m{Y}}$ increases. The only exception happens in the case when the data snapshots $\m{y}(t)$, $t=1,\dots,L$ are identical up to scaling factors. Unfortunately, similar to the single snapshot case, the above $\ell_{2,0}$ optimization problem is NP-hard to solve.

\subsection{Convex Relaxation}
\subsubsection{$\ell_{2,1}$ Optimization} \label{sec:l0_MMV}
The tightest convex relaxation of the $\ell_{2,0}$ norm is given by the $\ell_{2,1}$ norm that is defined as:
\equ{\twonen{\m{X}}= \sum_{n} \twon{\m{X}_n}.}
Though in the definition in \eqref{eq:l20norm} the $\ell_2$ norm in the $\ell_{2,0}$ norm can be replaced by other norms, its use is important in the $\ell_{2,1}$ norm.
Based on \eqref{eq:l20min}, the following $\ell_{2,1}$ optimization problem is proposed in the absence of noise \cite{malioutov2005sparse,cotter2005sparse}:
\equ{\min_{\m{X}} \norm{\m{X}}_{2,1}, \st \m{Y} = \m{A}\m{X}. \label{eq:l21min}}
As reported in the literature, the performance of $\ell_{2,1}$ optimization approach can be generally improved by increasing the number of measurement vectors. Theoretically, this can be shown to be true under the assumption that the jointly sparse signals are randomly drawn such that the rows of the source signals $\m{S}_k$, $k=1,\dots,K$ are at general positions \cite{eldar2010average}. It is worth noting that the theoretical guarantee cannot be improved without assumptions on the source signals. To see this, consider the case when the columns of $\m{S}$ are identical up to scaling factors. Then, acquiring more snapshots does not provide useful information for DOA estimation. In this respect, the result of \cite{eldar2010average} can be referred to as \emph{average case} analysis while those accounting for the aforementioned extreme case can be called \emph{worst case} analysis.

In parallel to \eqref{formu:Lasso} and \eqref{formu:BPDN}, in the presence of noise we can solve the LASSO problem:
\equ{\min_{\m{X}} \lambda\twonen{\m{X}}+\frac{1}{2}\frobn{\m{A}\m{X}-\m{Y}}^2 \label{formu:Lasso_MMV}}
where $\lambda>0$ is a regularization parameter (to be specified), or the BPDN problem:
\equ{\min_{\m{X}} \twonen{\m{X}}, \st \frobn{\m{A}\m{X}-\m{Y}}\leq\eta, \label{formu:BPDN_MMV}}
where $\eta\geq \twon{\m{E}}$ is an upper bound on the noise energy. Note that it is generally difficult to choose $\lambda$ in \eqref{formu:Lasso_MMV}. Given the noise variance, results on choosing $\lambda$ have recently been provided in \cite{bhaskar2013atomic, yang2016gridless1, li2016off} in the case of ULA and SLA. Readers are referred to Section \ref{sec:gridless} for details.

Finally, we note that most, if not all, of the computational approaches to $\ell_1$ optimization, e.g., those mentioned in Subsection \ref{sec:l1}, can be easily extended to deal with $\ell_{2,1}$ optimization in the case of multiple snapshots. Once $\m{X}$ is solved for, we can form a power spectrum by computing the power of each row of $\m{X}$ from which the estimated DOAs can be obtained.

\subsubsection{Dimensionality Reduction via $\ell_{2,1}$-SVD}
In DOA estimation applications the number of snapshots $L$ can be large, which can significantly increase the computational workload of $\ell_{2,1}$ optimization. In the case of $L>K$ a dimensionality reduction technique was proposed in \cite{malioutov2005sparse} inspired by the conventional subspace methods, e.g., MUSIC. In particular, suppose there is no noise; then the data snapshots $\m{Y}$ lie in a $K$-dimensional subspace. In the presence of noise, therefore, we can decompose $\m{Y}$ into the signal and noise subspaces, keep the signal subspace and use it in \eqref{eq:l21min}-\eqref{formu:BPDN_MMV} in lieu of $\m{Y}$. Mathematically, we compute the singular value decomposition (SVD)
\equ{\m{Y} = \m{U}\m{L}\m{V}^H. \label{eq:YSVD}}
Define a reduced $M\times K$ dimensional data matrix
\equ{\m{Y}_{\text{SV}} = \m{U}\m{L}\m{D}_K^T = \m{Y}\m{V}\m{D}_K^T}
that contains most of the signal power, where $\m{D}_K= \mbra{\m{I}_K, \m{0}}$ with $\m{I}_K$ being an identity matrix of order $K$. Also let $\m{X}_{\text{SV}} = \m{X}\m{V}\m{D}_K^T$ and $\m{E}_{\text{SV}} = \m{E}\m{V}\m{D}_K^T$. Using these notations we can write a new data model:
\equ{\m{Y}_{\text{SV}} = \m{A}\m{X}_{\text{SV}} + \m{E}_{\text{SV}}. \label{eq:datamodel_SVD}}
Note that \eqref{eq:datamodel_SVD} is in exactly the same form as \eqref{eq:datamodel_D} but with reduced dimensionality. So similar $\ell_{2,1}$ optimization problems can be formulated as \eqref{formu:Lasso_MMV} and \eqref{formu:BPDN_MMV}, which are referred to as $\ell_{2,1}$-SVD.

The following comments on $\ell_{2,1}$-SVD are in order. Note that the true source number $K$ has been known to obtain $\m{Y}_{\text{SV}}$. However, $\ell_{2,1}$-SVD is not very sensitive to this choice and therefore an appropriate estimate of $K$ is sufficient in practice \cite{malioutov2005sparse}. Nevertheless, parameter tuning remains a difficult problem for $\ell_{2,1}$-SVD ($\lambda$ and $\eta$ in \eqref{formu:Lasso_MMV} and \eqref{formu:BPDN_MMV}). Though some solutions have been proposed for the standard $\ell_{2,1}$ optimization methods in \eqref{formu:Lasso_MMV} and \eqref{formu:BPDN_MMV}, given the noise level, they cannot be applied to $\ell_{2,1}$-SVD due to the change in the data structure. Regarding this aspect, it is somehow hard to compare the DOA estimation performances of the standard $\ell_{2,1}$ optimization and $\ell_{2,1}$-SVD; though it is argued in \cite{malioutov2005sparse} that, as compared to the standard form, $\ell_{2,1}$-SVD can improve the robustness to noise by keeping only the signal subspace.

\subsubsection{Another Dimensionality Reduction Technique} \label{sec:dimred_D}
We present here another dimensionality reduction technique that reduces the number of snapshots from $L$ to $M$ and has the same performance as the original $\ell_{2,1}$ optimization. The technique was proposed in \cite{yang2016enhancing}, inspired by a similar technique used for the gridless sparse methods (see Section \ref{sec:gridless}). For convenience, we introduce it following the idea of $\ell_{2,1}$-SVD. In $\ell_{2,1}$-SVD the number of snapshots is reduced from $L$ to $K$ by keeping only the $K$-dimensional signal subspace; in contrast the technique in \cite{yang2016enhancing} suggests keeping both the signal and noise subspaces. To be specific, suppose that $L>M$ and $\m{Y}$ has rank $r\leq M$ (note that typically $r=M$ in the presence of noise). Then, given the SVD in \eqref{eq:YSVD}, we retain a reduced $M\times r$ dimensional data matrix
\equ{\m{Y}_{\text{DR}} = \m{U}\m{L}\m{D}_r^T = \m{Y}\m{V}\m{D}_r^T}
that preserves all of the data power since $\m{Y}$ has only $r$ nonzero singular values, where $\m{D}_r$ is defined similarly to $\m{D}_K$. We similarly define $\m{X}_{\text{DR}} = \m{X}\m{V}\m{D}_r^T$ and $\m{E}_{\text{DR}} = \m{E}\m{V}\m{D}_r^T$ to obtain the data model
\equ{\m{Y}_{\text{DR}} = \m{A}\m{X}_{\text{DR}} + \m{E}_{\text{DR}}. \label{eq:datamodel_SVD2}}
For the LASSO problem, as an example, it can be shown that {\em equivalent} solutions can be obtained before and after the dimensionality reduction. To be specific, if $\widehat{\m{X}}_{\text{DR}}$ is the solution to the following LASSO problem:
\equ{\min_{\m{X}_{\text{DR}}} \lambda\twonen{\m{X}_{\text{DR}}}+\frac{1}{2}\frobn{\m{A}\m{X}_{\text{DR}}-\m{Y}_{\text{DR}}}^2, \label{formu:Lasso_MMV2}}
then $\widehat{\m{X}} =\widehat{\m{X}}_{\text{DR}}\m{D}_r\m{V}^H$ is the solution to \eqref{formu:Lasso_MMV}. $\widehat{\m{X}}_{\text{DR}}$ and $\widehat{\m{X}}$ are equivalent in the sense that the corresponding rows of $\widehat{\m{X}}_{\text{DR}}$ and $\widehat{\m{X}}$ have the same power, resulting in identical power spectra (to see this, note that $\widehat{\m{X}}\widehat{\m{X}}^H=\widehat{\m{X}}_{\text{DR}}\m{D}_r\m{V}^H \m{V}\m{D}_r^T\widehat{\m{X}}_{\text{DR}}^H = \widehat{\m{X}}_{\text{DR}}\widehat{\m{X}}_{\text{DR}}^H$, whose diagonal contains the powers of the rows).

We next prove the above result. To do so, for any $\m{X}$, split $\m{X}\m{V}$ into two parts: $\m{X}\m{V} = \mbra{\m{X}_{\text{DR}}, \m{X}_2}$. Note that $\m{Y}\m{V} = \mbra{\m{Y}_{\text{DR}}, \m{0}}$. By the fact that $\m{V}$ is a unitary matrix, it can be easily shown that
{\lentwo\equa{\twonen{\m{X}}
&=& \twonen{\m{X}\m{V}} = \twonen{\mbra{\m{X}_{\text{DR}}, \m{X}_2}},\\ \frobn{\m{A}\m{X}-\m{Y}}^2
&=& \frobn{\m{A}\m{X}\m{V}-\m{Y}\m{V}}^2 = \frobn{\m{A}\m{X}_{\text{DR}}-\m{Y}_{\text{DR}}}^2 + \frobn{\m{A}\m{X}_2}^2.
}}It immediately follows that
\equ{\lambda\twonen{\m{X}}+\frac{1}{2}\frobn{\m{A}\m{X}-\m{Y}}^2 \geq \lambda\twonen{\m{X}_{\text{DR}}}+\frac{1}{2}\frobn{\m{A}\m{X}_{\text{DR}}-\m{Y}_{\text{DR}}}^2 \label{eq:DRinequ}}
and the equality holds if and only if $\m{X}_2 = \m{0}$, or equivalently, $\m{X} =\m{X}_{\text{DR}}\m{D}_r\m{V}^H$. We can obtain the stated result by minimizing both sides of \eqref{eq:DRinequ} with respect to $\m{X}$.

Note that the above result also holds if $\m{Y}_{\text{DR}}$ is replaced by any full-column-rank matrix $\widetilde{\m{Y}}$ satisfying $\widetilde{\m{Y}}\widetilde{\m{Y}}^H = \m{Y}\m{Y}^H$, since there always exists a unitary matrix $\m{V}$ such that $\widetilde{\m{Y}} = \m{Y}\m{V}\m{D}_r^T$ as for $\m{Y}_{\text{DR}}$. Therefore, the SVD of the $M\times L$ dimensional data matrix $\m{Y}$, which can be computationally expensive in the case of $L\gg M$, can be replaced by the Cholesky decomposition or the eigenvalue decomposition of the $M\times M$ matrix $\m{Y}\m{Y}^H$ (which is the sample data covariance matrix up to a scaling factor). Another fact that makes this dimensional reduction technique superior to $\ell_{2,1}$-SVD is that the parameter $\lambda$ or $\eta$ can be tuned as in the original $\ell_{2,1}$ optimization, for which solutions are available if the noise level is given.

\subsection{$\ell_{2,q}$ Optimization}
Corresponding to the $\ell_q$, $0<q<1$ norm considered in Subsection \ref{sec:lq}, we can define the $\ell_{2,q}$ norm to exploit the joint sparsity in $\m{X}$ as:
\equ{\norm{\m{X}}_{2,q} = \sbra{\sum_n \twon{\m{X}_n}^q}^{\frac{1}{q}},}
which is a nonconvex relaxation of the $\ell_{2,0}$ norm.
In lieu of \eqref{eq:l21min} and \eqref{formu:Lasso_MMV}, in the noiseless case, we can solve the following equality constrained problem:
\equ{\min_{\m{X}}\norm{\m{X}}_{2,q}^q, \st \m{A}\m{X} = \m{Y}, \label{eq:l2qmin} }
or the following regularized form in the noisy case:
\equ{\min_{\m{X}} \lambda\norm{\m{X}}_{2,q}^q+\frac{1}{2}\frobn{\m{A}\m{X}-\m{Y}}^2. \label{eq:l2qregmin}}
To locally solve \eqref{eq:l2qmin} and \eqref{eq:l2qregmin}, the FOCUSS algorithm was extended in \cite{cotter2005sparse} to this multiple snapshot case to obtain M-FOCUSS. For \eqref{eq:l2qmin}, as in the single snapshot case, M-FOCUSS solves the following weighted least squares problem in each iteration:
\equ{\min_{\m{X}} \sum_n w_n \norm{\m{X}}_{2}^2, \st \m{A}\m{X} = \m{Y}, \label{eq:MFOCUSS} }
where the weight coefficients $w_n = \norm{\m{X}}_{2}^{q-2}$ are updated based on the latest solution $\m{X}$. Since \eqref{eq:MFOCUSS} can be solved in closed form, an iterative algorithm can be implemented. Note that \eqref{eq:l2qregmin} can be similarly solved as \eqref{eq:l2qmin}.

To circumvent the need for tuning the regularization parameter $\lambda$ in \eqref{eq:l2qregmin}, we can extend SLIM \cite{tan2011sparse} to this multiple snapshot case as follows. Assume that the noise is i.i.d.~Gaussian with variance $\eta$ and that $\m{X}$ follows a prior distribution with the pdf given by
\equ{f\sbra{\m{X}} \propto \prod_{n} e^{-\frac{2}{q}\sbra{\twon{\m{X}_n}^q-1}}. \label{eq:fX}}
In \eqref{eq:fX}, the $\ell_2$ norm is performed on each row of $\m{X}$ to exploit the joint sparsity. As in the single snapshot case, SLIM computes the MAP estimator of $\m{X}$ which is the solution of the following $\ell_{2,q}$ optimization problem:
\equ{\min_{\m{x}} ML\log \eta+\eta^{-1}\frobn{\m{A}\m{X}-\m{Y}}^2 + \frac{2}{q}\norm{\m{X}}_{2,q}^q. \label{eq:mslim}}
Using a reweighting technique similar to that in M-FOCUSS, we can iteratively update $\m{X}$ and $\eta$ in closed form and obtain the multiple snapshot version of SLIM. Finally, note that the dimensionality reduction technique presented in  Subsection \ref{sec:dimred_D} can also be applied to the $\ell_{2,q}$ optimization problems in \eqref{eq:l2qmin}, \eqref{eq:l2qregmin} and \eqref{eq:mslim} for algorithm acceleration \cite{yang2016enhancing}.

\subsection{Sparse Iterative Covariance-based Estimation (SPICE)} \label{sec:SPICE}

\subsubsection{Generalized Least Squares}
To introduce SPICE, we first present the so-called generalized least squares method. To derive it, we need some statistical assumptions on the sources $\m{X}$ and the noise $\m{E}$. We assume that $\lbra{\m{x}(1),\dots,\m{x}(L), \m{e}(1),\dots,\m{e}(L)}$ are uncorrelated with one another and
{\lentwo\equa{\bE \m{e}(t)\m{e}^H(t)
&=& \sigma\m{I},\\ \bE\m{x}(t)\m{x}^H(t)
&=& \m{P}= \diag\sbra{\m{p}}, \quad t=1,\dots,L,
}}where $\sigma\geq0$ and $p_n\geq 0$, $n=1,\dots,\overline{N}$ are the parameters of interest (note that the following derivations also apply to the case of heteroscedastic noise with $\bE \m{e}(t)\m{e}^H(t)= \diag\sbra{\sigma_1,\dots,\sigma_M}$ with no or minor modifications). It follows that the snapshots $\lbra{\m{y}(1),\dots,\m{y}(L)}$ are uncorrelated with one another and have the following covariance matrix:
\equ{\m{R}\sbra{\m{p},\sigma} = \bE\m{y}(t)\m{y}^H(t) = \m{A}\m{P}\m{A}^H + \sigma\m{I}= \m{A}'\m{P}'\m{A}'^H, \label{eq:RA1PAH}}
where $\m{A}'= \mbra{\m{A}, \m{I}}$ and $\m{P}'= \diag\sbra{\m{P}, \sigma\m{I}}$.
Note that $\m{R}$ is linear in $\sbra{\m{p},\sigma}$.
Let $\widetilde{\m{R}} = \frac{1}{L}\m{Y}\m{Y}^H$ be the sample covariance matrix. Given $\widetilde{\m{R}}$, to estimate $\m{R}$ (in fact, the parameters $\m{p}$ and $\sigma$ therein), we consider the generalized least squares method. First, we vectorize $\widetilde{\m{R}}$ and let $\widetilde{\m{r}} = \vect\sbra{\widetilde{\m{R}}}$ and $\m{r} = \vect\sbra{\m{R}}$. Since $\widetilde{\m{R}}$ is an unbiased estimate of the data covariance matrix, it holds that
\equ{\bE\widetilde{\m{r}} = \m{r}.}
Moreover, we can calculate the covariance matrix of $\widetilde{\m{r}}$, which is given by (see, e.g., \cite{ottersten1998covariance})
\equ{\cov\sbra{\widetilde{\m{r}}} = \frac{1}{L}\m{R}^T\otimes \m{R}, \label{eq:covofcov}}
where $\otimes$ denotes the Kronecker product.
In the generalized least squares method we minimize the following criterion \cite{anderson1984multivariate,ottersten1998covariance}:
\equ{\begin{split}
&\frac{1}{L}\sbra{\widetilde{\m{r}} - \bE\widetilde{\m{r}}}^H\cov^{-1}\sbra{\widetilde{\m{r}}} \sbra{\widetilde{\m{r}} - \bE\widetilde{\m{r}}} \\
&= \sbra{\widetilde{\m{r}} - \m{r}}^H \mbra{\m{R}^{-T}\otimes \m{R}^{-1}} \sbra{\widetilde{\m{r}} - \m{r}}\\
&= \vect^H\sbra{\widetilde{\m{R}} - \m{R}} \mbra{\m{R}^{-T}\otimes \m{R}^{-1}} \vect\sbra{\widetilde{\m{R}} - \m{R}}\\
&= \vect^H\sbra{\widetilde{\m{R}} - \m{R}} \vect\lbra{\m{R}^{-1} \sbra{\widetilde{\m{R}} - \m{R}} \m{R}^{-1} } \\
&= \tr\lbra{\sbra{\widetilde{\m{R}} - \m{R}} \m{R}^{-1} \sbra{\widetilde{\m{R}} - \m{R}} \m{R}^{-1}}\\
&= \frobn{\m{R}^{-\frac{1}{2}} \sbra{\widetilde{\m{R}} - \m{R}} \m{R}^{-\frac{1}{2}}}^2. \end{split} \label{eq:generalLS}}
The criterion in \eqref{eq:generalLS} has good statistical properties; for example, under certain conditions it provides a large-snapshot maximum likelihood (ML) estimator of the parameters $\sbra{\m{p},\sigma}$ of interest. Unfortunately, \eqref{eq:generalLS} is nonconvex in $\m{R}$ and hence nonconvex in $\sbra{\m{p},\sigma}$. Therefore, there is no guarantee that it can be globally minimized.

Inspired by \eqref{eq:generalLS}, the following convex criterion was proposed in \cite{li1999computationally}:
\equ{\frobn{\widetilde{\m{R}}^{-\frac{1}{2}} \sbra{\widetilde{\m{R}} - \m{R}} \widetilde{\m{R}}^{-\frac{1}{2}}}^2, \label{eq:li1999}}
in which $\cov\sbra{\widetilde{\m{r}}}$ in \eqref{eq:covofcov} is replaced by its consistent estimate, viz.~$\frac{1}{L}\widetilde{\m{R}}^T\otimes \widetilde{\m{R}}$. The resulting estimator remains a large-snapshot ML estimator. But it is only usable in the case of $L\geq M$ when $\widetilde{\m{R}}$ is nonsingular. The SPICE approach, which is discussed next, relies on \eqref{eq:generalLS} or \eqref{eq:li1999} (see below for details).

\subsubsection{SPICE}
The SPICE algorithm has been proposed and studied in \cite{stoica2011new,stoica2011spice,stoica2012spice,stoica2014weighted}.
In SPICE, the following covariance fitting criterion is adopted in the case of $L\geq M$ whenever $\widetilde{\m{R}}$ is nonsingular:
\equ{h_1 = \frobn{\m{R}^{-\frac{1}{2}} \sbra{\widetilde{\m{R}} - \m{R}} \widetilde{\m{R}}^{-\frac{1}{2}}}^2. \label{eq:h1}}
In the case of $L<M$, in which $\widetilde{\m{R}}$ is singular, the following criterion is used instead:
\equ{h_2=\frobn{\m{R}^{-\frac{1}{2}} \sbra{\widetilde{\m{R}} - \m{R}}}^2. \label{eq:h2}}

A simple calculation shows that
\equ{\begin{split}h_1
&= \tr\sbra{\m{R}^{-1}\widetilde{\m{R}}}+ \tr\sbra{\widetilde{\m{R}}^{-1}\m{R}}-2M \\
&= \tr\sbra{\widetilde{\m{R}}^{\frac{1}{2}}\m{R}^{-1}\widetilde{\m{R}}^{\frac{1}{2}}} + \sum_{n=1}^{\overline{N}} \sbra{\m{a}^H_n\widetilde{\m{R}}^{-1}\m{a}_n} p_n + \tr\sbra{\widetilde{\m{R}}^{-1}}\sigma - 2M. \end{split} \label{eq:h1_2}}
It follows that the optimization problem of SPICE based on $h_1$ can be equivalently formulated as:
\equ{\min_{\m{p}\succeq\m{0},\sigma>0} \tr\sbra{\widetilde{\m{R}}^{\frac{1}{2}}\m{R}^{-1}\widetilde{\m{R}}^{\frac{1}{2}}} + \sum_{n=1}^{\overline{N}} \sbra{\m{a}^H_n\widetilde{\m{R}}^{-1}\m{a}_n} p_n + \tr\sbra{\widetilde{\m{R}}^{-1}}\sigma. \label{eq:h1min}}
Note that the first term of the above objective function can be written as:
\equ{\tr\sbra{\widetilde{\m{R}}^{\frac{1}{2}}\m{R}^{-1}\widetilde{\m{R}}^{\frac{1}{2}}} = \min\tr\sbra{\m{X}}, \st \begin{bmatrix}\m{X} & \widetilde{\m{R}}^{\frac{1}{2}} \\ \widetilde{\m{R}}^{\frac{1}{2}} & \m{R} \end{bmatrix}\geq \m{0}}
and hence it is convex in $\m{R}$ as well as in $\sbra{\m{p},\sigma}$. It follows that $h_1$ is convex in $\sbra{\m{p},\sigma}$. Similarly, it holds for $h_2$ that
\equ{\begin{split}h_2
&= \tr\sbra{\m{R}^{-1}\widetilde{\m{R}}^2}+\tr\sbra{\m{R}}-2\tr\sbra{\widetilde{\m{R}}}\\
&= \tr\sbra{\widetilde{\m{R}}\m{R}^{-1}\widetilde{\m{R}}} + \sum_{n=1}^{\overline{N}} \twon{\m{a}_n}^2p_n + M\sigma -2\tr\sbra{\widetilde{\m{R}}}. \end{split}}
The resulting optimization problem is given by:
\equ{\min_{\m{p}\succeq\m{0},\sigma>0} \tr\sbra{\widetilde{\m{R}}\m{R}^{-1}\widetilde{\m{R}}} + \sum_{n=1}^{\overline{N}} \twon{\m{a}_n}^2p_n + M\sigma, \label{eq:h2min}}
which is in a form similar to \eqref{eq:h1min} and therefore is convex as well. Although both \eqref{eq:h1min} and \eqref{eq:h2min} can be cast as second order cone programs (SOCP) or semidefinite programs (SDP) (as shown above), for which standard solvers are available, they cannot be easily solved in practice based on these formulations due to the high dimensionality of the problem (note that $\overline{N}$ can be very large).

We now introduce the SPICE algorithm to cope with the aforementioned computational problems. We focus on the case of $L\geq M$ but similar results also hold in the case of $L< M$. The main result hat underlies SPICE is the following reformulation (see, e.g., \cite{stoica2011spice}):
\equ{\tr\sbra{\widetilde{\m{R}}^{\frac{1}{2}}\m{R}^{-1}\widetilde{\m{R}}^{\frac{1}{2}}} = \min_{\m{C}} \tr\sbra{\m{C}^H \m{P}'^{-1}\m{C}}, \st \m{A}'\m{C} = \widetilde{\m{R}}^{\frac{1}{2}} \label{eq:SPICE_axil}}
and showing that the solution of $\m{C}$ is given by
\equ{\m{C} = \m{P}'\m{A}'^H\m{R}^{-1}\widetilde{\m{R}}^{\frac{1}{2}}. \label{eq:C}}
Inserting \eqref{eq:SPICE_axil} into \eqref{eq:h1min}, we see that the minimization of $h_1$ can be equivalently written as:
\equ{\begin{split}
&\min_{\m{C},\m{p}\succeq\m{0},\sigma>0} \tr\sbra{\m{C}^H \m{P}'^{-1}\m{C}} + \sum_{n=1}^{\overline{N}} \sbra{\m{a}^H_n\widetilde{\m{R}}^{-1}\m{a}_n} p_n + \tr\sbra{\widetilde{\m{R}}^{-1}}\sigma,\\
&\st \m{A}'\m{C} = \widetilde{\m{R}}^{\frac{1}{2}}. \end{split} \label{eq:h1min2}}
Based on \eqref{eq:h1min2}, the SPICE algorithm is derived by iteratively solving for $\m{C}$ and for $\sbra{\m{p},\sigma}$. First, $\m{p}$ and $\sigma$ are initialized using, e.g., the conventional beamformer. Then, $\m{C}$ is updated using \eqref{eq:C} with the latest estimates of $\m{p}$ and $\sigma$. After that, we update $\m{p}$ and $\sigma$ by fixing $\m{C}$ and repeat the process until convergence. Note that $\sbra{\m{p},\sigma}$ can also be determined in closed form, for fixed $\m{C}$. To see this, observe that
\equ{\tr\sbra{\m{C}^H \m{P}'^{-1}\m{C}} = \sum_{n=1}^{\overline{N}} \frac{\twon{\m{C}_n}^2}{p_n} + \frac{\sum_{n={\overline{N}+1}}^{\overline{N}+M} \twon{\m{C}_n}^2}{\sigma}, \label{eq:Csplit}}
where $\m{C}_n$ denotes the $n$th row of $\m{C}$.
Inserting \eqref{eq:Csplit} in \eqref{eq:h1min2}, the solutions $p_n$, $n=1,\dots,\overline{N}$ and $\sigma$ can be obtained as:
{\lentwo\equa{p_n
&=& \frac{\twon{\m{C}_n}}{\sqrt{\m{a}^H_n\widetilde{\m{R}}^{-1}\m{a}_n}}, \quad n=1,\dots,\overline{N}, \label{eq:pn_SPICE} \\ \sigma
&=& \sqrt{ \frac{\sum_{n={\overline{N}+1}}^{\overline{N}+M} \twon{\m{C}_n}^2}{\tr\sbra{\widetilde{\m{R}}^{-1}}}}. \label{eq:sigma_SPICE}
}}Since the problem is convex and the objective function is monotonically decreasing in the iterative process, the SPICE algorithm is expected to converge to the global minimum. Note that the main computational cost of SPICE is for computing $\m{C}$ in \eqref{eq:pn_SPICE} and \eqref{eq:sigma_SPICE} according to \eqref{eq:C}. It follows that SPICE has a per-iteration computational complexity of $O\sbra{\overline{N}M^2}$ given that $\overline{N}>M$. Note also that, as compared to the original SPICE algorithm in \cite{stoica2011new,stoica2011spice}, a certain normalization step of the power estimates is removed here to avoid a global scaling ambiguity of the final power estimates (see also \cite{stoica2012spice}).

We next discuss how the signal sparsity and joint sparsity are exploited in SPICE. Inserting \eqref{eq:pn_SPICE} and \eqref{eq:sigma_SPICE} into \eqref{eq:h1min2}, we see that the SPICE problem is equivalent to:
\equ{\begin{split}
&\min_{\m{C}} \sum_{n=1}^{\overline{N}} \sqrt{\m{a}^H_n\widetilde{\m{R}}^{-1}\m{a}_n} \twon{\m{C}_n} + \sqrt{ \tr\sbra{\widetilde{\m{R}}^{-1}} \sum_{n={\overline{N}+1}}^{\overline{N}+M} \twon{\m{C}_n}^2},\\
&\st \m{A}'\m{C} = \widetilde{\m{R}}^{\frac{1}{2}}. \end{split} \label{eq:h1min3}}
Note that the first term of the objective function in \eqref{eq:h1min3} is nothing but a weighted sum of the $\ell_2$ norm of the first $\overline{N}$ rows of $\m{C}$ (a.k.a.~a weighted $\ell_{2,1}$ norm) and thus promotes the row-sparsity of $\m{C}$. Therefore, it is expected that most of $\twon{\m{C}_n}$, $n=1,\dots,\overline{N}$ will be equal to zero. This together with \eqref{eq:pn_SPICE} implies that most of $p_n$, $n=1,\dots,\overline{N}$ will be zero and so sparsity is achieved. The joint sparsity is achieved by the assumption that the entries in each row of $\m{X}$ have identical variance $p_n$.

SPICE is related to the square-root LASSO in the single snapshot case. In particular, it was shown in \cite{rojas2013note,babu2014connection} that the SPICE problem in \eqref{eq:h2min} is equivalent to
\equ{\min_{\m{x}} \onen{\m{x}} + \twon{\m{y}-\m{A}\m{x}},}
which is nothing but the square-root LASSO in \eqref{eq:SR_lasso} with $\tau=1$.

Finally, note that the decomposition in \eqref{eq:RA1PAH} is not unique in general (see Corollary \ref{cor:Tinsignalnoise} for detail). A direct consequence of this observation is that the SPICE algorithm generally does not provide unique estimates of $\m{p}$ and $\sigma$ \cite{yang2014discretization}. This problem will be fixed in the gridless version of SPICE that will be introduced in Sections \ref{sec:GLS_SMV} and \ref{sec:GLS_MMV}.

\subsection{Maximum Likelihood Estimation}
The joint sparsity can also be exploited in the MLE, in a similar way as in SPICE. Assume that $\m{x}(t)$, $t=1,\dots,L$ are i.i.d. multivariate Gaussian distributed with mean zero and covariance $\m{P}=\diag\sbra{\m{p}}$. Also assume i.i.d. Gaussian noise with variance $\sigma$ and that $\m{X}$ and $\m{E}$ are independent. Then, we have that the data snapshots $\m{y}(t)$, $t=1,\dots,L$ are i.i.d. Gaussian distributed with mean zero and covariance $\m{R} = \m{A}\m{P}\m{A}^H +\sigma\m{I}$. The negative log-likelihood function associated with $\m{Y}$ is therefore given by
\equ{\cL\sbra{\m{p},\sigma} = \log\abs{\m{R}} + \tr\sbra{\m{R}^{-1}\widetilde{\m{R}}}}
where $\widetilde{\m{R}}$ is the sample covariance matrix as defined in the preceding subsection. It follows that the parameters $\m{p}$ and $\sigma$ can be estimated by solving the problem:
\equ{\min_{\m{p},\sigma}\log\abs{\m{R}} + \tr\sbra{\m{R}^{-1}\widetilde{\m{R}}}. \label{eq:negloglike_MMV}}
Owing to the analogy between \eqref{eq:negloglike_MMV} and its single snapshot counterpart (see \eqref{eq:negloglike}), it should come as no surprise that the algorithms developed for the single snapshot case can also be applied to the multiple snapshot case with minor modifications. As an example, using a similar MM procedure, LIKES \cite{stoica2012spice} can be extended to this multiple snapshot case.

The multiple snapshot MLE  has been studied within the framework of SBL or Bayesian compressed sensing (see, e.g., \cite{wipf2007empirical,ji2009multi,liu2012efficient,carlin2013directions}). To exploit the joint sparsity of $\m{x}(t)$, $t=1,\dots,L$, an identical sparse prior was assumed for all of them. The EM algorithm can also be used to perform parameter estimation via minimizing the objective in \eqref{eq:negloglike_MMV}.


\subsection{Remarks on Grid Selection}
Based on the on-grid data model in \eqref{eq:datamodel_D}, we have introduced several sparse optimization methods for DOA estimation in the preceding subsections. While we have focused on how the temporal redundancy or joint sparsity of the snapshots can be exploited, a major problem that remains unresolved is grid selection, i.e., the selection of the grid points $\overline{\m{\theta}}$ in the data model \eqref{eq:datamodel_D}. Since discrete grid points are used to approximate the continuous DOA domain, intuitively, the grid should be chosen as fine as possible to improve the approximation accuracy. However, this can be problematic in two respects. Theoretically, a dense grid results in highly coherent atoms and hence few DOAs can be estimated according to the analysis based on the mutual coherence and RIP. Moreover, a too dense grid is not acceptable from an algorithmic viewpoint, since it will dramatically increase the computational complexity of an algorithm and also might cause slow convergence and numerical instability due to nearly identical adjacent atoms \cite{stoica2012sparse,austin2013dynamic}.

To overcome the theoretical bottleneck mentioned above, the so-called coherence-inhibiting techniques have been proposed and incorporated in the existing sparse optimization methods to avoid solutions with closely located atoms \cite{duarte2013spectral,fannjiang2012coherence}. Additionally, with the development of recent gridless sparse methods it was shown that the local coherence between nearly located atoms actually does not matter for the convex relaxation approach if the true DOAs are appropriately separated \cite{candes2013towards} (details will be provided in Section \ref{sec:gridless}).

To improve the computational speed and accuracy, a heuristic grid refinement strategy was proposed that suggests using a coarse grid at the initial stage and then gradually refining it based on the latest estimates of the DOAs \cite{malioutov2005sparse}. A grid selection approach was also proposed by quantifying the similarity between the atoms in a grid bin \cite{stoica2012sparse}. In particular, suppose that in \eqref{eq:datamodel_D} the DOA interval $\sbra{\theta_n-\frac{r}{2}, \theta_n+\frac{r}{2}}$ is approximated by some grid point $\theta_n$, where $r>0$ denotes the grid interval. Then, on this interval the similarity is measured by the rank of the matrix defined by
\equ{\m{C}_n = \int_{\theta_n-\frac{r}{2}}^{\theta_n+\frac{r}{2}} \m{a}\sbra{v}\m{a}^H\sbra{v}\text{d}v. }
If $\rank\sbra{\m{C}_n} \approx 1$, then it is said that the grid is dense enough; otherwise, a denser grid is required. However, a problem with this criterion is that it can only be evaluated heuristically.

In summary, grid selection is an important problem that affects the practical DOA estimation accuracy, the computational speed and the theoretical analysis. A completely satisfactory solution to this problem within the framework of the on-grid methods seems hard to obtain since there always exist mismatches between the adopted discrete grid points and the true continuous DOAs.

\section{Off-Grid Sparse Methods} \label{sec:offgrid}
We have discussed the on-grid sparse methods in the preceding section, for which grid selection is a difficult problem and will inevitably result in grid mismatch. To resolve the grid mismatch problem, in this section we turn to the so-called off-grid sparse methods. In these methods, a grid is still required to perform sparse estimation but, unlike the on-grid methods, the DOA estimates are not restricted to be on the grid. We will mainly talk about two kinds of off-grid sparse methods: one is based on a fixed grid and joint estimation of the sparse signal and the grid offset, and the other relies on a dynamic grid. The main focus of the following discussions is on how to solve the grid mismatch.

\subsection{Fixed Grid}

\subsubsection{Data Model}
With a fixed grid $\overline{\m{\theta}}=\lbra{\overline{\theta}_1,\dots,\overline{\theta}_{\overline{N}}}$, an off-grid data model can be introduced as follows \cite{zhu2011sparsity}. Suppose without loss of generality that $\overline{\m{\theta}}$ consists of uniformly spaced grid points with the grid interval $r = \theta_2-\theta_1\propto \frac{1}{\overline{N}}$. For any DOA $\theta_k$, suppose $\overline{\theta}_{n_k}$ is the nearest grid point with $\abs{\theta_k-\overline{\theta}_{n_k}}\leq \frac{r}{2}$. We approximate the steering vector/atom $\m{a}\sbra{\theta_k}$ using a first-order Taylor expansion:
\equ{\m{a}\sbra{\theta_k} \approx \m{a}\sbra{\overline{\theta}_{n_k}} + \m{b}\sbra{\overline{\theta}_{n_k}}\sbra{\theta_k-\overline{\theta}_{n_k}},}
where $\m{b}\sbra{\overline{\theta}_{n_k}}= \m{a}'\sbra{\overline{\theta}_{n_k}}$ (the derivative of $\m{a}\sbra{\theta}$). Similar to \eqref{eq:datamodel_D}, we then obtain the following data model:
\equ{\m{Y} = \m{\Phi}\sbra{\m{\beta}} \m{X} +\m{E}, \label{eq:datamodel_D_OG}}
where $\m{\Phi}\sbra{\m{\beta}} = \m{A} + \m{B}\diag\sbra{\m{\beta}}$, $\m{A}=\mbra{\m{a}\sbra{\overline{\theta}_1},\dots,\m{a}\sbra{\overline{\theta}_{\overline{N}}}}$ is as defined previously, $\m{B}=\mbra{\m{b}\sbra{\overline{\theta}_1},\dots,\m{b}\sbra{\overline{\theta}_1}}$ and $\m{\beta}=\mbra{\beta_1,\dots,\beta_{\overline{N}}}\in\mbra{-\frac{r}{2},\frac{r}{2}}^{\overline{N}}$, with
{\lentwo\equa{x_n(t)
&=& \left\{\begin{array}{ll} s_k(t), & \text{ if } \overline{\theta}_{n} = \overline{\theta}_{n_k}; \\ 0, & \text{ otherwise, } \end{array}\right.\\ \beta_n
&=& \left\{\begin{array}{ll} \theta_k-\overline{\theta}_{n_k}, & \text{ if } \overline{\theta}_{n} = \overline{\theta}_{n_k}; \\ 0, & \text{ otherwise, } \end{array}\right.
\quad n=1,\dots,\overline{N},\; t=1,\dots,L.
}}It follows from \eqref{eq:datamodel_D_OG} that the DOA estimation problem can be formulated as sparse representation with uncertain parameters. In particular, once the row-sparse matrix $\m{X}$ and $\m{\beta}$ can be estimated from $\m{Y}$, then the DOAs can be estimated using the row-support of $\m{X}$ shifted by the offset $\m{\beta}$.

Compared to the on-grid model in \eqref{eq:datamodel_D}, the additional grid offset parameters $\beta_n$, $n=1,\dots,\overline{N}$ are introduced in the off-grid model in \eqref{eq:datamodel_D_OG}. Note that \eqref{eq:datamodel_D_OG} reduces to \eqref{eq:datamodel_D} if $\m{\beta}=\m{0}$. While \eqref{eq:datamodel_D} is based on a zeroth order approximation of the true data model, which causes grid mismatch, \eqref{eq:datamodel_D_OG} can be viewed as a first-order approximation in which the grid mismatch can be partially compensated by jointly estimating the grid offset. Based on the off-grid model in \eqref{eq:datamodel_D_OG}, several methods have been proposed for DOA estimation by jointly estimating $\m{X}$ and $\m{\beta}$ (see, e.g., \cite{zhu2011sparsity,zheng2011directions,yang2012robustly, yang2013off, tan2014joint, zhang2014off, lasserre2015bayesian, si2015off, chen2015modified, wang2016novel, wu2016direction,zhao2016array, han2015two,bernhardt2016compressed, fei2016off, shen2017underdetermined, yang2016efficient,sun2016iterative}). Out of these methods we present the $\ell_1$-based optimization and SBL in the next subsections.


\subsubsection{$\ell_1$ Optimization}
Inspired by the standard sparse signal recovery approach, several $\ell_1$ optimization methods have been proposed to solve the off-grid DOA estimation problem. In \cite{zhu2011sparsity}, a sparse total least-squares (STLS) approach was proposed which, in the single snapshot case, solves the following LASSO-like problem:
\equ{\min_{\m{x},\m{\beta}} \lambda_1 \onen{\m{x}} + \frac{1}{2}\twon{\m{y} - \mbra{\m{A}+\m{B}\diag\sbra{\m{\beta}}}\m{x}}^2 + \lambda_2\twon{\m{\beta}}^2,\label{eq:STLS}}
where $\lambda_1$ and $\lambda_2$ are regularization parameters. In \eqref{eq:STLS}, the prior information that $\m{\beta}\in\mbra{-\frac{r}{2},\frac{r}{2}}^{\overline{N}}$ is not used. To heuristically control the magnitude of $\m{\beta}$, its power is also minimized. Note that the problem in \eqref{eq:STLS} is nonconvex due to the bilinear term $\diag\sbra{\m{\beta}}\m{x}$. To solve \eqref{eq:STLS}, an alternating algorithm is adopted, iteratively solving for $\m{x}$ and $\m{\beta}$. Moreover, \eqref{eq:STLS} can be easily extended to the multiple snapshot case by using $\ell_{2,1}$ optimization to exploit the joint sparsity in $\m{X}$ (as in the preceding section). A difficult problem of these methods is parameter tuning, i.e., how to choose $\lambda_1$ and $\lambda_2$.

To exploit the prior knowledge that $\m{\beta}\in\mbra{-\frac{r}{2},\frac{r}{2}}^{\overline{N}}$, the following BPDN-like formulation was proposed in the single snapshot case \cite{yang2012robustly}:
\equ{\min_{\m{x},\;\m{\beta}\in \mbra{-\frac{r}{2},\frac{r}{2}}^{\overline{N}}} \onen{\m{x}}, \st \twon{\m{y} - \mbra{\m{A}+\m{B}\diag\sbra{\m{\beta}}}\m{x}}\leq \eta. \label{eq:BPDN_beta}}
Note that \eqref{eq:BPDN_beta} can be easily extended to the multiple snapshot case by using the $\ell_{2,1}$ norm. In \eqref{eq:BPDN_beta}, $\eta$ can be set according to information about the noise level and a possible estimate of the modeling error. Similar to \eqref{eq:STLS}, \eqref{eq:BPDN_beta} is nonconvex, and a similar alternating algorithm can be implemented to monotonically decrease the value of the objective function. Note that if $\m{\beta}$ is initialized as a zero vector, then the first iteration coincides with the standard BPDN.

It was shown in \cite{yang2012robustly} that if the matrix $\mbra{\m{A}, \m{B}}$ satisfies a certain RIP condition, then both $\m{x}$ and $\m{\beta}$ can be stably reconstructed, as in the standard sparse representation problem, with the reconstruction error being proportional to the noise level $\eta$. This means that in the ideal case of $\eta=0$ (assuming there is no noise or modeling error), $\m{x}$ and $\m{\beta}$ can be exactly recovered. A key step in showing this result is reformulating \eqref{eq:BPDN_beta} as
\equ{\min_{\m{x},\;\m{\beta}\in \mbra{-\frac{r}{2},\frac{r}{2}}^{\overline{N}}} \onen{\m{x}}, \st \twon{\m{y} - \begin{bmatrix} \m{A} & \m{B} \end{bmatrix} \begin{bmatrix}\m{x} \\ \m{\beta}\odot\m{x}\end{bmatrix}}\leq \eta, \label{eq:BPDN_beta2}}
where $\odot$ denotes the element-wise product. Although the RIP condition cannot be easily applied to the case of dense grid, the aforementioned result implies, to some extent, the superior performance of this off-grid optimization method as compared to the on-grid approach.

Following the lead of \cite{yang2012robustly}, a convex optimization method was proposed in \cite{tan2014joint} by exploiting the joint sparsity of $\m{x}$ and $\m{v}= \m{\beta}\odot\m{x}$. In particular, the following problem was formulated:
\equ{\min_{\m{x},\;\m{v}} \lambda \twonen{\begin{bmatrix}\m{x} & \m{v} \end{bmatrix}} + \frac{1}{2}\twon{\m{y} - \begin{bmatrix} \m{A} & \m{B} \end{bmatrix} \begin{bmatrix}\m{x} \\ \m{v}\end{bmatrix}}^2, \label{eq:BPDN_JP}}
which is equivalent to the following problem, for appropriate parameter choices:
\equ{\min_{\m{x},\;\m{v}} \twonen{\begin{bmatrix}\m{x} & \m{v} \end{bmatrix}}, \st \twon{\m{y} - \begin{bmatrix} \m{A} & \m{B} \end{bmatrix} \begin{bmatrix}\m{x} \\ \m{v}\end{bmatrix}}\leq \eta. \label{eq:BPDN_JP2}}
This approach is advantageous in that it is convex and can be globally solved in a polynomial time, with similar theoretical guarantees as provided in \cite{yang2012robustly}. However, it is worth noting that the prior knowledge on $\m{\beta}$ cannot be exploited in this method. Additionally, the obtained solution for $\beta_n=\frac{v_n}{x_n}$ might not even be real. To resolve this problem, \cite{tan2014joint} suggests a two-stage solution: 1) obtain $\m{x}$ from \eqref{eq:BPDN_JP}, and 2) fix $\m{x}$ and solve for $\m{\beta}$ by minimizing $\twon{\m{y} - \mbra{\m{A}+\m{B}\diag\sbra{\m{\beta}}}\m{x}}$.

\subsubsection{Sparse Bayesian Learning}
A systematic approach to off-grid DOA estimation, called off-grid sparse Bayesian inference (OGSBI), was proposed in \cite{yang2013off} within the framework of SBL in the multiple snapshot case. In order to estimate the additional parameter $\m{\beta}$, it is assumed that $\beta_n$, $n=1,\dots,\overline{N}$ are i.i.d. uniformly distributed on the interval $\mbra{-\frac{r}{2},\frac{r}{2}}$. In the resulting EM algorithm, the posterior distribution of the row-sparse signal $\m{X}$ can be computed in the expectation step as in the standard SBL. In the maximization step, $\m{\beta}$ is also updated, in addition to updating the power $\m{p}$ of the row-sparse signal and the noise variance $\sigma$. As in the standard SBL, the likelihood is guaranteed to monotonically increase and hence convergence of the algorithm can be obtained.

\subsection{Dynamic Grid}

\subsubsection{Data Model}
The data model now uses a dynamic grid $\overline{\m{\theta}}$ in the sense that the grid points $\theta_n$, $n=1,\dots,\overline{N}$ are not fixed:
\equ{\m{Y} = \m{A}\sbra{\overline{\m{\theta}}} \m{X} + \m{E}. \label{eq:datamodel_D_DG}}
For this model we need to jointly estimate the row-sparse matrix $\m{X}$ and the grid $\overline{\m{\theta}}$. Once they are obtained, the DOAs are estimated using those grid points of $\overline{\m{\theta}}$ corresponding to the nonzero rows of $\m{X}$. Since $\overline{\theta}_n$'s are estimated from the data and can be any values in the continuous DOA domain, this off-grid data model is accurate and does not suffer from grid mismatch. However, the difficulty lies in designing an algorithm for the joint estimation of $\m{X}$ and $\overline{\m{\theta}}$, due to the nonlinearity of the mapping $\m{a}\sbra{\theta}$. Note that the following algorithms that we will introduce are designated as off-grid methods, instead of gridless, since grid selection remains involved in them (e.g., choice of $\overline{N}$ and initialization of $\overline{\m{\theta}}$), which affects the computational speed and accuracy of the algorithms.

\subsubsection{Algorithms}
Several algorithms have been proposed based on the data model in \eqref{eq:datamodel_D_DG}. The first class is within the framework of SBL (see, e.g., \cite{shutin2011sparse,shutin2013incremental,hu2012compressed,hu2013fast}). But instead of using the EM algorithm as previously, a variational EM algorithm (or variational Bayesian inference) is typically exploited to carry out the sparse signal and parameter estimation. The reason is that the posterior distribution of the sparse vector $\m{x}$ usually cannot be explicitly computed here, and that distribution is required by the EM but not by the variational EM. The main difficulty of these algorithms is the update of $\m{\theta}$ due to the strong nonlinearity. Because closed-form solutions are not available, only numerical approaches can be used.

Another class of methods uses $\ell_1$ optimization. In the single snapshot case, as an example, the paper \cite{austin2013dynamic} used a small $\overline{N}\geq K$ and attempted to solve the following $\ell_1$ optimization problem by iteratively updating $\m{x}$ and $\m{\theta}$:
\equ{\min_{\m{x},\overline{\m{\theta}}} \lambda \onen{\m{x}} + \frac{1}{2}\twon{\m{y} - \m{A}\sbra{\overline{\m{\theta}}} \m{x}}^2. \label{eq:BP_DG}}
To avoid the possible convergence of some $\overline{\theta}_n$'s to the same value, an additional (nonconvex) term $g\sbra{\overline{\m{\theta}}}$ is included to penalize closely located parameters:
\equ{\min_{\m{x},\overline{\m{\theta}}} \lambda_1 \onen{\m{x}} + \frac{1}{2}\twon{\m{y} - \m{A}\sbra{\overline{\m{\theta}}} \m{x}}^2 + \lambda_2 g\sbra{\overline{\m{\theta}}}, \label{eq:BP_DG2} }
where $\lambda_1$ and $\lambda_2$ are regularization parameters that need be tuned. Note that both \eqref{eq:BP_DG} and \eqref{eq:BP_DG2} are nonconvex. Even for given $\m{x}$, it is difficult to solve for $\overline{\m{\theta}}$. Moreover, parameter tuning is tricky. Note that $\ell_q$, $q<1$ optimization was also considered in \cite{austin2013dynamic} to enhance sparsity but it suffers from similar problems.

To promote sparsity, similar to $\ell_1$ optimization, the following problem was proposed in \cite{fang2014super,fang2016super}:
\equ{\min_{\m{x},\overline{\m{\theta}}} \sum_{n=1}^{\overline{N}} \lambda\log\sbra{\abs{x_n}^2+\epsilon} + \twon{\m{y} - \m{A}\sbra{\overline{\m{\theta}}} \m{x}}^2. \label{eq:fangjun}}
To locally solve \eqref{eq:fangjun}, $\m{x}$ and $\overline{\m{\theta}}$ are iteratively updated. To solve for $\m{x}$ in closed form, the first term of the objective in \eqref{eq:fangjun} is replaced by a quadratic surrogate function that guarantees the decrease of the objective. The gradient descent method is then used to solve for $\overline{\m{\theta}}$. While it is generally difficult to choose $\lambda$, \cite{fang2016super} suggested setting $\lambda$ proportional to the inverse of the noise variance, leaving a constant coefficient to be tuned.

To conclude, in this section we introduced several off-grid sparse optimization methods for DOA estimation. By adopting a dynamic grid or estimating the grid offset jointly with the sparse signal, the grid mismatch encountered in the on-grid sparse methods can be overcome. However, this introduces more variables that need be estimated and complicates the algorithm design. As a consequence, most of the presented algorithms involve nonconvex optimization and thus only local convergence can be guaranteed (except for the algorithm in \cite{tan2014joint}). Moreover, few theoretical guarantees can be obtained for most of the algorithms (see, however, \cite{yang2012robustly,tan2014joint}).

\section{Gridless Sparse Methods} \label{sec:gridless}
In this section, we present several recent DOA estimation approaches that are designated as the gridless sparse methods. As their name suggests, these methods do not require gridding of the direction domain. Instead, they directly operate in the continuous domain and therefore can completely resolve the grid mismatch problem. Moreover, they are convex and have strong theoretical guarantees. However, so far, this kind of methods can only be applied to uniform or sparse linear arrays. Therefore, naturally, in this section we treat the DOA estimation problem as frequency estimation, following the discussions in Section \ref{sec:model}.

The rest of this section is organized as follows. We first revisit the data model in the context of ULAs and SLAs. We then introduce a mathematical tool known as the Vandermonde decomposition of Toeplitz covariance matrices, which is crucial for most gridless sparse methods. Finally, we discuss a number of gridless sparse methods for DOA/frequency estimation in the case of a single snapshot, followed by the case of multiple snapshots. The atomic norm and gridless SPICE methods will be particularly highlighted.


\subsection{Data Model} \label{sec:model_restate}
For convenience, we restate the data model that will be used in this section. For an $M$-element ULA, the array data are modeled as:
\equ{\m{Y} = \m{A}(\m{f})\m{S}+\m{E},\label{formu:observation_model2}}
where $f_k\in\bT$, $k=1,\dots,K$ are the frequencies of interest, which have a one-to-one relationship to the DOAs, $\m{A}(\m{f}) = \mbra{\m{a}\sbra{f_1},\dots, \m{a}\sbra{f_K}}\in\bC^{M\times K}$ is the array manifold matrix, and $\m{a}\sbra{f} = \mbra{1,e^{i2\pi f},\dots, e^{i2\pi(M-1)f}}^T\in\bC^M$ is the steering vector.

For an $M$-element SLA, suppose that the array is obtained from an $N$-element virtual ULA by retaining the antennas indexed by the set $\Omega = \lbra{\Omega_1,\dots,\Omega_M}$, where $N\geq M$ and $1\leq \Omega_1<\dots<\Omega_M\leq N$. In this case, we can view \eqref{formu:observation_model2} as the data model with the virtual ULA. Then the data model of the SLA is given by
\equ{\m{Y}_{\Omega} = \m{A}_{\Omega}(\m{f})\m{S}+\m{E}_{\Omega}. \label{formu:observation_model3}}
Therefore, \eqref{formu:observation_model2} can be considered as a special case of \eqref{formu:observation_model3} in which $M=N$ and $\Omega = \lbra{1,\dots,N}$. Given $\m{Y}$ (or $\m{Y}_{\Omega}$ and $\Omega$), the objective is to estimate the frequencies $f_k$'s (note that the source number $K$ is usually unknown).

In the single snapshot case the above problem coincides with the line spectral estimation problem. Since the first gridless sparse methods were developed for the latter problem, we present them in the single snapshot case and then discuss how they can be extended to the multiple snapshot case by exploiting the joint sparsity of the snapshots. Before doing that, an important mathematical tool is introduced in the following subsection.

\subsection{Vandermonde Decomposition of Toeplitz Covariance Matrices} \label{sec:VD}
The Vandermonde decomposition of Toeplitz covariance matrices plays an important role in this section. This classical result was discovered by Carath\'{e}odory and Fej\'{e}r in 1911 \cite{caratheodory1911zusammenhang}. It has become important in the area of data analysis and signal processing since the 1970s when it was rediscovered by Pisarenko and used for frequency retrieval from the data covariance matrix \cite{pisarenko1973retrieval}. From then on, the Vandermonde decomposition has formed the basis of a prominent subset of methods for frequency and DOA estimation, viz.~the subspace-based methods. To see why it is so, let us consider the data model in \eqref{formu:observation_model2} and assume uncorrelated sources. In the noiseless case, the data covariance matrix is given by
\equ{\m{R} = \bE\m{y}(t)\m{y}^H(t) = \m{A}\sbra{\m{f}}\diag\sbra{\m{p}}\m{A}^H\sbra{\m{f}}, \label{eq:Rnoiseless}}
where $p_k>0$, $k=1,\dots,K$ are the powers of the sources. It can be easily verified that $\m{R}$ is a (Hermitian) Toeplitz matrix that can be written as:
\equ{\m{R} = \m{T}\sbra{\m{u}}= \begin{bmatrix}u_1 & u_2 & \cdots & u_N\\ {u}_2^* & u_1 & \cdots & u_{N-1}\\ \vdots & \vdots & \ddots & \vdots \\ {u}_N^* & {u}_{N-1}^* & \cdots & u_1\end{bmatrix},}
where $\m{u}\in\bC^N$. Moreover, $\m{R}$ is PSD and has rank $K$ under the assumption that $K<N$. The Vandermonde decomposition result states that any rank-deficient, PSD Toeplitz matrix $\m{T}$ can be uniquely decomposed as in \eqref{eq:Rnoiseless}. Equivalently stated, this means that the frequencies can be exactly retrieved from the data covariance matrix. Formally, the result is stated in the following theorem, a proof of which (inspired by \cite{gurvits2002largest}) is also provided; note that the proof suggests a way of computing the decomposition.

\begin{thm} Any PSD Toeplitz matrix $\m{T}(\m{u})\in\bC^{N\times N}$ of rank $r\leq N$ admits the following $r$-atomic Vandermonde decomposition:
\equ{\m{T} = \sum_{k=1}^r p_k \m{a}\sbra{f_k}\m{a}^H\sbra{f_k} = \m{A}\sbra{\m{f}}\diag\sbra{\m{p}}\m{A}^H\sbra{\m{f}}, \label{eq:VD}}
where $p_k>0$, and $f_k\in\bT$, $k=1,\dots,r$ are distinct. Moreover, the decomposition in \eqref{eq:VD} is unique if $r< N$. \label{thm:VD}
\end{thm}
\begin{proof} We first consider the case of $r=\rank\sbra{\m{T}}\leq N-1$. Since $\m{T}\geq \m{0}$, there exists $\m{V}\in\bC^{N\times r}$ satisfying $\m{T} = \m{V}\m{V}^H$. Let $\m{V}_{-N}$ and $\m{V}_{-1}$ be the matrices obtained from $\m{V}$ by removing its last and first row, respectively. By the structure of $\m{T}$, we have that $\m{V}_{-N}\m{V}_{-N}^H = \m{V}_{-1}\m{V}_{-1}^H$. Thus there must exist an $r\times r$ unitary matrix $\m{Q}$ satisfying $\m{V}_{-1} = \m{V}_{-N}\m{Q}$ (see, e.g., \cite[Theorem 7.3.11]{horn2012matrix}). It follows that $\m{V}_j = \m{V}_1\m{Q}^{j-1}$, $j=2,\dots,N$ and therefore,
\equ{u_j = \m{V}_1 \m{Q}^{1-j}\m{V}_1^H, \quad j=1,\dots,N, \label{eq:tjinvU}}
where $\m{V}_j$ is the $j$th row of $\m{V}$.
Next, write the eigen-decomposition of the unitary matrix $\m{Q}$, which is guaranteed to exist, as
\equ{\m{Q} = \widetilde{\m{Q}} \diag\sbra{z_1,\dots,z_r} \widetilde{\m{Q}}^H, \label{eqw:UinUz}}
where $\widetilde{\m{Q}}$ is also an $r\times r$ unitary matrix and $z_k$'s are the eigenvalues. Since the eigenvalues of a unitary matrix must have unit magnitude, we can find $f_k\in\bT$, $k=1,\dots,r$ satisfying $z_k = e^{i2\pi f_k}$, $k=1,\dots,r$. Inserting \eqref{eqw:UinUz} into \eqref{eq:tjinvU} and letting $p_k = \abs{\m{V}_1\widetilde{\m{Q}}_{:k}}^2>0$, $k=1,\dots,r$, where $\widetilde{\m{Q}}_{:k}$ denotes the $k$th column of $\widetilde{\m{Q}}$, we have that
\equ{u_j = \sum_{k=1}^r p_k e^{-i2\pi (j-1) f_k}.}
It follows that \eqref{eq:VD} holds. Moreover, $f_k$, $k=1,\dots,r$ are distinct since otherwise $\rank\sbra{\m{T}} < r$, which cannot be true.

We now consider the case of $r=N$ for which $\m{T} >0$. Let us arbitrarily choose $f_N \in\bT$ and let $p_N = \sbra{\m{a}^H\sbra{f_N} \m{T}^{-1} \m{a}\sbra{f_N}}^{-1}>0$. Moreover, we define a new vector $\m{u}'\in\bC^N$ by
\equ{u'_j = u_j - p_N e^{-i2\pi (j-1) f_N}. \label{eq:t1j}}
It can be readily verified that
\lentwo{\equa{\m{T}\sbra{\m{u}'}
&=& \m{T}\sbra{\m{u}} - p_N\m{a}\sbra{f_N}\m{a}^H\sbra{f_N}, \label{eq:Tu1} \\ \m{T}\sbra{\m{u}'}
&\geq& \m{0}, \label{eq:Tu1psd} \\\rank\sbra{\m{T}\sbra{\m{u}'}}
&=& N-1. \label{eq:Tu1rank}
}}Therefore, from the result in the case of $r\leq N-1$ proven above, $\m{T}\sbra{\m{u}'}$ admits a Vandermonde decomposition as in \eqref{eq:VD} with $r=N-1$. It then follows from \eqref{eq:Tu1} that $\m{T}\sbra{\m{u}}$ admits a Vandermonde decomposition with $N$ ``atoms''.

We finally show the uniqueness in the case of $r\leq N-1$. To do so, suppose there exists another decomposition $\m{T} = \m{A}\sbra{\m{f}'} \m{P}'\m{A}^H\sbra{\m{f}'}$ in which $p'_j>0$, $j=1,\dots,r$ and $f'_j\in\bT$ are distinct. It follows from the equation
\equ{\m{A}\sbra{\m{f}'} \m{P}'\m{A}^H\sbra{\m{f}'}=\m{A}\sbra{\m{f}} \m{P}\m{A}^H\sbra{\m{f}}}
that there exists an $r\times r$ unitary matrix $\m{Q}'$ such that $\m{A}\sbra{\m{f}'}\m{P}'^{\frac{1}{2}} = \m{A}\sbra{\m{f}} \m{P}^{\frac{1}{2}} \m{Q}'$ and therefore,
\equ{\m{A}\sbra{\m{f}'} = \m{A}\sbra{\m{f}} \m{P}^{\frac{1}{2}} \m{Q}' \m{P}'^{-\frac{1}{2}}.}
This means that for every $j = 1,\dots,r$, $\m{a}\sbra{f'_j}$ lies in the range space spanned by $\lbra{\m{a}\sbra{f_k}}_{k=1}^r$. By the fact that $r\leq N-1$ and that any $N$ atoms $\m{a}\sbra{f_k}$ with distinct $f_k$'s are linearly independent, we have that $f'_j\in\lbra{f_k}_{k=1}^r$ and thus the two sets $\lbra{f'_j}_{j=1}^r$ and $\lbra{f_k}_{k=1}^r$ are identical. It follows that the above two decompositions of $\m{T}\sbra{\m{u}}$ must be identical.
\end{proof}

Note that the proof of Theorem \ref{thm:VD} provides a computational approach to the Vandermonde decomposition. We simply consider the case of $r\leq N-1$, since in the case of $r= N$ we can arbitrarily choose $f_N\in\bT$ first. We use the following result:
\equ{\sbra{\m{V}_{-N}^H\m{V}_{-1} - z_k \m{V}_{-N}^H\m{V}_{-N}}\widetilde{\m{Q}}_{:k} = 0, \label{eq:geneigen}}
which can be shown along the lines of the above proof. To retrieve the frequencies and the powers from $\m{T}$, we first compute $\m{V}\in\bC^{N\times r}$ satisfying $\m{T} = \m{V}\m{V}^H$ using, e.g., the Cholesky decomposition. After that, we use \eqref{eq:geneigen} and compute $z_k$ and $\widetilde{\m{Q}}_{:k}$, $k=1,\dots,r$ as the eigenvalues and the normalized eigenvectors of the matrix pencil $\sbra{\m{V}_{-N}^H\m{V}_{-1}, \m{V}_{-N}^H\m{V}_{-N}}$. Finally, we obtain $f_k = \frac{1}{2\pi}\Im \ln z_k\in\bT$ and $p_k = \abs{\m{V}_1\widetilde{\m{Q}}_{:k}}^2$, $k=1,\dots,r$, where $\Im$ gives the imaginary part of its argument. In fact, this matrix pencil approach is similar to the ESPRIT algorithm that computes the frequency estimates from an estimate of the data covariance matrix.

In the presence of homoscedastic noise, the data covariance matrix $\m{R}$ remains Toeplitz. In this case, it is natural to decompose the Toeplitz covariance matrix as the sum of the signal covariance and the noise covariance. Consequently, the following corollary of Theorem \ref{thm:VD} can be useful in such a case. The proof is straightforward and will be omitted.
\begin{cor} Any PSD Toeplitz matrix $\m{T}(\m{u})\in\bC^{N\times N}$ can be uniquely decomposed as:
\equ{\m{T} = \sum_{k=1}^r p_k \m{a}\sbra{f_k}\m{a}^H\sbra{f_k} + \sigma\m{I} = \m{A}\sbra{\m{f}}\diag\sbra{\m{p}}\m{A}^H\sbra{\m{f}} + \sigma \m{I}, \label{eq:VD2}}
where $\sigma = \lambda_{\text{min}}\sbra{\m{T}}$ (the smallest eigenvalue of $\m{T}$), $r=\rank\sbra{\m{T}-\sigma\m{I}}< N$, $p_k>0$, and $f_k\in\bT$, $k=1,\dots,r$ are distinct. \label{cor:Tinsignalnoise}
\end{cor}

\begin{rem} Note that the uniqueness of the decomposition in Corollary \ref{cor:Tinsignalnoise} is guaranteed by the condition that $\sigma = \lambda_{\text{min}}\sbra{\m{T}}$. If the condition is violated by letting $0\leq \sigma < \lambda_{\text{min}}\sbra{\m{T}}$ (in such a case $\m{T}$ has full rank and $r\geq N$), then the decomposition in \eqref{eq:VD2} cannot be unique. \label{rem:uniqueTdecomp}
\end{rem}

The Vandermonde decomposition of Toeplitz covariance matrices forms an important tool in several recently proposed gridless sparse methods. In particular, these methods transform the frequency estimation problem into the estimation of a PSD Toeplitz matrix in which the frequencies are encoded. Once the matrix is computed, the frequencies can be retrieved from its Vandermonde decomposition. Therefore, these gridless sparse methods can be viewed as being covariance-based by interpreting the Toeplitz matrix as the data covariance matrix (though it might not be, since certain statistical assumptions may not be satisfied). In contrast to conventional subspace methods that estimate the frequencies directly from the sample covariance matrix, the gridless methods utilize more sophisticated optimization approaches to estimate the data covariance matrix by exploiting its special structures, e.g., Toeplitz, low rank and PSDness, and therefore are expected to achieve superior performance.

\subsection{The Single Snapshot Case}
In this subsection we introduce several gridless sparse methods for DOA/frequency estimation in the single snapshot case (a.k.a.~the line spectral estimation problem). Two kinds of methods will be discussed: deterministic optimization methods, e.g., the atomic norm and the Hankel-based nuclear norm methods \cite{candes2013towards,candes2013super,tang2012compressive,bhaskar2013atomic, tang2012compressed,tang2015near, chen2014robust,yang2015gridless,azais2015spike, duval2015exact,fernandez2016super,cai2016robust, sun2015noncontact}, and a covariance fitting method that is a gridless version of SPICE \cite{yang2014discretization, yang2015gridless,stoica2014gridless,yang2016gridless1,sun2016new,zhang2016gridless,bao2016dlsla}. The connections between these methods will also be investigated. By `deterministic' we mean that we do not make any statistical assumptions on the signal of interest. Instead, the signal is deterministic and it is sought as the sparsest candidate, measured by a certain sparse metric, among a prescribed set.

\subsubsection{A General Framework for Deterministic Methods}

In the single snapshot case, corresponding to \eqref{formu:observation_model3}, the data model in the SLA case is given by:
\equ{\m{y}_{\Omega}=\m{z}_{\Omega} + \m{e}_{\Omega}, \quad \m{z}= \m{A}\sbra{\m{f}}\m{s}, \label{eq:modelSMV}}
where $\m{z}$ denotes the noiseless signal. Note that the ULA is a special case with $\Omega = \lbra{1,\dots,N}$. For deterministic sparse methods, in general, we need to solve a constrained optimization problem of the following form:
\equ{\min_{\m{z}} \cM\sbra{\m{z}}, \st \twon{\m{z}_{\Omega} - \m{y}_{\Omega}}\leq \eta, \label{eq:sparseoptinz}}
where the noise is assumed to be bounded: $\twon{\m{e}_{\Omega}}\leq \eta$. In \eqref{eq:sparseoptinz}, $\m{z}$ is the sinusoidal signal of interest, and $\cM\sbra{\m{z}}$ denotes a sparse metric that is defined such that by minimizing $\cM\sbra{\m{z}}$ the number of components/atoms $\m{a}\sbra{f}$ used to express $\m{z}$ is reduced, and these atoms give the frequency estimates.
Instead of \eqref{eq:sparseoptinz}, we may solve a regularized optimization problem given by:
\equ{\min_{\m{z}} \lambda\cM\sbra{\m{z}} + \frac{1}{2}\twon{\m{z}_{\Omega} - \m{y}_{\Omega}}^2, \label{eq:sparseoptinz2}}
where the noise is typically assumed to be Gaussian and $\lambda>0$ is a regularization parameter. In the extreme noiseless case, by letting $\eta\rightarrow0$ and $\lambda\rightarrow0$, both \eqref{eq:sparseoptinz} and \eqref{eq:sparseoptinz2} reduce to the following problem:
\equ{\min_{\m{z}} \cM\sbra{\m{z}}, \st \m{z}_{\Omega} = \m{y}_{\Omega}. \label{eq:sparseoptinz3}}
We next discuss different choices of $\cM\sbra{\m{z}}$.

\subsubsection{Atomic $\ell_0$ Norm}

To promote sparsity to the greatest extent possible, inspired by the literature on sparse recovery and compressed sensing, the natural choice of $\cM\sbra{\m{z}}$ is an $\ell_0$ norm like sparse metric, referred to as the atomic $\ell_0$ (pseudo-)norm. Let us formally define the set of atoms used here:
\equ{\cA= \lbra{\m{a}\sbra{f,\phi}=\m{a}\sbra{f}\phi: \; f\in\bT, \phi\in\bC, \abs{\phi}=1}.}
It is evident from \eqref{eq:modelSMV} that the true signal $\m{z}$ is a linear combination of $K$ atoms in the atomic set $\cA$. The atomic $\ell_0$ (pseudo-)norm, denoted by $\norm{\m{z}}_{\cA,0}$, is defined as the minimum number of atoms in $\cA$ that can synthesize $\m{z}$:
\equ{\begin{split}\norm{\m{z}}_{\cA,0}=
&\inf_{c_k, f_k, \phi_k} \lbra{\cK:\; \m{z} = \sum_{k=1}^{\cK} \m{a}\sbra{f_k,\phi_k}c_k, f_k\in\bT, \abs{\phi_k}=1, c_k>0 } \\
=& \inf_{f_k, s_k} \lbra{\cK:\; \m{z} = \sum_{k=1}^{\cK} \m{a}\sbra{f_k}s_k, f_k\in\bT}. \end{split} \label{eq:atom0n}}

To provide a finite-dimensional formulation for $\norm{\m{z}}_{\cA,0}$, the Vandermonde decomposition of Toeplitz covariance matrices is invoked. To be specific, let $\m{T}\sbra{\m{u}}$ be a Toeplitz matrix and impose the condition that
\equ{\begin{bmatrix} x & \m{z}^H \\ \m{z} & \m{T}\sbra{\m{u}} \end{bmatrix}\geq \m{0}, \label{eq:psdcond}}
where $x$ is a free variable to be optimized. It follows from \eqref{eq:psdcond} that $\m{T}\sbra{\m{u}}$ is PSD and thus admits a $\rank\sbra{\m{T}\sbra{\m{u}}}$-atomic Vandermonde decomposition. Moreover, $\m{z}$ lies in the range space of $\m{T}\sbra{\m{u}}$. Therefore, $\m{z}$ can be expressed by $\rank\sbra{\m{T}\sbra{\m{u}}}$ atoms. This means that the atomic $\ell_0$ norm is linked to the rank of $\m{T}\sbra{\m{u}}$. Formally, we have the following result.

\begin{thm}[\cite{tang2012compressed}] $\norm{\m{z}}_{\cA,0}$ defined in \eqref{eq:atom0n} equals the optimal value of the following rank minimization problem:
\equ{\min_{x,\m{u}} \rank\sbra{\m{T}\sbra{\m{u}}}, \st \eqref{eq:psdcond}. \label{formu:AL0_rankmin}}
\label{thm:AL0_rankmin}
\end{thm}

By Theorem \ref{thm:AL0_rankmin}, the atomic $\ell_0$ norm method needs to solve a rank minimization problem that, as might have been expected, cannot be easily solved. By the rank minimization formulation, the frequencies of interest are actually encoded in the PSD Toeplitz matrix $\m{T}\sbra{\m{u}}$. If $\m{T}\sbra{\m{u}}$ can be solved for, then the frequencies can be retrieved from its Vandermonde decomposition. Therefore, the Toeplitz matrix $\m{T}\sbra{\m{u}}$ in \eqref{formu:AL0_rankmin} can be viewed as the covariance matrix of the noiseless signal $\m{z}$ as if certain statistical assumptions were satisfied (however, those assumptions are not required here). Note that the Toeplitz structure of the covariance matrix is explicitly enforced, the PSDness is imposed by the constraint in \eqref{eq:psdcond}, and the low-rankness is the objective.

\subsubsection{Atomic Norm} \label{sec:AN_S}
A practical choice of the sparse metric $\cM\sbra{\m{z}}$ is the atomic norm that is a convex relaxation of the atomic $\ell_0$ norm. The resulting optimization problems in \eqref{eq:sparseoptinz}-\eqref{eq:sparseoptinz3} are referred to as atomic norm minimization (ANM). The concept of atomic norm was first proposed in \cite{chandrasekaran2012convex} and it generalizes several norms commonly used for sparse representation and recovery, e.g., the $\ell_1$ norm and the nuclear norm, for appropriately chosen atoms. The atomic norm is basically equivalent to the total variation norm \cite{rudin1987real} that was adopted, e.g., in \cite{candes2013towards}. We have decided to use the atomic norm in this article since it is simpler to present and easier to understand. Formally, the atomic norm is defined as the gauge function of $\text{conv}\sbra{\cA}$, the convex hull of $\cA$ \cite{chandrasekaran2012convex}:
\equ{\begin{split}\atomn{\m{z}}=
& \inf\lbra{t>0:\; \m{z}\in t \text{conv}\sbra{\cA}}\\
=&\inf_{c_k,f_k, \phi_{k}} \lbra{\sum_{k} c_k:\; \m{z} = \sum_k \m{a}\sbra{f_k,\phi_k}c_{k}, f_k\in\bT,\abs{\phi_k}=1, c_k>0 } \\
=& \inf_{f_k, s_{k}} \lbra{\sum_{k} \abs{s_{k}}:\; \m{z} = \sum_k \m{a}\sbra{f_k}s_{k}, f_k\in\bT }. \end{split} \label{eq:atomn}}
By definition, the atomic norm can be viewed as a continuous counterpart of the $\ell_1$ norm used in the discrete setting. Different from the $\ell_1$ norm, however, it is unclear how to compute the atomic norm from the definition. In fact, initially this has been a major obstacle in applying the atomic norm technique \cite{chandrasekaran2012convex,bhaskar2011atomic}. To solve this problem, a computationally efficient SDP formulation of $\atomn{\m{z}}$ is provided in the following result. A proof of the result is also provided, which helps illustrate how the frequencies can be obtained.

\begin{thm}[\cite{tang2012compressed}] $\norm{\m{z}}_{\cA}$ defined in \eqref{eq:atomn} equals the optimal value of the following SDP:
\equ{\min_{x,\m{u}} \frac{1}{2} x + \frac{1}{2} u_1, \st \eqref{eq:psdcond}. \label{formu:AN_SDP}} \label{thm:AN_SDP}
\end{thm}

\begin{proof} Let $F$ be the optimal value of the objective in \eqref{formu:AN_SDP}. We need to show that $F =\norm{\m{z}}_{\cA}$.

We first show that $F \leq \norm{\m{z}}_{\cA}$. To do so, let $\m{z} = \sum_k c_k\m{a}\sbra{f_k,\phi_k} = \sum_k \m{a}\sbra{f_k}s_k$ be an atomic decomposition of $\m{z}$. Then let $\m{u}$ be such that $\m{T}(\m{u}) = \sum_k c_k \m{a}\sbra{f_k}\m{a}^H\sbra{f_k}$ and $x = \sum_k c_k$. It follows that
\equ{\begin{bmatrix} x & \m{z}^H \\ \m{z} & \m{T} \end{bmatrix} = \sum_k c_k\begin{bmatrix} \phi_k^* \\ \m{a}\sbra{f_k} \end{bmatrix} \begin{bmatrix} \phi_k^* \\ \m{a}\sbra{f_k} \end{bmatrix}^H \geq \m{0}.}
Therefore, $x$ and $\m{u}$ constructed as above form a feasible solution to the problem in \eqref{formu:AN_SDP}, at which the objective value equals
\equ{\frac{1}{2}x + \frac{1}{2}u_1 = \sum_k c_k.}
It follows that $F \leq \sum_k c_k$. Since the inequality holds for any atomic decomposition of $\m{z}$, we have that $F \leq \norm{\m{z}}_{\cA}$ by the definition of the atomic norm.

On the other hand, suppose that $\sbra{\widehat{x}, \widehat{\m{u}}}$ is an optimal solution to the problem in \eqref{formu:AN_SDP}. By the fact that $\m{T}\sbra{\widehat{\m{u}}}\geq \m{0}$ and applying Theorem \ref{thm:VD}, we have that $\m{T}\sbra{\widehat{\m{u}}}$ admits a Vandermonde decomposition as in \eqref{eq:VD} with $\sbra{r, p_k, f_k}$ denoted by $\sbra{\widehat{r}, \widehat{p}_k, \widehat{f}_k}$.
Moreover, since $\begin{bmatrix} \widehat{x} & \m{z}^H \\ \m{z} & \m{T}\sbra{\widehat{\m{u}}} \end{bmatrix}\geq\m{0}$, we have that $\m{z}$ lies in the range space of $\m{T}\sbra{\widehat{\m{u}}}$ and thus has the following atomic decomposition:
\equ{\m{z} = \sum_{k=1}^{\widehat{r}} \widehat{c}_k\m{a}\sbra{\widehat{f}_k,\widehat{\phi}_k} = \sum_{k=1}^{\widehat{r}}\m{a}\sbra{\widehat{f}_k}\widehat{s}_k. \label{eq:yatomdec}}
Moreover, it holds that
{\lentwo\equa{\widehat{x}
&\geq& \m{z}^H \mbra{\m{T}\sbra{\widehat{\m{u}}}}^{\dag}\m{z} = \sum_{k=1}^{\widehat{r}} \frac{\widehat{c}_k^{2}}{\widehat{p}_k},\\ \widehat{u}_0
&=& \sum_{k=1}^{\widehat{r}} \widehat{p}_k.
}}It therefore follows that
\equ{\begin{split}F
&= \frac{1}{2}\widehat{x} + \frac{1}{2}\widehat{u}_1 \\
&\geq \frac{1}{2}\sum_k \frac{\widehat{c}_k^{2}}{\widehat{p}_k} + \frac{1}{2}\sum_k \widehat{p}_k \\
&\geq \sum_{k} \widehat{c}_k\\
&\geq \norm{\m{z}}_{\cA}. \end{split} \label{eq:Fleqsum} }
Combining \eqref{eq:Fleqsum} and the inequality $F \leq \norm{\m{z}}_{\cA}$ that was shown previously, we conclude that $F = \norm{\m{z}}_{\cA}$, which completes the proof. It is worth noting that by \eqref{eq:Fleqsum} we must have that $\widehat{p}_k = \widehat{c}_k=\abs{\widehat{s}_k}$ and $\norm{\m{z}}_{\cA} = \sum_{k} \widehat{c}_k = \sum_{k} \abs{\widehat{s}_k}$. Therefore, the atomic decomposition in \eqref{eq:yatomdec} achieves the atomic norm.
\end{proof}

Interestingly (but not surprisingly), the SDP in \eqref{formu:AN_SDP} is actually a convex relaxation of the rank minimization problem in \eqref{formu:AL0_rankmin}. Concretely, the second term $\frac{1}{2}u_1$ in the objective function in \eqref{formu:AN_SDP} is actually the nuclear norm or the trace norm of $\m{T}\sbra{\m{u}}$ (up to a scaling factor), which is a commonly used convex relaxation of the matrix rank, while the first term $\frac{1}{2}x$ is a regularization term that prevents a trivial solution.

Similar to the atomic $\ell_0$ norm, the frequencies in the atomic norm approach are also encoded in the Toeplitz matrix $\m{T}\sbra{\m{u}}$. Once the resulting SDP problem is solved, the frequencies can be retrieved from the Vandermonde decomposition of $\m{T}\sbra{\m{u}}$. Therefore, similar to the $\ell_0$ norm, the atomic norm can also be viewed as being covariance-based. The only difference lies in enforcing the low-rankness of the `covariance' matrix $\m{T}\sbra{\m{u}}$. The atomic $\ell_0$ norm directly uses the rank function that exploits the low-rankness to the greatest extent possible but cannot be practically solved. In contrast to this, the atomic norm uses a convex relaxation, the nuclear norm (or the trace norm), and provides a practically feasible approach.

In the absence of noise, the theoretical performance of ANM has been studied in \cite{candes2013towards,tang2012compressed}. In the case of ULA where all the entries of $\m{y}$ are observed, the ANM problem derived from \eqref{eq:sparseoptinz3} actually admits a trivial solution $\m{z}=\m{y}$. But the following SDP resulting from \eqref{formu:AN_SDP} still makes sense and can be used for frequency estimation:
\equ{\min_{x,\m{u}} \frac{1}{2} x + \frac{1}{2} u_1, \st \begin{bmatrix} x & \m{y}^H \\ \m{y} & \m{T}\sbra{\m{u}} \end{bmatrix}\geq \m{0}. \label{formu:AN_SDP2}}
Let $\cT= \lbra{f_1,\dots,f_K}$ and define the minimum separation of $\cT$ as the closest wrap-around distance between any two elements:
\equ{\Delta_{\cT} = \inf_{1\leq j\neq k \leq K} \min\lbra{\abs{f_j-f_k}, 1-\abs{f_j-f_k}}.}
The following theoretical guarantee for the atomic norm is provided in \cite{candes2013towards}.

\begin{thm}[\cite{candes2013towards}] $\m{y}=\sum_{j=1}^K c_j\m{a}\sbra{f_j,\phi_j}$ is the unique atomic decomposition satisfying $\norm{\m{y}}_{\cA}=\sum_{j=1}^K c_j$ if $\Delta_{\cT}\geq \frac{1}{\lfloor(N-1)/4\rfloor}$ and $N\geq257$. \label{thm:completedata}
\end{thm}

By Theorem \ref{thm:completedata}, in the noiseless case the frequencies can be exactly recovered by solving the SDP in \eqref{formu:AN_SDP2} if the frequencies are separated by at least $\frac{4}{N}$ (note that this frequency separation condition is sufficient but not necessary and it has recently been relaxed to $\frac{2.52}{N}$ in \cite{fernandez2016super}). Moreover, the condition $N\geq257$ is a technical requirement that should not pose any serious problem in practice.

In the SLA case, the SDP resulting from \eqref{eq:sparseoptinz3} is given by:
\equ{\min_{x,\m{u},\m{z}} \frac{1}{2} x + \frac{1}{2} u_1, \st \begin{bmatrix} x & \m{z}^H \\ \m{z} & \m{T}\sbra{\m{u}} \end{bmatrix}\geq \m{0}, \m{z}_{\Omega} = \m{y}_{\Omega}. \label{formu:AN_SDP3}}
The following result shows that the frequencies can be exactly recovered by solving \eqref{formu:AN_SDP3} if sufficiently many samples are observed and the same frequency separation condition as above is satisfied.

\begin{thm}[\cite{tang2012compressed}] Suppose we observe $\m{y}=\sum_{j=1}^K c_j\m{a}\sbra{f_j,\phi_j}$
on the index set $\Omega$, where $\Omega\subset\lbra{1,\dots,N}$ is of size $M$ and is selected uniformly at random. Assume that $\lbra{\phi_j}_{j=1}^K$ are drawn i.i.d. from
the uniform distribution on the complex unit circle.\footnote{This condition has been relaxed in \cite{yang2016exact}, where it was assumed that $\lbra{\phi_j}_{j=1}^K$ are independent with zero mean.} If $\Delta_{\cT}\geq \frac{1}{\lfloor(N-1)/4\rfloor}$, then there exists a numerical constant $C$ such that
\equ{M\geq C\max\lbra{\log^2\frac{N}{\delta}, K\log\frac{K}{\delta}\log\frac{N}{\delta}} \label{formu:AN_bound}}
is sufficient to guarantee that, with probability at least $1-\delta$, $\m{y}$ is the unique optimizer for \eqref{formu:AN_SDP3} and $\m{y}=\sum_{j=1}^K c_j\m{a}\sbra{f_j,\phi_j}$ is the unique atomic decomposition satisfying $\norm{\m{y}}_{\cA}=\sum_{j=1}^K c_j$. \label{thm:incompletedata}
\end{thm}

In the presence of noise, the SDP resulting from the unconstrained formulation in \eqref{eq:sparseoptinz2} is given by:
\equ{\min_{x,\m{u},\m{z}} \frac{\lambda}{2} \sbra{x + u_1} + \frac{1}{2}\twon{\m{z}_{\Omega} - \m{y}_{\Omega}}^2, \st \begin{bmatrix} x & \m{z}^H \\ \m{z} & \m{T}\sbra{\m{u}} \end{bmatrix}\geq \m{0}. \label{formu:AN_SDP4}}
While it is clear that the regularization parameter $\lambda$ is used to balance the signal sparsity and the data fidelity, it is less clear how to choose it. Under the assumption of i.i.d. Gaussian noise, this choice has been studied in \cite{bhaskar2013atomic,tang2015near,yang2015gridless}. For ULAs the paper \cite{bhaskar2013atomic} shows that if we let $\lambda \approx \sqrt{M\log M\sigma}$, where $\sigma$ denotes the noise variance, then the signal estimate $\widehat{\m{z}}$ given by \eqref{formu:AN_SDP4} has the following per-element expected reconstruction error:
\equ{\frac{1}{M} \bE \twon{\widehat{\m{z}} - \m{z}}^2 \leq \sqrt{\frac{\log M}{M}\sigma} \cdot\sum_{k=1}^K c_k. \label{eq:zboundslow}}
This error bound implies that if $K=o\sbra{\sqrt{\frac{M}{\log M}}}$, then the estimate $\widehat{\m{z}}$ is statistically consistent. Moreover, the paper \cite{tang2015near} shows that if we let $\lambda = C\sqrt{M\log M\sigma}$, where $C>1$ is a constant (not explicitly given), and if the frequencies are sufficiently separated as in Theorem \ref{thm:completedata}, then the following error bound can be obtained with high probability:
\equ{\frac{1}{M} \twon{\widehat{\m{z}} - \m{z}}^2 = O\sbra{\frac{K\log M}{M}\sigma},}
which is nearly minimax optimal. This implies that the estimate is consistent if $K=o\sbra{\frac{M}{\log M}}$. Furthermore, the frequencies and the amplitudes can be stably estimated as well \cite{tang2015near}. Finally, note that the result in \cite{bhaskar2013atomic} has been generalized to the SLA case in \cite{yang2015gridless}. It was shown that if we let $\lambda \approx \sqrt{M\log N\sigma}$ in such a case, it holds similarly to \eqref{eq:zboundslow} that
\equ{\frac{1}{M} \bE \twon{\widehat{\m{z}}_{\Omega} - \m{z}_{\Omega}}^2 \leq \sqrt{\frac{\log N}{M}\sigma}\cdot \sum_{k=1}^K c_k.}
This means that the estimate $\widehat{\m{z}}_{\Omega}$ is consistent if $K=o\sbra{\sqrt{\frac{M}{\log N}}}$.

\subsubsection{Hankel-based Nuclear Norm}
Another choice of $\cM\sbra{\m{z}}$ is the Hankel-based nuclear norm that was proposed in \cite{chen2014robust}. This metric is introduced based on the following observation. Given $\m{z}$ as in \eqref{eq:modelSMV}, let us form the Hankel matrix: \equ{\m{H}\sbra{\m{z}} = \begin{bmatrix} z_1 & z_2 & \dots & z_n \\ z_2 & z_3 & \dots & z_{n+1} \\ \vdots & \vdots & \ddots & \vdots \\ z_m & z_{m+1} & \dots & z_N \end{bmatrix}, \label{eq:hankel}}
where $m+n = N+1$. It follows that
\equ{\m{H}\sbra{\m{z}} = \sum_{k=1}^K s_k \begin{bmatrix} 1 \\ e^{i2\pi f_k} \\ \vdots \\ e^{i2\pi (m-1)f_k} \end{bmatrix} \begin{bmatrix} 1 & e^{i2\pi f_k} & \dots & e^{i2\pi (n-1)f_k} \end{bmatrix}. \label{eq:Hankdec}}
If $K<\min\sbra{m,n}$, then we have that $\m{H}\sbra{\m{z}}$ is a low rank matrix with
\equ{\rank\sbra{\m{H}\sbra{\m{z}}} = K. }
To reconstruct $\m{z}$, therefore, we may consider the reconstruction of $\m{H}\sbra{\m{z}}$ by choosing the sparse metric as $\rank\sbra{\m{H}\sbra{\m{z}}}$. If $\m{z}$ can be determined for the resulting rank minimization problem, then the frequencies may be recovered from $\m{z}$.

Since the rank minimization cannot be easily solved, we seek a convex relaxation of $\rank\sbra{\m{H}\sbra{\m{z}}}$. The nuclear norm is a natural choice, which leads to:
\equ{\cM\sbra{\m{z}} = \norm{\m{H}\sbra{\m{z}}}_{\star}. \label{eq:MisnnHz}}
The optimization problems resulting from \eqref{eq:sparseoptinz}-\eqref{eq:sparseoptinz3} by using \eqref{eq:MisnnHz} are referred to as enhanced matrix completion (EMaC) in \cite{chen2014robust}. Note that the nuclear norm can be formulated as the following SDP \cite{fazel2001rank}:
\equ{\norm{\m{H}\sbra{\m{z}}}_{\star} = \min_{\m{Q}_1,\m{Q}_2}\frac{1}{2}\mbra{\tr\sbra{\m{Q}_1} + \tr\sbra{\m{Q}_2}}, \st \begin{bmatrix} \m{Q}_1 & \m{H}\sbra{\m{z}}^H \\ \m{H}\sbra{\m{z}} & \m{Q}_2 \end{bmatrix} \geq \m{0}. \label{eq:Hnnsdp}}
As a result, like the atomic norm method, the EMaC problems can be cast as SDP and solved using off-the-shelf solvers.

Theoretical guarantees for EMaC have been provided in the SLA case in \cite{chen2014robust} which, to some extent, are similar to those for the atomic norm method. In particular, it was shown that the signal $\m{z}$ can be exactly recovered in the absence of noise and stably recovered in the presence of bounded noise if the number of measurements $M$ exceeds a constant times the number of sinusoids $K$ up to a polylog factor, as given by \eqref{formu:AN_bound}, and if a certain coherence condition is satisfied. It is argued in \cite{chen2014robust} that the coherence condition required by EMaC can be weaker than the frequency separation condition required by ANM and thus higher resolution might be obtained by EMaC as compared to ANM. Connections between the two methods will be studied in the following subsection.

\subsubsection{Connection between ANM and EMaC}
To investigate the connection between ANM and EMaC, we define the following set of complex exponentials as a new atomic set:
\equ{\cA'= \lbra{\m{a}'\sbra{\phi} = \mbra{1,\phi,\dots,\phi^{N-1}}^T:\; \phi\in\bC}.}
It is evident that, as compared to $\cA'$, the complex exponentials are restricted to have constant modulus in $\cA$ with $\abs{\phi}=1$ and thus $\cA$ is a subset of $\cA'$. For any $\m{z}\in\bC^N$, we can similarly define the atomic $\ell_0$ norm with respect to $\cA'$, denoted by $\norm{\m{z}}_{\cA',0}$. We have the following result.
\begin{thm} For any $\m{z}$ it holds that
\equ{\norm{\m{z}}_{\cA',0} \geq \rank\sbra{\m{H}\sbra{\m{z}}}. \label{eq:ineqatom0zH}}
Moreover, if $\rank\sbra{\m{H}\sbra{\m{z}}}<\min\sbra{m,n}$, then
\equ{\norm{\m{z}}_{\cA',0} = \rank\sbra{\m{H}\sbra{\m{z}}} \label{eq:eqatom0zH}}
except for degenerate cases. \label{thm:atom0nH}
\end{thm}
\begin{proof} Suppose that $\norm{\m{z}}_{\cA',0} = K'$. This means that there exists a $K'$-atomic decomposition for $\m{z}$ with respect to $\cA'$:
\equ{\m{z} = \sum_{k=1}^{K'} \m{a}\sbra{\phi_k} s'_k, \quad \phi_k\in\bC. \label{eq:zatomdecH}}
It follows that $\m{H}\sbra{\m{z}}$ admits a decomposition similar to \eqref{eq:Hankdec} and therefore that $\rank\sbra{\m{H}\sbra{\m{z}}}\leq K'$ and hence \eqref{eq:ineqatom0zH} holds.

The second part can be shown by applying the Kronecker's theorem for Hankel matrices (see, e.g., \cite{rochberg1987toeplitz}). In particular, the Kronecker's theorem states that if $\rank\sbra{\m{H}\sbra{\m{z}}} = K' < \min\sbra{m,n}$, then $\m{z}$ can be written as in \eqref{eq:zatomdecH} except for degenerate cases. According to the definition of $\norm{\m{z}}_{\cA',0}$, we have that
\equ{\norm{\m{z}}_{\cA',0} \leq K' = \rank\sbra{\m{H}\sbra{\m{z}}}}
This together with \eqref{eq:ineqatom0zH} concludes the proof of \eqref{eq:eqatom0zH}.
\end{proof}

By Theorem \ref{thm:atom0nH}, we have linked $\rank\sbra{\m{H}\sbra{\m{z}}}$, which motivated the use of its convex relaxation $\norm{\m{H}\sbra{\m{z}}}_{\star}$, to an atomic $\ell_0$ norm. In the regime of interest here $\m{H}\sbra{\m{z}}$ is low-rank and hence $\rank\sbra{\m{H}\sbra{\m{z}}}$ is almost identical to the atomic $\ell_0$ norm induced by $\cA'$. To compare $\norm{\m{z}}_{\cA,0}$ and  $\norm{\m{z}}_{\cA',0}$, we have the following result.

\begin{thm} For any $\m{z}$ it holds that
\equ{\norm{\m{z}}_{\cA',0} \leq\norm{\m{z}}_{\cA,0}.} \label{thm:comp2atom0n}
\end{thm}
\begin{proof} The inequality is a direct consequence of the fact that $\cA\subset \cA'$.
\end{proof}

By Theorem \ref{thm:comp2atom0n} the newly defined $\norm{\m{z}}_{\cA',0}$, which is closely related to $\rank\sbra{\m{H}\sbra{\m{z}}}$, is actually a lower bound on $\norm{\m{z}}_{\cA,0}$. It is worth noting that this lower bound is obtained by ignoring a known structure of the signal: from $\cA$ to $\cA'$ we have neglected the prior knowledge that each exponential component of $\m{z}$ has constant modulus. As a consequence, using $\norm{\m{z}}_{\cA',0}$ as the sparse metric instead of $\norm{\m{z}}_{\cA,0}$, we cannot guarantee in general, especially in the noisy case, that each component of the obtained signal $\m{z}$ corresponds to one frequency. Note that this is also true for the convex relaxation metric $\norm{\m{H}\sbra{\m{z}}}_{\star}$ as compared to $\norm{\m{z}}_{\cA}$. In contrast to this, the frequencies can be directly retrieved from the solution of the atomic norm method. From this point of view, the atomic norm method may be expected to outperform EMaC due to its better capability to capture the signal structure.

On the other hand, EMaC might have higher resolution than the atomic norm method. This is indeed true in the noiseless ULA case where the signal $\m{z}$ is completely known. In this extreme case EMaC does not suffer from any resolution limit, while the atomic norm requires a frequency separation condition for its successful operation (at least theoretically).

\subsubsection{Covariance Fitting Method: Gridless SPICE (GLS)} \label{sec:GLS_SMV}

GLS was introduced in \cite{yang2014discretization,yang2015gridless} as a gridless version of the SPICE method presented in Subsection \ref{sec:SPICE}. Since SPICE is covariance-based and the data covariance matrix is a highly nonlinear function of the DOA parameters of interest, gridding is performed in SPICE to linearize the problem based on the zeroth order approximation. But this is not required in the case of ULAs or SLAs. The key idea of GLS is to re-parameterize the data covariance matrix using a PSD Toeplitz matrix $\m{T}\sbra{\m{u}}$, which is linear in the new parameter vector $\m{u}$, by making use of the Vandermonde decomposition of Toeplitz covariance matrices. To derive GLS, naturally, we make the same assumptions as for SPICE.

We first consider the ULA case. Assume that the noise is homoscedastic (note that, like SPICE, GLS can be extended to the case of heteroscedastic noise). It follows from the arguments in Subsection \ref{sec:VD} that the data covariance matrix $\m{R}$ is a Toeplitz matrix. Therefore, $\m{R}$ can be linearly re-parameterized as:
\equ{\m{R} = \m{T}\sbra{\m{u}}, \quad \m{T}\sbra{\m{u}}\geq \m{0}. \label{eq:R_ULA_equalsigma}}
For a single snapshot, SPICE minimizes the following covariance fitting criterion:
\equ{\frobn{\m{R}^{-\frac{1}{2}}\sbra{\m{y}\m{y}^H-\m{R}}}^2 = \twon{\m{y}}^2\cdot\m{y}^H\m{R}^{-1}\m{y}+\tr\sbra{\m{R}}-2\twon{\m{y}}^2. \label{formu:criterion_SMV}}
Inserting \eqref{eq:R_ULA_equalsigma} into \eqref{formu:criterion_SMV}, the resulting GLS optimization problem is given by:
\equ{\begin{split}
&\min_{\m{u}} \twon{\m{y}}^2\cdot\m{y}^H\m{T}^{-1}\sbra{\m{u}}\m{y}+ \tr\sbra{\m{T}\sbra{\m{u}}}, \st \m{T}\sbra{\m{u}}\geq\m{0}\\
\Leftrightarrow& \min_{x,\m{u}} \twon{\m{y}}^2 x + Mu_1, \st \m{T}\sbra{\m{u}}\geq\m{0} \text{ and } x\geq \m{y}^H\m{T}^{-1}\sbra{\m{u}}\m{y}\\
\Leftrightarrow& \min_{x,\m{u}} \twon{\m{y}}^2 x + Mu_1, \st \begin{bmatrix}x& \m{y}^H \\ \m{y} & \m{T}\sbra{\m{u}} \end{bmatrix}\geq\m{0}.\end{split} \label{formu:SDP_identical_SMV}}
Therefore, the covariance fitting problem has been cast as an SDP that can be solved in a polynomial time. Once the problem is solved, the data covariance estimate $\widehat{\m{R}}=\m{T}\sbra{\widehat{\m{u}}}$ is obtained, where $\widehat{\m{u}}$ denotes the solution of $\m{u}$. Finally, the estimates of the parameters $\sbra{\widehat{\m{f}}, \widehat{\m{p}}, \widehat{\sigma}}$ can be obtained from the decomposition of $\widehat{\m{R}}$ by applying Corollary \ref{cor:Tinsignalnoise}. Moreover, we note that the GLS optimization problem in \eqref{formu:SDP_identical_SMV} is very similar to the SDP for the atomic norm. Their connection will be discussed later in more detail.

In the SLA case, corresponding to \eqref{eq:R_ULA_equalsigma}, the covariance fitting criterion of SPICE is given by:
\equ{\frobn{\m{R}_{\Omega}^{-\frac{1}{2}}\sbra{\m{y}_{\Omega}\m{y}_{\Omega}^H- \m{R}_{\Omega}}}^2 = \twon{\m{y}_{\Omega}}^2\cdot\m{y}_{\Omega}^H \m{R}_{\Omega}^{-1}\m{y}_{\Omega}+\tr\sbra{\m{R}_{\Omega}} -2\twon{\m{y}_{\Omega}}^2, \label{formu:criterion_SMV_SLA}}
where $\m{R}_{\Omega}$ denotes the covariance matrix of $\m{y}_{\Omega}$. To explicitly express $\m{R}_{\Omega}$, we let $\m{\Gamma}_{\Omega}\in\lbra{0,1}^{M\times N}$ be the row-selection matrix satisfying
\equ{\m{y}_{\Omega} = \m{\Gamma}_{\Omega} \m{y},}
where $\m{y}\in\bC^N$ denotes the full data vector of the virtual $N$-element ULA. More concretely, $\m{\Gamma}_{\Omega}$ is such that its $j$th row contains all 0s but a single 1 at the $\Omega_j$th position. It follows that
\equ{\m{R}_{\Omega} = \bE\m{y}_{\Omega}\m{y}_{\Omega}^H = \m{\Gamma}_{\Omega} \cdot \bE\m{y}\m{y}^H \cdot \m{\Gamma}_{\Omega}^T = \m{\Gamma}_{\Omega} \m{R} \m{\Gamma}_{\Omega}^T.}
This means that $\m{R}_{\Omega}$ is a submatrix of the covariance matrix $\m{R}$ of the virtual full data $\m{y}$. Therefore, using the parameterization of $\m{R}$ as in \eqref{eq:R_ULA_equalsigma}, $\m{R}_{\Omega}$ can be linearly parameterized as:
\equ{\m{R}_{\Omega} = \m{\Gamma}_{\Omega} \m{T}\sbra{\m{u}} \m{\Gamma}_{\Omega}^T, \quad \m{T}\sbra{\m{u}}\geq \m{0}.\label{eq:R_SLA1}}
Inserting \eqref{eq:R_SLA1} into \eqref{formu:criterion_SMV_SLA}, we have that the GLS optimization problem resulting from \eqref{formu:criterion_SMV_SLA} can be cast as the following SDP:
\equ{\min_{x,\m{u}} \twon{\m{y}_{\Omega}}^2 x + Mu_1, \st \begin{bmatrix}x& \m{y}_{\Omega}^H \\ \m{y}_{\Omega} & \m{\Gamma}_{\Omega} \m{T}\sbra{\m{u}} \m{\Gamma}_{\Omega}^T \end{bmatrix}\geq\m{0}. \label{formu:SDP_identical_SMV_SLA}}
Once the SDP is solved, the parameters of interest can be retrieved from the full data covariance estimate $\widehat{\m{R}}=\m{T}\sbra{\widehat{\m{u}}}$ as in the ULA case.

GLS is guaranteed to produce a sparse solution with at most $N-1$ sources. This is a direct consequence of the frequency retrieval step, see Corollary \ref{cor:Tinsignalnoise}. In general, GLS overestimates the true source number $K$ in the presence of noise. This is reasonable because in GLS we do not assume any knowledge of the source number or of the noise variance(s).

An automatic source number estimation approach (a.k.a.~model order selection) has been proposed in \cite{yang2015gridless} based on the eigenvalues of the data covariance estimate $\widehat{\m{R}}$. The basic intuition behind it is that the larger eigenvalues correspond to the sources while the smaller ones are caused more likely by noise. The SORTE algorithm \cite{he2010detecting} is adopted in \cite{yang2015gridless} to identify these two groups of eigenvalues and thus provide an estimate of the source number. Furthermore, based on the source number, the frequency estimates can be refined by using a subspace method such as MUSIC. Readers are referred to \cite{yang2015gridless} for details.

\subsubsection{Connection between ANM and GLS} \label{sec:ANM_GLS_SS}
GLS is strongly connected to ANM. In this subsection, we show that GLS is equivalent to ANM as if there were no noise in $\m{y}$ and with a slightly different frequency retrieval process \cite{yang2015gridless}. We note that a similar connection in the discrete setting has been provided in \cite{rojas2013note,babu2014connection}.

In the ULA case this connection can be shown based on the following equivalences/equalities:
\equ{\begin{split} \eqref{formu:SDP_identical_SMV}
&\Leftrightarrow \min_{\m{u}} \twon{\m{y}}^2\cdot\m{y}^H\m{T}^{-1}\sbra{\m{u}}\m{y}+ Mu_1, \st \m{T}\sbra{\m{u}}\geq\m{0} \\
&\Leftrightarrow \min_{\m{u}} M\lbra{\mbra{\frac{\twon{\m{y}}}{\sqrt{M}}\m{y}^H}\m{T}^{-1}\sbra{\m{u}} \mbra{\frac{\twon{\m{y}}}{\sqrt{M}}\m{y}}+ u_1 }, \st \m{T}\sbra{\m{u}}\geq\m{0} \\
&\Leftrightarrow \min_{x,\m{u}} M \lbra{ x + u_1}, \st \begin{bmatrix}x& \frac{\twon{\m{y}}}{\sqrt{M}}\m{y}^H \\ \frac{\twon{\m{y}}}{\sqrt{M}}\m{y} & \m{T}\sbra{\m{u}} \end{bmatrix}\geq\m{0} \\
&\Leftrightarrow 2M \atomn{\frac{\twon{\m{y}}}{\sqrt{M}}\m{y}} \\
&\Leftrightarrow  2\sqrt{M}\twon{\m{y}} \cdot \atomn{\m{y}}.
\end{split}}
This means that GLS is equivalent to computing the atomic norm of the noisy data $\m{y}$ (up to a scaling factor).

In the SLA case, we need the following equality that holds for $\m{R}>\m{0}$ \cite{yang2015gridless}:
\equ{\m{y}_{\Omega}^H \mbra{\m{\Gamma}_{\Omega}\m{R} \m{\Gamma}^T_{\Omega}} ^{-1}\m{y}_{\Omega} = \min_{\m{z}} \m{z}^H \m{R}^{-1}\m{z}, \st \m{z}_{\Omega} = \m{y}_{\Omega}. \label{eq:augment}}
As a result, we have that
\equ{\begin{split} \eqref{formu:SDP_identical_SMV_SLA}
&\Leftrightarrow \min_{\m{u}} \twon{\m{y}_{\Omega}}^2\cdot\m{y}_{\Omega}^H \mbra{\m{\Gamma}_{\Omega}\m{T}\sbra{\m{u}} \m{\Gamma}^T_{\Omega}} ^{-1}\m{y}_{\Omega}+ Mu_1, \st \m{T}\sbra{\m{u}}\geq\m{0} \\
&\Leftrightarrow \min_{\m{u},\m{z}} \twon{\m{y}_{\Omega}}^2\cdot\m{z}^H \m{T}^{-1}\sbra{\m{u}} \m{z}+ Mu_1, \st \m{T}\sbra{\m{u}}\geq\m{0}, \m{z}_{\Omega} = \m{y}_{\Omega} \\
&\Leftrightarrow \min_{\m{u},\m{z}} M\lbra{\mbra{\frac{\twon{\m{y}_{\Omega}}}{\sqrt{M}}\m{z}^H}\m{T}^{-1}\sbra{\m{u}} \mbra{\frac{\twon{\m{y}_{\Omega}}}{\sqrt{M}}\m{z}}+ u_1 }, \st \m{T}\sbra{\m{u}}\geq\m{0}, \m{z}_{\Omega} = \m{y}_{\Omega} \\
&\Leftrightarrow \min_{x,\m{u},\m{z}} M \lbra{ x + u_1}, \st \begin{bmatrix}x& \frac{\twon{\m{y}_{\Omega}}}{\sqrt{M}}\m{z}^H \\ \frac{\twon{\m{y}_{\Omega}}}{\sqrt{M}}\m{z} & \m{T}\sbra{\m{u}} \end{bmatrix}\geq\m{0}, \m{z}_{\Omega} = \m{y}_{\Omega} \\
&\Leftrightarrow \min_{\m{z}} 2M \atomn{\frac{\twon{\m{y}_{\Omega}}}{\sqrt{M}}\m{z}}, \st \m{z}_{\Omega} = \m{y}_{\Omega} \\
&\Leftrightarrow \min_{\m{z}} 2\sqrt{M}\twon{\m{y}_{\Omega}} \cdot \atomn{\m{z}}, \st \m{z}_{\Omega} = \m{y}_{\Omega} \\
&\Leftrightarrow \min_{\m{z}} \atomn{\m{z}}, \st \m{z}_{\Omega} = \m{y}_{\Omega}.
\end{split}}
This means, like in the ULA case, that GLS is equivalent to ANM subject to the data consistency as if there were no noise.

Finally, note that GLS is practically attractive since it does not require knowledge of the noise level. Regarding this fact, note that GLS is different from ANM in the frequency retrieval process: whereas Theorem \ref{thm:VD} is used in ANM, Corollary \ref{cor:Tinsignalnoise} is employed in GLS since the noise variance is also encoded in the Toeplitz covariance matrix $\m{T}\sbra{\m{u}}$, besides the frequencies.


\subsection{The Multiple Snapshot Case: Covariance Fitting Methods}
In this and the following subsections we will study gridless DOA estimation methods for multiple snapshots. The methods that we will introduce are mainly based on or inspired by those in the single snapshot case in the preceding subsection. The key techniques of these methods exploit the temporal redundancy of the multiple snapshots for possibly improved performance. We have decided to introduce the covariance fitting methods first since they appeared earlier than their deterministic peers. In this kind of methods, differently from the deterministic ones, certain statistical assumptions on the sources (like for SPICE) are required to explicitly express the data covariance matrix. We will discuss three covariance-based gridless sparse methods: GLS in \cite{yang2014discretization}, the SMV-based atomic norm method in \cite{tan2014direction} and the low rank matrix denoising approach in \cite{pal2014grid}. While GLS is applicable to an arbitrary number of snapshots, the latter two can only be used if there are sufficiently many snapshots.

\subsubsection{Gridless SPICE (GLS)} \label{sec:GLS_MMV}
In the presence of multiple snapshots, GLS is derived in a similar way as in the single snapshot case by utilizing the convex reparameterization of the data covariance matrix $\m{R}$ in \eqref{eq:R_ULA_equalsigma}. For convenience, some derivations of SPICE provided in Subsection \ref{sec:SPICE} will be repeated here. We first consider the ULA case. Let $\widetilde{\m{R}} = \frac{1}{L}\m{Y}\m{Y}^H$ be the sample covariance. In the case of $L\geq M$ when the sample covariance $\widetilde{\m{R}}$ is nonsingular, the following SPICE covariance fitting criterion is minimized:
\equ{\frobn{\m{R}^{-\frac{1}{2}}\sbra{\widetilde{\m{R}}-\m{R}}\widetilde{\m{R}}^{-\frac{1}{2}}}^2 = \tr\sbra{\m{R}^{-1}\widetilde{\m{R}}}+ \tr\sbra{\widetilde{\m{R}}^{-1}\m{R}}-2M. \label{formu:criterion1}}
Inserting \eqref{eq:R_ULA_equalsigma} into \eqref{formu:criterion1}, we have the following equivalences:
\equ{\begin{split}&\min_{\m{u}} \tr\sbra{\widetilde{\m{R}}^{\frac{1}{2}}\m{T}^{-1}\sbra{\m{u}}\widetilde{\m{R}}^{\frac{1}{2}}}+ \tr\sbra{\widetilde{\m{R}}^{-1}\m{T}\sbra{\m{u}}}, \st \m{T}\sbra{\m{u}}\geq\m{0}\\
\Leftrightarrow& \min_{\m{X},\m{u}} \tr\sbra{\m{X}}+ \tr\sbra{\widetilde{\m{R}}^{-1}\m{T}\sbra{\m{u}}}, \st \m{T}\sbra{\m{u}}\geq\m{0} \text{ and } \m{X}\geq \widetilde{\m{R}}^{\frac{1}{2}}\m{T}^{-1}\sbra{\m{u}}\widetilde{\m{R}}^{\frac{1}{2}}\\
\Leftrightarrow& \min_{\m{X},\m{u}} \tr\sbra{\m{X}}+ \tr\sbra{\widetilde{\m{R}}^{-1}\m{T}\sbra{\m{u}}}, \st \begin{bmatrix}\m{X}& \widetilde{\m{R}}^{\frac{1}{2}} \\ \widetilde{\m{R}}^{\frac{1}{2}} & \m{T}\sbra{\m{u}} \end{bmatrix}\geq\m{0}.\end{split} \label{formu:SDP_Llarge_identical}}
Therefore, the covariance fitting problem is cast as SDP. Once the problem is solved, the estimates of the parameters $\sbra{\widehat{\m{f}}, \widehat{\m{p}}, \widehat{\sigma}}$ can be obtained from the data covariance estimate $\widehat{\m{R}}=\m{T}\sbra{\widehat{\m{u}}}$ by applying Corollary \ref{cor:Tinsignalnoise}.

In the case of $L<M$ when $\widetilde{\m{R}}$ is singular, we instead minimize the criterion \equ{\frobn{\m{R}^{-\frac{1}{2}}\sbra{\widetilde{\m{R}}-\m{R}}}^2 = \tr\sbra{\m{R}^{-1}\widetilde{\m{R}}^2}+\tr\sbra{\m{R}}-2\tr\sbra{\widetilde{\m{R}}}. \label{formu:criterion2}}
Similarly, inserting \eqref{eq:R_ULA_equalsigma} into \eqref{formu:criterion2}, we obtain the following SDP:
\equ{\min_{\m{X},\m{u}} \tr\sbra{\m{X}}+ \tr\sbra{\m{T}\sbra{\m{u}}}, \st \begin{bmatrix}\m{X}& \widetilde{\m{R}} \\ \widetilde{\m{R}} & \m{T}\sbra{\m{u}} \end{bmatrix}\geq\m{0}. \label{formu:SDP_Lsmall_identical}}
The parameter estimates $\sbra{\widehat{\m{f}}, \widehat{\m{p}}, \widehat{\sigma}}$ can be obtained in the same manner as above.

The dimensionality of the SDP in \eqref{formu:SDP_Lsmall_identical} can be reduced \cite{yang2016gridless1}. To do so, let $\widetilde{\m{Y}} = \frac{1}{L}\m{Y}\sbra{\m{Y}^H\m{Y}}^{\frac{1}{2}}\in\bC^{M\times L}$ and observe that
\equ{\widetilde{\m{R}}^2 = \frac{1}{L^2}\m{Y}\m{Y}^H\m{Y}\m{Y}^H = \widetilde{\m{Y}}\widetilde{\m{Y}}^H. \label{eq:R2inYtilde}}
Inserting \eqref{eq:R2inYtilde} into \eqref{formu:criterion2}, we obtain another SDP formulation of the covariance fitting problem:
\equ{\min_{\m{X},\m{u}} \tr\sbra{\m{X}}+ \tr\sbra{\m{T}\sbra{\m{u}}}, \st \begin{bmatrix}\m{X}& \widetilde{\m{Y}}^H \\ \widetilde{\m{Y}} & \m{T}\sbra{\m{u}} \end{bmatrix}\geq\m{0}. \label{formu:SDP_Lsmall_identical2}}
Compared to \eqref{formu:SDP_Lsmall_identical}, the dimensionality of the semidefinite matrix in \eqref{formu:SDP_Lsmall_identical2} is reduced from $2M\times 2M$ to $\sbra{M+L} \times \sbra{M+L}$ (note that in this case $L<M$). This reduction can be significant in the case of a small number of snapshots.

Similar results can be obtained in the SLA case. In particular, let $\m{R}_{\Omega}$ and $\widetilde{\m{R}}_{\Omega} = \frac{1}{L}\m{Y}_{\Omega}\m{Y}_{\Omega}^H$ denote the data covariance and the sample covariance, respectively. Using the linear reparameterization of $\m{R}_{\Omega}$ in \eqref{eq:R_SLA1}, similar SDPs to those above can be formulated. In the case of $L\geq M$ when $\widetilde{\m{R}}_{\Omega}$ is nonsingular, we have the following SDP:
\equ{\begin{split}
\min_{\m{X},\m{u}} \tr\sbra{\m{X}}+ \tr\sbra{\m{\Gamma}_{\Omega}^T\widetilde{\m{R}}_{\Omega}^{-1}\m{\Gamma}_{\Omega}\m{T}\sbra{\m{u}}}, \st \begin{bmatrix}\m{X}& \widetilde{\m{R}}_{\Omega}^{\frac{1}{2}} \\ \widetilde{\m{R}}_{\Omega}^{\frac{1}{2}} & \m{\Gamma}_{\Omega}\m{T}\sbra{\m{u}}\m{\Gamma}_{\Omega}^T \end{bmatrix}\geq\m{0}.\end{split} \label{formu:SDP_Llarge_identical3}}
When $L< M$, the SDP is
\equ{\min_{\m{X},\m{u}} \tr\sbra{\m{X}}+ \tr\sbra{\m{\Gamma}_{\Omega}^T\m{\Gamma}_{\Omega}\m{T}\sbra{\m{u}}}, \st \begin{bmatrix}\m{X}& \widetilde{\m{R}}_{\Omega} \\ \widetilde{\m{R}}_{\Omega} & \m{\Gamma}_{\Omega}\m{T}\sbra{\m{u}}\m{\Gamma}_{\Omega}^T \end{bmatrix}\geq\m{0}, \label{formu:SDP_Lsmall_identical3}}
where $\widetilde{\m{R}}_{\Omega}$ can also be replaced by $\frac{1}{L}\m{Y}_{\Omega}\sbra{\m{Y}_{\Omega}^H\m{Y}_{\Omega}}^{\frac{1}{2}}\in\bC^{M\times L}$ for dimensionality reduction.
Once the SDP is solved, the parameters of interest can be retrieved as in the ULA case.

%

As in the single snapshot case, GLS is guaranteed to produce a sparse solution with at most $N-1$ sources. Besides this, GLS has other attractive properties as detailed below.

Let us assume that the antenna array is a redundancy array or, equivalently, that the full matrix $\m{T}\sbra{\m{u}}$ can be recovered from its submatrix $\m{\Gamma}_{\Omega} \m{T}\sbra{\m{u}} \m{\Gamma}_{\Omega}^T$ (see, e.g., \cite{linebarger1993difference, viberg2014introduction}). Note that all ULAs and many SLAs are redundancy arrays. Then the GLS estimator is statistically consistent in the snapshot number $L$ if the true source number $K\leq N-1$. To see this, let us consider a ULA first. As $L\rightarrow\infty$, $\widetilde{\m{R}}$ approaches the true data covariance matrix that is denoted by $\m{R}^o$ and has the same Toeplitz structure as $\m{R}$. Hence, it follows from \eqref{formu:criterion1} that $\widehat{\m{R}} = \m{R}^o$ and the true parameters can be retrieved from $\widehat{\m{R}}$. In the SLA case, similarly, the covariance matrix estimate $\widehat{\m{R}}_{\Omega}$ converges to the true one as $L\rightarrow\infty$. Then, the assumption of redundancy array can be used to show that all the information in the Toeplitz matrix $\m{T}\sbra{\m{u}}$ for frequency retrieval can be recovered from $\widehat{\m{R}}_{\Omega}$ and hence that the true parameters can be obtained.

Furthermore, under the stronger assumption of Gaussian sources and noise, GLS is an asymptotic ML estimator when $K\leq N-1$ and a redundancy array is employed. This is true because the global solution of the SPICE covariance fitting problem is a large-snapshot realization of the ML estimator ( \cite{stoica1989reparametrization,ottersten1998covariance}) and because GLS globally solves the problem. As a direct consequence of this, GLS has improved resolution as $L$ increases and the resolution limit vanishes as $L$ grows to infinity. Another consequence is that GLS can estimate more sources than the number of antennas. In fact, a redundancy array with aperture $N$ can usually be formed by using $M \approx \sqrt{3(N-1)}$ antennas \cite{linebarger1993difference}. This means that up to about $\frac{1}{3}M^2$ sources can be estimated using GLS with only $M$ antennas.

It is worth noting that generally the above properties of GLS are not shared by its discrete version, viz.~SPICE, due to an ambiguity problem, even when the on-grid assumption of SPICE holds. To see this, let us consider the ULA case as an example and suppose that SPICE can accurately estimate the data covariance matrix $\m{R}=\m{T}\sbra{\m{u}}$, as GLS does. Note that when $\m{R}$ has full rank, which is typically the case in practice, the solution of GLS is provided by the unique decomposition of $\m{R}$ from Corollary \ref{cor:Tinsignalnoise}. But this uniqueness cannot be guaranteed in SPICE according to Remark \ref{rem:uniqueTdecomp} (note that the condition that $r<N$ of Corollary \ref{cor:Tinsignalnoise} might not be satisfied in SPICE).

\subsubsection{SMV-based Atomic Norm Minimization (ANM-SMV)}
Within the SMV super-resolution framework introduced in \cite{candes2013towards}, an ANM approach was proposed in \cite{tan2014direction}, designated here as ANM-SMV. While the paper \cite{tan2014direction} focused on co-prime arrays (\cite{vaidyanathan2011sparse}), which form a special class of SLAs, ANM-SMV can actually be applied to a broader class of SLAs such as redundancy arrays or even general SLAs. To simplify our discussions, we consider without loss of generality a redundancy array, denoted by $\Omega$.

Under the assumption of uncorrelated sources, as for GLS, the data covariance matrix $\m{R}_{\Omega}$ can be expressed as:
\equ{\m{R}_{\Omega} = \m{A}_{\Omega}\sbra{\m{f}}\diag\sbra{\m{p}}\m{A}_{\Omega}^H\sbra{\m{f}} + \sigma\m{I},}
where $p_k>0$, $k=1,\dots,K$ denote the source powers.
According to the discussions in Section \ref{sec:GLS_SMV}, $\m{R}_{\Omega}$ is a submatrix of a Toeplitz covariance matrix:
\equ{\m{R}_{\Omega} = \m{\Gamma}_{\Omega}\m{T}\sbra{\m{u}}\m{\Gamma}_{\Omega}^T + \sigma\m{I}, \label{eq:ROmegainTs}}
where $\m{u} \in\bC^N$ and $\m{T}\sbra{\m{u}} = \m{A}\sbra{\m{f}}\diag\sbra{\m{p}}\m{A}^H\sbra{\m{f}}$. Since the frequencies have been encoded in $\m{T}\sbra{\m{u}}$, ANM-SMV, like GLS, also carries out covariance fitting to estimate $\m{u}$ and thus the frequencies. But the technique used by ANM-SMV is different. To describe this technique, note that $\m{u}$ can be expressed as
\equ{\m{u} = \m{A}^*\sbra{\m{f}} \m{p}. \label{eq:uinfp}}
Let us define $\m{v}\in\bC^{2N-1}$ such that
\equ{v_j = \left\{\begin{array}{ll} u_{N-j+1}, & j=1,\dots,N, \\ u_{j-N+1}^*, & j=N+1,\dots,2N-1. \end{array} \right. \label{eq:vinu}}
Given $\m{u}$ in \eqref{eq:uinfp}, note that $\m{v}$ can be viewed as a snapshot of a virtual $\sbra{2N-1}$-element ULA on which $K$ sources impinge. Based on this observation, ANM-SMV attempts to solve the following ANM problem:
\equ{\min_{\m{v}} \atomn{\m{v}}, \st \twon{\m{v} - \widetilde{\m{v}}} \leq \eta, \label{eq:ANM_SMV}}
where $\widetilde{\m{v}}$ denotes an estimate of $\m{v}$, which will be given later, and $\eta$ is an upper bound on the error of $\widetilde{\m{v}}$. By casting \eqref{eq:ANM_SMV} as SDP, this problem can be solved and the frequencies can then be estimated as those composing the solution $\m{v}$.

The remaining task is to compose the estimate $\widetilde{\m{v}}$ and analyze its error bound $\eta$. To do so, the noise variance $\sigma$ is assumed to be known. Using \eqref{eq:ROmegainTs}, an estimate of $\m{\Gamma}_{\Omega}\m{T}\sbra{\m{u}}\m{\Gamma}_{\Omega}^T$ is formed as $\widetilde{\m{R}}_{\Omega} - \sigma\m{I}$. After that, an estimate of $\m{u}$ is obtained by averaging the corresponding entries of the estimate $\widetilde{\m{R}}_{\Omega} - \sigma\m{I}$. This can be done since $\Omega$ is assumed to be a redundancy array. Finally, $\widetilde{\m{v}}$ is obtained by using \eqref{eq:vinu}. Under the assumption of i.i.d. Gaussian sources and noise and assuming sufficiently many snapshots, an error bound $\eta\propto\sigma$ is provided in \cite{tan2014direction} in the case of co-prime arrays. This bound might be extended to other cases but that is beyond the scope of this article. Based on the above observations and the result in \cite{fernandez2013support}, it can be shown that ANM-SMV can stably estimate the frequencies, provided that they are sufficiently separated, with a probability that increases with the snapshot number $L$.

\subsubsection{Nuclear Norm Minimization followed by MUSIC (NNM-MUSIC)}
Using a low rank matrix recovery technique and a subspace method, another covariance-based method was proposed in \cite{pal2014grid} that is called NNM-MUSIC. Based on \eqref{eq:ROmegainTs}, the paper \cite{pal2014grid} proposed a two-step approach: 1) First estimate the full data covariance matrix $\m{T}\sbra{\m{u}}$ by exploiting its low rank structure, and 2) Then estimate the frequencies from $\m{T}\sbra{\m{u}}$ using MUSIC.

In the first step, the following NNM problem is solved to estimate $\m{T}\sbra{\m{u}}$:
\equ{\min_{\m{R}} \norm{\m{R}}_{\star}, \st \frobn{\m{R} - \m{T}\sbra{\widetilde{\m{u}}}} \leq \eta.
\label{eq:NNM_MUSIC}}
In \eqref{eq:NNM_MUSIC}, $\m{T}\sbra{\widetilde{\m{u}}}$ denotes an estimate of the data covariance matrix, which is obtained by averaging the corresponding entries of the sample covariance matrix $\widetilde{\m{R}}_{\Omega}$, and $\eta$ measures the distance between the true low rank matrix $\m{T}\sbra{\m{u}}$ and the above estimate in the Frobenius norm. Once $\m{R}$ is solved for, MUSIC is adopted to estimate the frequencies from the subspace of $\m{R}$.

By setting $\eta = \sqrt{N}\sigma$ and assuming $L\rightarrow \infty$, it was shown in \cite{pal2014grid} that solving \eqref{eq:NNM_MUSIC} can exactly recover $\m{T}\sbra{\m{u}}$. However, we note that it is not an easy task to choose $\eta$ in practice. Moveover, although the Toeplitz structure embedded in the data covariance matrix $\m{R}_{\Omega}$ is exploited to form the estimate $\widetilde{\m{u}}$, this structure is not utilized in the matrix denoising step. It was argued in \cite{pal2014grid} that the PSDness of $\m{R}$ can be preserved by solving \eqref{eq:NNM_MUSIC} if $\m{T}\sbra{\widetilde{\m{u}}}$ is PSD, which holds true if sufficiently many snapshots are available.

\subsubsection{Comparison of GLS, ANM-SMV and NNM-MUSIC}
In this subsection we compare the three covariance-based methods, namely GLS, ANM-SMV and NNM-MUSIC. While these methods are derived under similar assumptions on sources and noise, we argue that GLS can outperform ANM-SMV and NNM-MUSIC in several aspects. First, GLS is hyperparameter-free and can consistently estimate the noise variance $\sigma$, whereas ANM-SMV and NNM-MUSIC usually require knowledge of this variance since the error bounds $\eta$ in \eqref{eq:ANM_SMV} and \eqref{eq:NNM_MUSIC} are functions of $\sigma$. In fact, even when $\sigma$ is known the choice of $\eta$ is still not easy. Second, ANM-SMV and NNM-MUSIC are usable only with sufficiently many snapshots (which are needed to obtain a reasonable estimate of $\m{u}$ as well as a reasonable error bound $\eta$), while GLS can be used even with a single snapshot. Third, GLS and NNM-MUSIC are statistically consistent but ANM-SMV is not. Note that ANM-SMV still suffers from a resolution limit even if $\widetilde{\m{v}}$ in \eqref{eq:ANM_SMV} is exactly estimated given an infinitely number of snapshots. Fourth, GLS is a large-snapshot realization of the ML estimation while ANM-SMV and NNM-MUSIC are not.

Last but not least, ANM-SMV and NNM-MUSIC cannot exactly recover the frequencies in the absence of noise since the estimate $\widetilde{\m{v}}$ or $\widetilde{\m{u}}$ will suffer from some approximation error with finite snapshots. In contrast to this, GLS can exactly recover the frequencies under a mild frequency separation condition due to its connection to the atomic norm that will be discussed in Section \ref{sec:atomnorm_MMV} (considering the single snapshot case as an example).

\subsection{The Multiple Snapshot Case: Deterministic Methods}
In this subsection we present several deterministic gridless sparse methods for the case of multiple snapshots. The main idea is to utilize the temporal redundancy of the snapshots. Different from the covariance-based methods in the preceding subsection, these deterministic optimization methods are derived without statistical assumptions on the sources (though weak technical assumptions might be needed to provide theoretical guarantees). As in the single snapshot case, we first provide a general framework for such methods. We then discuss the potential advantages of multiple snapshots for DOA/frequency estimation based on an MMV atomic $\ell_0$ norm formulation. After that, the MMV atomic norm method will be presented. Finally, a possible extension of EMaC to the multiple snapshot case is discussed.

\subsubsection{A General Framework}

The data model in \eqref{formu:observation_model3} can be written as:
\equ{\m{Y}_{\Omega}=\m{Z}_{\Omega} + \m{E}_{\Omega}, \quad \m{Z}= \m{A}\sbra{\m{f}}\m{S}, \label{eq:modelMMV}}
where $\m{Z}$ denotes the noiseless multiple snapshot signal that contains the frequencies of interest. In the presence of bounded noise with $\frobn{\m{E}_{\Omega}}\leq \eta$, we solve the constrained optimization problem:
\equ{\min_{\m{Z}} \cM\sbra{\m{Z}}, \st \frobn{\m{Z}_{\Omega} - \m{Y}_{\Omega}}\leq \eta. \label{eq:sparseoptinz_MS}}
We can instead solve the regularized optimization problem given by
\equ{\min_{\m{Z}} \lambda\cM\sbra{\m{Z}} + \frac{1}{2}\frobn{\m{Z}_{\Omega} - \m{Y}_{\Omega}}^2, \label{eq:sparseoptinz2_MS}}
where $\lambda>0$ is a regularization parameter. In the extreme noiseless case, both \eqref{eq:sparseoptinz_MS} and \eqref{eq:sparseoptinz2_MS} degenerate to the following problem:
\equ{\min_{\m{Z}} \cM\sbra{\m{Z}}, \st \m{Z}_{\Omega} = \m{Y}_{\Omega}. \label{eq:sparseoptinz3_MS}}
In \eqref{eq:sparseoptinz_MS}-\eqref{eq:sparseoptinz3_MS}, $\cM\sbra{\m{Z}}$ denotes a sparse metric. By minimizing $\cM\sbra{\m{Z}}$ we attempt to reduce the number of frequencies composing $\m{Z}$.

\subsubsection{Atomic $\ell_0$ Norm}
In this subsection we study the advantage of using multiple snapshots for frequency estimation and extend the result in Section \ref{sec:l0_MMV} from the discrete to the continuous setting. To do so, we extend the atomic $\ell_0$ norm from the single to the multiple snapshot case. Note that the noiseless multiple snapshot signal $\m{Z}$ in \eqref{eq:modelMMV} can be written as:
\equ{\m{Z} = \sum_{k=1}^K \m{a}\sbra{f_k} \m{s}_k = \sum_{k=1}^K c_k \m{a}\sbra{f_k} \m{\phi}_k, }
where $\m{s}_k= \mbra{s_{k1},\dots,s_{kL}}\in\bC^{1\times L}$, $c_k = \twon{\m{s}_k}>0$, and $\m{\phi}_k = c_k^{-1}\m{s}_k\in\bC^{1\times L}$ with $\twon{\m{\phi}_k} = 1$. We therefore define the set of atoms in this case as:
\equ{\cA= \lbra{\m{a}\sbra{f_k,\m{\phi}_k}=\m{a}\sbra{f_k} \m{\phi}_k:\; f_k\in\bT, \m{\phi}_k \in\bC^{1\times L}, \twon{\m{\phi}_k}=1} \label{eq:atomset_MS}}
Note that $\m{Z}$ is a linear combination of $K$ atoms in $\cA$. The atomic $\ell_0$ norm of $\m{Z}$ induced by the new atomic set $\cA$ is given by:
\equ{\begin{split}\norm{\m{Z}}_{\cA,0}=
&\inf_{c_k, f_k, \m{\phi}_k} \lbra{\cK:\; \m{Z} = \sum_{k=1}^{\cK} \m{a}\sbra{f_k,\m{\phi}_k}c_k, f_k\in\bT, \twon{\m{\phi}_k}=1, c_k>0 } \\
=& \inf_{f_k, \m{s}_k} \lbra{\cK:\; \m{Z} = \sum_{k=1}^{\cK} \m{a}\sbra{f_k}\m{s}_k, f_k\in\bT}. \end{split} \label{eq:atom0n_MS}}

Using the atomic $\ell_0$ norm, in the noiseless case, the problem resulting from \eqref{eq:sparseoptinz3_MS} is given by:
\equ{\min_{\m{Z}} \norm{\m{Z}}_{\cA,0}, \st \m{Z}_{\Omega} = \m{Y}_{\Omega}. \label{eq:atom0nm_MS}}
To show the advantage of multiple snapshots, we define the continuous dictionary (see also \eqref{eq:cAtheta})
\equ{\cA_{\Omega}^1 = \lbra{\m{a}_{\Omega}\sbra{f}:\; f\in\bT}}
and let $\text{spark}\sbra{\cA_{\Omega}^1}$ be its spark. We have the following theoretical guarantee for \eqref{eq:atom0nm_MS} that generalizes the result in \cite[Theorem 2.4]{chen2006theoretical}.

\begin{thm}[\cite{yang2016exact}] $\m{Y}=\sum_{j=1}^K c_j\m{a}\sbra{f_j,\m{\phi}_j}$ is the unique solution to \eqref{eq:atom0nm_MS} if
\equ{K< \frac{\spark\sbra{\cA_{\Omega}^1}-1+\rank \sbra{\m{Y}_{\Omega}}}{2}. \label{Kbound}}
Moreover, the atomic decomposition above is the only one that satisfies $\norm{\m{Y}}_{\cA,0}=K$.
\label{thm:AL0_guanrantee}
\end{thm}

Note that the condition in \eqref{Kbound} coincides with that in \eqref{eq:cond2} required to guarantee parameter identifiability for DOA estimation. In fact it can be shown that, in the noiseless case, the frequencies/DOAs can be uniquely determined by the atomic $\ell_0$ norm optimization if and only if they can be uniquely identified (see, e.g., \cite[Proposition 2]{yang2016vandermonde}). Furthermore, the above result also holds true for general array geometries and even for general parameter estimation problems, provided that the atomic $\ell_0$ norm is appropriately defined.

By Theorem \ref{thm:AL0_guanrantee} the frequencies can be exactly determined by \eqref{eq:atom0nm_MS} if the number of sources $K$ is sufficiently small with respect to the array geometry $\Omega$ and the observed data $\m{Y}_{\Omega}$. Note that the number of recoverable frequencies can be increased, as compared to the SMV case, if $\rank \sbra{\m{Y}_{\Omega}}>1$, which always happens but in the trivial case when the multiple snapshots in $\m{Y}_{\Omega}$ are identical up to scaling factors.

As in the single snapshot case, computing $\norm{\m{Z}}_{\cA,0}$ can be cast as a rank minimization problem. To see this, let $\m{T}\sbra{\m{u}}$ be a Toeplitz matrix and impose the condition that
\equ{\begin{bmatrix} \m{X} & \m{Z}^H \\ \m{Z} & \m{T}\sbra{\m{u}} \end{bmatrix}\geq \m{0}, \label{eq:psdcond_MS}}
where $\m{X}$ is a free matrix variable. By invoking the Vandermonde decomposition, we see that $\m{T}\sbra{\m{u}}$ admits a $\rank\sbra{\m{T}\sbra{\m{u}}}$-atomic Vandermonde decomposition. Moreover, $\m{Z}$ lies in the range space of $\m{T}\sbra{\m{u}}$ and thus it can be expressed by $\rank\sbra{\m{T}\sbra{\m{u}}}$ atoms. Formally, we have the following result.

\begin{thm}[\cite{yang2014continuous,yang2016exact}] $\norm{\m{Z}}_{\cA,0}$ defined in \eqref{eq:atom0n_MS} equals the optimal value of the following rank minimization problem:
\equ{\min_{\m{X},\m{u}} \rank\sbra{\m{T}\sbra{\m{u}}}, \st \eqref{eq:psdcond_MS}. \label{formu:AL0_rankmin_MS}}
\label{thm:AL0_rankmin_MS}
\end{thm}

While the rank minimization problem cannot be easily solved, we next discuss its convex relaxation, namely the atomic norm method.

\subsubsection{Atomic Norm} \label{sec:atomnorm_MMV}
To provide a practical approach, we study the atomic norm induced by the atomic set $\cA$ defined in \eqref{eq:atomset_MS}. As in the single snapshot case, we have that
\equ{\begin{split}\atomn{\m{Z}}
&=\inf\lbra{t>0: \m{Z}\in t\text{conv}\sbra{\cA}} \\
&=\inf_{c_k,f_k, \m{\phi}_{k}} \lbra{\sum_{k} c_k:\; \m{Z} = \sum_k \m{a}\sbra{f_k,\m{\phi}_k}c_{k}, f_k\in\bT,\twon{\m{\phi}_k}=1, c_k>0 } \\
&= \inf_{f_k, \m{s}_{k}} \lbra{\sum_{k} \twon{\m{s}_{k}}:\; \m{Z} = \sum_k \m{a}\sbra{f_k}\m{s}_{k}, f_k\in\bT }. \end{split} \label{eq:atomn_MS}}
Note that $\atomn{\m{Z}}$ is in fact a continuous counterpart of the $\ell_{2,1}$ norm.
Moreover, it is shown in the following result that $\atomn{\m{Z}}$ can also be cast as SDP.

\begin{thm}[\cite{yang2016exact,li2016off}] $\norm{\m{Z}}_{\cA}$ defined in \eqref{eq:atomn_MS} equals the optimal value of the following SDP:
\equ{\min_{\m{X},\m{u}} \frac{1}{2\sqrt{N}}\mbra{\tr\sbra{\m{X}} + \tr\sbra{\m{T}\sbra{\m{u}}}}, \st \eqref{eq:psdcond_MS}. \label{formu:AN_SDP_MS}} \label{thm:AN_SDP_MS}
\end{thm}

The proof of Theorem \ref{thm:AN_SDP_MS} is omitted since it is very similar to that in the case of a single snapshot. Similarly, the frequencies and the powers are encoded in the Toeplitz matrix $\m{T}\sbra{\m{u}}$ and therefore, once $\m{T}\sbra{\m{u}}$ is obtained, they can be retrieved from its Vandermonde decomposition.

By applying Theorem \ref{thm:AN_SDP_MS}, in the noiseless case, the ANM problem resulting from \eqref{eq:sparseoptinz3_MS} can be cast as the following SDP:
\equ{\min_{\m{X},\m{u},\m{Z}} \tr\sbra{\m{X}} + \tr\sbra{\m{T}\sbra{\m{u}}},\st \eqref{eq:psdcond_MS} \text{ and } \m{Z}_{\Omega}=\m{Y}_{\Omega}.  \label{formu:AN_sdp_MS}}
For \eqref{formu:AN_sdp_MS}, we have theoretical guarantees similar to those in the single snapshot case; in other words, the frequencies can be exactly recovered by solving \eqref{formu:AN_sdp_MS} under appropriate conditions. Formally, we have the following results that generalize those in the single snapshot case.

\begin{thm}[\cite{yang2016exact}] $\m{Y}=\sum_{j=1}^K c_j\m{a}\sbra{f_j,\m{\phi}_j}$ is the unique atomic decomposition satisfying $\norm{\m{Y}}_{\cA}=\sum_{j=1}^K c_j$ if $\Delta_{\cT}\geq \frac{1}{\lfloor(N-1)/4\rfloor}$ and $N\geq257$. \label{thm:completedata_MS}
\end{thm}

\begin{thm}[\cite{yang2016exact}] Suppose we observe the rows of $\m{Y}=\sum_{j=1}^K c_j\m{a}\sbra{f_j,\m{\phi}_j}$
that are indexed by $\Omega$, where $\Omega\subset\mbra{1,\dots,N}$ is of size $M$ and is selected uniformly at random. Assume that $\lbra{\m{\phi}_j}$ are independent random variables on the unit hyper-sphere with $\bE\m{\phi}_j=\m{0}$. If $\Delta_{\cT}\geq \frac{1}{\lfloor(N-1)/4\rfloor}$, then there exists a numerical constant $C$ such that
\equ{M\geq C\max\lbra{\log^2\frac{\sqrt{L}N}{\delta}, K\log\frac{K}{\delta}\log\frac{\sqrt{L}N}{\delta}} \label{formu:AN_bound_MS}}
is sufficient to guarantee that, with probability at least $1-\delta$, $\m{Y}$ is the unique solution to \eqref{formu:AN_sdp_MS} and $\m{Y}=\sum_{j=1}^K c_j\m{a}\sbra{f_j,\m{\phi}_j}$ is the unique atomic decomposition satisfying $\norm{\m{Y}}_{\cA}=\sum_{j=1}^K c_j$. \label{thm:incompletedata_MS}
\end{thm}

Note that we have not made any assumptions on the sources in Theorem \ref{thm:completedata_MS} and therefore it applies to all kinds of source signals, including coherent sources. As a result, one cannot expect that the theoretical guarantee improves over the single snapshot case.

Note that the dependence of $M$ on $L$ in the bound \eqref{formu:AN_bound_MS} is for controlling the probability of successful recovery. To make it clear, consider the case when we seek to recover the columns of $\m{Y}$ independently using the single snapshot ANM method. When $M$ satisfies (\ref{formu:AN_bound_MS}) for $L=1$, each column of $\m{Y}$ can be recovered with probability $1-\delta$. It follows that $\m{Y}$ can be exactly recovered with probability at least $1-L\delta$. In contrast to this, if we recover $\m{Y}$ using the multiple snapshot ANM method, then with the same number of measurements the success probability is improved to $1-\sqrt{L}\delta$ (to see this, replace $\delta$ in (\ref{formu:AN_bound_MS}) by $\sqrt{L}\delta$).

Note also that the assumption on $\lbra{\m{\phi}_j}$ in Theorem \ref{thm:incompletedata_MS} is weak in the sense that the sources can be coherent. To see this, suppose that the rows of $\m{S}$ are i.i.d. Gaussian with mean zero and a covariance matrix whose rank is equal to one. Then the sources are identical up to random global phases and hence they are independent but coherent. This explains why the theoretical guarantee given in Theorem \ref{thm:incompletedata_MS} does not improve over the similar result in the single snapshot case. In other words, the results of Theorems \ref{thm:completedata_MS} and \ref{thm:incompletedata_MS} are \emph{worst case} analysis. Although these results do not shed light on the advantage of multiple snapshots, numerical simulations show that the atomic norm approach significantly improves the recovery performance, compared to the case of $L=1$, when the source signals are at general positions (see, e.g., \cite{yang2014continuous,yang2016exact,li2016off}).

In the presence of noise, the ANM problem resulting from \eqref{eq:sparseoptinz2_MS} can be formulated as the following SDP:
\equ{\min_{\m{X},\m{u},\m{Z}} \frac{\lambda}{2\sqrt{N}}\mbra{\tr\sbra{\m{X}}+\tr\sbra{\m{T}\sbra{\m{u}}}} + \frac{1}{2}\frobn{\m{Z}_{\Omega} - \m{Y}_{\Omega}}^2, \st \eqref{eq:psdcond_MS}. \label{eq:sparseoptinz2_MS2}}
It was shown in \cite{li2016off} that in the ULA case the choice of $\lambda \approx \sqrt{M\sbra{L+\log M + \sqrt{2L\log M}}\sigma}$ results in a stable recovery of the signal $\m{Z}$, which generalizes the result in the single snapshot case. In the SLA case, a similar parameter tuning method can be derived by combining the techniques in \cite{li2016off} and \cite{yang2015gridless}.

\subsubsection{Hankel-based Nuclear Norm}
In this subsection we extend the EMaC method from the single to the multiple snapshot case. For each noiseless snapshot $\m{z}(t)$, we can form an $m\times n$ Hankel matrix $\m{H}\sbra{\m{z}(t)}$ as in \eqref{eq:hankel}, where $m+n=N+1$, and then combine them in the following $m\times nL$ matrix:
\equ{\m{H}\sbra{\m{Z}}= \mbra{\m{H}\sbra{\m{z}(1)}, \dots, \m{H}\sbra{\m{z}(L)}}.}
Using the decomposition in \eqref{eq:Hankdec}, we have that
\equ{\begin{split}\m{H}\sbra{\m{Z}} =
&\sum_{k=1}^K \begin{bmatrix} 1 \\ e^{i2\pi f_k} \\ \vdots \\ e^{i2\pi (m-1)f_k} \end{bmatrix}\times \\
&\mbra{ s_k(1), s_k(1)e^{i2\pi f_k}, \dots, s_k(1)e^{i2\pi (n-1)f_k}, \dots,  s_k(L), s_k(L)e^{i2\pi f_k}, \dots, s_k(L)e^{i2\pi (n-1)f_k} }. \end{split} \label{eq:BigHankdec}}
As a result,
\equ{\rank\sbra{\m{H}\sbra{\m{Z}}}\leq K. \label{eq:rankHZleqK}}
It follows that if $K<\min\sbra{m,nL}$, then $\m{H}\sbra{\m{Z}}$ is a low rank matrix. Therefore, we may recover $\m{H}\sbra{\m{Z}}$ by minimizing its nuclear norm, i.e., by letting (see \eqref{eq:Hnnsdp})
\equ{\begin{split}\cM\sbra{\m{Z}}
&= \norm{\m{H}\sbra{\m{Z}}}_{\star}\\
&=\min_{\m{Q}_1,\m{Q}_2}\frac{1}{2}\mbra{\tr\sbra{\m{Q}_1} + \tr\sbra{\m{Q}_2}}, \st \begin{bmatrix} \m{Q}_1 & \m{H}\sbra{\m{Z}}^H \\ \m{H}\sbra{\m{Z}} & \m{Q}_2 \end{bmatrix} \geq \m{0}. \end{split} \label{eq:MisnnBigHz}}
The resulting approach is referred to as M-EMaC.

A challenging problem when applying M-EMaC is the choice of the parameter $m$. Intuitively, we need to ensure that the equality holds in \eqref{eq:rankHZleqK} for the true signal $\m{Z}$ so that the data information can be appropriately encoded and the frequencies can be correctly recovered from $\m{H}\sbra{\m{Z}}$. This is guaranteed for a single snapshot once $\m{H}\sbra{\m{Z}}$ is rank deficient. Unfortunately, a similar argument does not hold in the case of multiple snapshots. In particular, it can be seen from \eqref{eq:BigHankdec} that the rank of $\m{H}\sbra{\m{Z}}$ also depends on the unknown source signals $\m{S}$. As an example, in the extreme case when all the sources are coherent, we have that
\equ{\rank\sbra{\m{H}\sbra{\m{Z}}}\leq \min\sbra{K, m,n},}
which can be much smaller than $nL$. As a result, if we know that $K<\frac{N}{2}$, then we can set $m=\lceil \frac{N}{2} \rceil$. We leave the parameter tuning in the case of $K\geq\frac{N}{2}$ as an open problem.

\subsection{Reweighted Atomic Norm Minimization}
In contrast to the EMaC, the atomic norm methods of Sections \ref{sec:AN_S} and \ref{sec:atomnorm_MMV} can better preserve the signal structure, which is important especially in the noisy case. However, a major disadvantage of the atomic norm is that it can theoretically guarantee successful frequency recovery if they are separated by at least $\frac{2.52}{N}$ \cite{fernandez2016super}. To overcome this ``resolution limit''\footnote{Since the aforementioned frequency separation condition is sufficient but not necessary, this might not be the real resolution limit of the atomic norm.}, we present in this subsection the reweighted atomic-norm minimization (RAM) method that was proposed in \cite{yang2016enhancing}. The key idea of RAM is to use a smooth surrogate for the atomic $\ell_0$ norm, which exploits the sparsity to the greatest extent possible and does not suffer from any resolution limit but is nonconvex and non-smooth, and then optimize the surrogate using a reweighted approach. Interestingly, the resulting reweighted approach is shown to be a reweighted atomic norm with a sound weighting function that gradually enhances sparsity and resolution. While several reweighted approaches have been proposed in the discrete setting (see, e.g., \cite{lobo2007portfolio,candes2008enhancing, wipf2010iterative,stoica2012spice}), RAM appears to be the first for continuous dictionaries. Since RAM can be applied to single or multiple snapshots as well as to ULA or SLA, we present the result in the most general multiple snapshot SLA case, as in the preceding subsection.

\subsubsection{A Smooth Surrogate for $\norm{\m{Z}}_{\cA,0}$}
To derive RAM, we first introduce a smooth surrogate for $\norm{\m{Z}}_{\cA,0}$ defined in \eqref{eq:atom0n_MS}. Note that if the surrogate function is given directly in the continuous frequency domain, then a difficult question is whether and how it can be formulated as a finite-dimensional optimization problem, as $\norm{\m{Z}}_{\cA,0}$ and $\norm{\m{Z}}_{\cA}$. To circumvent this problem, RAM operates instead in the re-parameterized $\m{u}$ domain. In particular, we have shown that $\norm{\m{Z}}_{\cA,0}$ is equivalent to the following rank minimization problem:
\equ{\min_{\m{X},\m{u}} \rank\sbra{\m{T}\sbra{\m{u}}}, \st \begin{bmatrix} \m{X} & \m{Z}^H \\ \m{Z} & \m{T}\sbra{\m{u}} \end{bmatrix}\geq \m{0}. \label{formu:AL0_rankmin0}}
Inspired by the literature on low rank matrix recovery (see, e.g., \cite{david1994algorithms,fazel2003log,mohan2012iterative}), the log-det heuristic is adopted as a smooth surrogate for the matrix rank, resulting in the following sparse metric:
\equ{\cM^{\epsilon}\sbra{\m{Z}} = \min_{\m{X},\m{u}} \log\abs{\m{T}\sbra{\m{u}}+\epsilon\m{I}} + \tr\sbra{\m{X}}, \st \begin{bmatrix} \m{X} & \m{Z}^H \\ \m{Z} & \m{T}\sbra{\m{u}} \end{bmatrix}\geq \m{0}. \label{eq:MepsZ} }
In \eqref{eq:MepsZ}, $\epsilon$ is a tuning parameter that avoids the first term in the objective from being $-\infty$ when $\m{T}\sbra{\m{u}}$ is rank-deficient, and $\tr\sbra{\m{X}}$ is included in the objective to prevent the trivial solution $\m{u}=\m{0}$, as in the SDP in \eqref{formu:AN_SDP_MS} for the atomic norm. Intuitively, as $\epsilon\rightarrow0$, the log-det heuristic approaches the rank function and $\cM^{\epsilon}\sbra{\m{Z}}$ would approach $\norm{\m{Z}}_{\cA,0}$. This is indeed true and is formally stated in the following result.

\begin{thm}[\cite{yang2016enhancing}] Let $r=\norm{\m{Z}}_{\cA,0}$ and let $\epsilon\rightarrow0$. Then, we have the following results:
\begin{enumerate}
 \item If $r\leq N-1$, then
     \equ{\cM^{\epsilon}\sbra{\m{Z}}\sim \sbra{r-N}\ln\frac{1}{\epsilon}, \label{formu:equivinf}}
     i.e., the two quantities above are equivalent infinities. Otherwise, $\cM^{\epsilon}\sbra{\m{Z}}$ approaches a constant depending only on $\m{Z}$;
 \item Let $\widehat{\m{u}}_{\epsilon}$ be the (global) solution $\m{u}$ to the optimization problem in \eqref{eq:MepsZ}. Then, the smallest $N-r$ eigenvalues of $\m{T}\sbra{\widehat{\m{u}}_{\epsilon}}$ are either zero or approach zero as fast as $\epsilon$;
 \item For any cluster point of $\widehat{\m{u}}_{\epsilon}$ at $\epsilon=0$, denoted by $\widehat{\m{u}}_0$, there exists an atomic decomposition $\m{Z}=\sum_{k=1}^r\m{a}\sbra{f_k}\m{s}_k$ such that $\m{T}\sbra{\widehat{\m{u}}_0}=\sum_{k=1}^r\twon{\m{s}_k}^2\m{a}\sbra{f_k}\m{a}\sbra{f_k}^H$. \footnote{$\m{u}_0$ is called a cluster point of a vector-valued function $\m{u}(x)$ at $x=x_0$ if there exists a sequence $\lbra{x_n}_{n=1}^{+\infty}$, $\lim_{n\rightarrow+\infty}x_n=x_0$, satisfying $\lim_{n\rightarrow+\infty}\m{u}(x_n)= \m{u}_0$.}
\end{enumerate} \label{thm:epsilontozero}
\end{thm}

Theorem \ref{thm:epsilontozero} shows that the sparse metric $\cM^{\epsilon}\sbra{\m{Z}}$ approaches $\norm{\m{Z}}_{\cA,0}$ as $\epsilon\rightarrow0$. Moreover, it characterizes the properties of the optimizer $\widehat{\m{u}}_{\epsilon}$, as $\epsilon\rightarrow0$, including the convergence speed of the smallest $N-\norm{\m{Z}}_{\cA,0}$ eigenvalues and the limiting form of $\m{T}\sbra{\widehat{\m{u}}_0}$ via the Vandermonde decomposition. It is worth noting that the term $\ln\frac{1}{\epsilon}$ in \eqref{formu:equivinf}, which becomes unbounded as $\epsilon\rightarrow0$, is not problematic in the optimization problem, since the objective function $\cM^{\epsilon}\sbra{\m{Z}}$ can be re-scaled by $\sbra{\ln\frac{1}{\epsilon}}^{-1}$ for any $\epsilon>0$ without altering the optimizer.

In another interesting extreme case when $\epsilon\rightarrow+\infty$, the following result shows that $\cM^{\epsilon}\sbra{\m{Z}}$ in fact plays the same rule as $\norm{\m{Z}}_{\cA}$.

\begin{thm}[\cite{yang2016enhancing}] Let $\epsilon\rightarrow+\infty$. Then,
\equ{\cM^{\epsilon}\sbra{\m{Z}}-N\ln\epsilon \sim 2\sqrt{N}\atomn{\m{Z}}\epsilon^{-\frac{1}{2}}, \label{formu:epsilontoinf}}
i.e., the two quantities above are equivalent infinitesimals. \label{thm:epsilontoinf}
\end{thm}

As a result, the new sparse metric $\cM^{\epsilon}\sbra{\m{Z}}$ bridges the atomic norm and the atomic $\ell_0$ norm. As $\epsilon$ approaches $+\infty$, it approaches the former which is convex and can be globally computed but suffers from a resolution limit. As $\epsilon$ approaches $0$, it approaches the latter that exploits sparsity to the greatest extent possible and has no resolution limit but cannot be directly computed.

\subsubsection{A Locally Convergent Iterative Algorithm}
Inserting \eqref{eq:MepsZ} into \eqref{eq:sparseoptinz_MS}, we obtain the following optimization problem:
\equ{\begin{split}
&\min_{\m{X},\m{u},\m{Z}} \log\abs{\m{T}\sbra{\m{u}}+\epsilon\m{I}} + \tr\sbra{\m{X}}, \\
&\st \begin{bmatrix} \m{X} & \m{Z}^H \\ \m{Z} & \m{T}\sbra{\m{u}} \end{bmatrix}\geq \m{0} \text{ and } \frobn{\m{Z}_{\Omega} - \m{Y}_{\Omega}}\leq \eta. \end{split} \label{eq:MepsMin}}
This problem is nonconvex since the log-det function is nonconvex. In fact, $\log\abs{\m{T}\sbra{\m{u}}+\epsilon\m{I}}$ is a concave function of $\m{u}$ since $\log\abs{\m{R}}$ is a concave function of $\m{R}$ on the positive semidefinite cone \cite{boyd2004convex}. A popular locally convergent approach to the minimization of such a concave $+$ convex function is the majorization-maximization (MM) algorithm (see, e.g., \cite{fazel2003log}). Let $\m{u}_j$ denote the $j$th iterate of the optimization variable $\m{u}$. Then, at the $\sbra{j+1}$th iteration we replace $\ln\abs{\m{T}\sbra{\m{u}}+\epsilon\m{I}}$ by its tangent plane at the current value $\m{u}=\m{u}_j$:
\equ{\begin{split}
&\ln\abs{\m{T}\sbra{\m{u}_j}+\epsilon\m{I}} + \tr\mbra{\sbra{\m{T}\sbra{\m{u}_j}+\epsilon\m{I}}^{-1}\m{T}\sbra{\m{u}-\m{u}_j}}\\
&= \tr\mbra{\sbra{\m{T}\sbra{\m{u}_j}+\epsilon\m{I}}^{-1}\m{T}\sbra{\m{u}}} + c_j, \end{split}}
where $c_j$ is a constant independent of $\m{u}$. As a result, the optimization problem at the $\sbra{j+1}$th iteration becomes
\equ{\begin{split}
&\min_{\m{X},\m{u},\m{Z}} \tr\mbra{\sbra{\m{T}\sbra{\m{u}_j}+\epsilon\m{I}}^{-1}\m{T}\sbra{\m{u}}} + \tr\sbra{\m{X}}, \\
&\st \begin{bmatrix} \m{X} & \m{Z}^H \\ \m{Z} & \m{T}\sbra{\m{u}} \end{bmatrix}\geq \m{0} \text{ and } \frobn{\m{Z}_{\Omega} - \m{Y}_{\Omega}}\leq \eta. \end{split} \label{formu:problem_j}}

Note that the problem in \eqref{formu:problem_j} is an SDP that can be globally solved. Since $\log\abs{\m{T}\sbra{\m{u}}+\epsilon\m{I}}$ is strictly concave in $\m{u}$, at each iteration its value decreases by an amount greater than the decrease of its tangent plane. It follows that by iteratively solving \eqref{formu:problem_j} the objective function in \eqref{eq:MepsMin} monotonically decreases and converges to a local minimum.

\subsubsection{Interpretation as RAM}
We show in this subsection that \eqref{formu:problem_j} is actually a weighted atomic norm minimization problem. To do so, let us define a weighted atomic set as (compare with the original atomic set $\cA$ defined in \eqref{eq:atomset_MS}):
\equ{\cA^w=\lbra{\m{a}^w\sbra{f,\m{\phi}}= w\sbra{f}\m{a}\sbra{f}\m{\phi}:\; f\in\bT, \m{\phi}\in\bC^{1\times L}, \twon{\m{\phi}}=1},}
where $w\sbra{f}\geq 0$ denotes a weighting function. For $\m{Z}\in\bC^{N\times L}$, define its weighted atomic norm as the atomic norm induced by $\cA^w$:
\equ{\begin{split}
\norm{\m{Z}}_{\cA^w}
&=\inf_{c_k,f_k,\m{\phi}_k}\lbra{\sum_k c_k:\; \m{Z} = \sum_k c_k \m{a}^w\sbra{f_k,\m{\phi}_k}, f_k\in\bT, \twon{\m{\phi}}=1, c_k>0}\\
&=\inf_{f_k,\m{s}_k}\lbra{\sum_k \frac{\twon{\m{s}_k}}{w\sbra{f_k}}:\; \m{Z} = \sum_k \m{a}\sbra{f_k}\m{s}_k }. \end{split}}
According to the definition above, $w\sbra{f}$ specifies the importance of the atom at $f$: the frequency $f\in\bT$ is more likely to be selected if $w\sbra{f}$ is larger. The atomic norm is a special case of the weighted atomic norm for a constant weighting function. Similar to the atomic norm, the proposed weighted atomic norm also admits a semidefinite formulation for an appropriate weighting function, which is stated in the following theorem.

\begin{thm}[\cite{yang2016enhancing}] Suppose that $w\sbra{f}= \frac{1}{\sqrt{\m{a}\sbra{f}^H\m{W}\m{a}\sbra{f}}}\geq 0$ with $\m{W}\in\bC^{N\times N}$. Then,
\equ{\begin{split}
\norm{\m{Z}}_{\cA^w}=&\min_{\m{X},\m{u}} \frac{\sqrt{N}}{2}\tr\sbra{\m{W}\m{T}\sbra{\m{u}}} + \frac{1}{2\sqrt{N}}\tr\sbra{\m{X}},\\
&\st \begin{bmatrix} \m{X} & \m{Z}^H \\ \m{Z} & \m{T}\sbra{\m{u}} \end{bmatrix}\geq \m{0}. \end{split} \label{eq:WAN_SDP}} \label{thm:weightAN}
\end{thm}

Let $\m{W}_j=\frac{1}{N}\sbra{\m{T}\sbra{\m{u}_j}+\epsilon\m{I}}^{-1}$ and $w_j\sbra{f} = \frac{1}{\sqrt{\m{a}\sbra{f}^H \m{W}_j \m{a}\sbra{f}}}$. It follows from Theorem \ref{thm:weightAN} that the optimization problem in (\ref{formu:problem_j}) can be exactly written as the following weighted atomic norm minimization problem:
\equ{\begin{split}
\min_{\m{Z}} \norm{\m{Z}}_{\cA^{w_j}}, \st \frobn{\m{Z}_{\Omega} - \m{Y}_{\Omega}}\leq \eta. \end{split} \label{formu:WAN_j}}
As a result, the whole iterative algorithm is referred to as reweighted atomic-norm minimization (RAM). Note that the weighting function is updated based on the latest solution $\m{u}$. If we let $w_0(f)$ be constant or equivalently, $\m{u}_0=\m{0}$, such that no preference of the atoms is specified at the first iteration, then the first iteration coincides with ANM. From the second iteration on, the preference is defined by the weighting function $w_j\sbra{f}$ given above. Note that $w_j^2(f)$ is nothing but the power spectrum of Capon's beamformer (interpreting $\m{T}\sbra{\m{u}_j}$ as the noiseless data covariance and $\epsilon$ as the noise variance); also note that similar weighting functions have also appeared in sparse optimization methods in the discrete setting (see, e.g., \cite{stoica2012spice,stoica2014weighted,wipf2010iterative}). The reweighting strategy makes the frequencies around those produced by the current iteration more preferable at the next iteration and thus enhances sparsity. At the same time, the weighting results in resolving finer details of the frequency spectrum in those areas and therefore enhances resolution. In a practical implementation of RAM, we can start with the standard ANM, which corresponds to the case of $\epsilon\rightarrow+\infty$ (by Theorem \ref{thm:epsilontoinf}), and then gradually decrease $\epsilon$ during the iterations.

\subsection{Connections between ANM and GLS}
We have extended both the atomic norm and the GLS methods from the single to the multiple snapshot case. These two methods were shown in Section \ref{sec:ANM_GLS_SS} to be strongly connected to each other in the single snapshot case, so it is natural to ask whether they are also connected in the multiple snapshot case. We answer this question in this subsection following \cite{yang2016gridless1}. In particular, for a small number of snapshots the GLS optimization problem is shown to be equivalent to ANM as if there were no noise, whereas for a large number of snapshots it is equivalent to a weighted ANM. Similar results can also be proved for their discrete versions, viz.~$\ell_{2,1}$ optimization and SPICE.

\subsubsection{The Case of $L<M$}
We first consider the ULA case where the GLS optimization problem is given by \eqref{formu:SDP_Lsmall_identical2}:
\equ{\min_{\m{X},\m{u}} \tr\sbra{\m{X}}+ \tr\sbra{\m{T}\sbra{\m{u}}}, \st \begin{bmatrix}\m{X}& \widetilde{\m{Y}}^H \\ \widetilde{\m{Y}} & \m{T}\sbra{\m{u}} \end{bmatrix}\geq\m{0}, \label{formu:SDP_Lsmall_identical4}}
where $\widetilde{\m{Y}} = \frac{1}{L}\m{Y}\sbra{\m{Y}^H\m{Y}}^{\frac{1}{2}}$. By comparing \eqref{formu:SDP_Lsmall_identical4} and the SDP formulation of the atomic norm in \eqref{formu:AN_SDP_MS}, it can be seen that GLS actually computes $\atomn{\widetilde{\m{Y}}}$ (up to a scaling factor).

A similar argument also holds true in the SLA case where the GLS optimization problem is given by \eqref{formu:SDP_Lsmall_identical3}:
\equ{\min_{\m{X},\m{u}} \tr\sbra{\m{X}}+ \frac{M}{N}\tr\sbra{\m{T}\sbra{\m{u}}}, \st \begin{bmatrix}\m{X}& \widetilde{\m{Y}_{\Omega}}^H \\ \widetilde{\m{Y}_{\Omega}} & \m{\Gamma}_{\Omega}\m{T}\sbra{\m{u}}\m{\Gamma}_{\Omega}^T \end{bmatrix}\geq\m{0}, \label{formu:SDP_Lsmall_identical5}}
where $\widetilde{\m{R}}_{\Omega}$ in \eqref{formu:SDP_Lsmall_identical3} is replaced here by $\widetilde{\m{Y}_{\Omega}} = \frac{1}{L}\m{Y}_{\Omega}\sbra{\m{Y}_{\Omega}^H\m{Y}_{\Omega}}^{\frac{1}{2}}$ to reduce the dimensionality. Given $\m{T}\sbra{\m{u}}\geq \m{0}$ and applying the identity in \eqref{eq:augment}, we have that
\equ{\begin{split}
&\tr\sbra{\widetilde{\m{Y}_{\Omega}}^H \mbra{\m{\Gamma}_{\Omega}\m{T}\sbra{\m{u}}\m{\Gamma}_{\Omega}^T}^{-1} \widetilde{\m{Y}_{\Omega}}} \\
&= \sum_{t=1}^L \widetilde{\m{Y}_{\Omega}}^H(t) \mbra{\m{\Gamma}_{\Omega}\m{T}\sbra{\m{u}}\m{\Gamma}_{\Omega}^T}^{-1} \widetilde{\m{Y}_{\Omega}}(t) \\
&= \min_{\m{z}(t)} \sum_{t=1}^L \m{z}^H(t) \mbra{\m{T}\sbra{\m{u}}}^{-1} \m{z}(t), \st \m{z}_{\Omega}(t) =\widetilde{\m{Y}_{\Omega}}(t) \\
&= \min_{\m{Z}} \tr\sbra{\m{Z}^H \mbra{\m{T}\sbra{\m{u}}}^{-1} \m{Z}}, \st \m{Z}_{\Omega} =\widetilde{\m{Y}_{\Omega}}. \end{split}}
It follows that
\equ{\begin{split}\eqref{formu:SDP_Lsmall_identical5}
\Leftrightarrow& \min_{\m{u}} \tr\sbra{\widetilde{\m{Y}_{\Omega}}^H \mbra{\m{\Gamma}_{\Omega}\m{T}\sbra{\m{u}}\m{\Gamma}_{\Omega}^T}^{-1} \widetilde{\m{Y}_{\Omega}}} + \frac{M}{N}\tr\sbra{\m{T}\sbra{\m{u}}} \\
\Leftrightarrow& \min_{\m{u},\m{Z}} \tr\sbra{\m{Z}^H \mbra{\m{T}\sbra{\m{u}}}^{-1} \m{Z}} + \frac{M}{N}\tr\sbra{\m{T}\sbra{\m{u}}}, \st \m{Z}_{\Omega} =\widetilde{\m{Y}_{\Omega}} \\
\Leftrightarrow& \min_{\m{X},\m{u},\m{Z}} \frac{M}{N}\tr\sbra{\m{X}}+ \frac{M}{N}\tr\sbra{\m{T}\sbra{\m{u}}},\\
& \st \begin{bmatrix}\m{X}& \sqrt{\frac{N}{M}}\m{Z}^H \\ \sqrt{\frac{N}{M}}\m{Z} & \m{T}\sbra{\m{u}} \end{bmatrix}\geq\m{0} \text{ and } \m{Z}_{\Omega} =\widetilde{\m{Y}_{\Omega}}. \end{split}}
The above SDP is nothing but the following ANM problem (up to a scaling factor):
\equ{\min_{\m{Z}}\atomn{\m{Z}}, \st \m{Z}_{\Omega} =\widetilde{\m{Y}_{\Omega}}. }

We have therefore shown that when $L<M$ the GLS optimization problem is equivalent to certain atomic norm formulations obtained by transforming the observed snapshots. Note that the joint sparsity is preserved in the transformed snapshots in the limiting noiseless case. Therefore, in the absence of noise, by applying the results on the atomic norm, GLS is expected to exactly recover the frequencies under the frequency separation condition. This is true in the ULA case where Theorem \ref{thm:completedata_MS} can be directly applied. However, technically, a similar theoretical guarantee cannot be provided in the SLA case since the assumption on the phases in Theorem \ref{thm:incompletedata_MS} might not hold true for the transformed source signals.

\subsubsection{The Case of $L\geq M$}

In this case and for a ULA the GLS optimization problem is given by (see \eqref{formu:SDP_Llarge_identical}):
\equ{\min_{\m{X},\m{u}} \tr\sbra{\m{X}}+ \tr\sbra{\widetilde{\m{R}}^{-1}\m{T}\sbra{\m{u}}}, \st \begin{bmatrix}\m{X}& \widetilde{\m{R}}^{\frac{1}{2}} \\ \widetilde{\m{R}}^{\frac{1}{2}} & \m{T}\sbra{\m{u}} \end{bmatrix}\geq\m{0}. \label{formu:SDP_Llarge_identical4}}
According to \eqref{eq:WAN_SDP}, this is nothing but computing the weighted atomic norm $\norm{\widetilde{\m{R}}^{\frac{1}{2}}}_{\cA^w}$ (up to a scaling factor), where the weighting function is given by $w(f)=\frac{1}{\sqrt{\m{a}^H\sbra{f} \widetilde{\m{R}}^{-1} \m{a}\sbra{f}}}$. Note that
$w^2(f)=\frac{1}{\m{a}^H\sbra{f} \widetilde{\m{R}}^{-1} \m{a}\sbra{f}}$ is the power spectrum of the Capon's beamformer.

For an SLA, the GLS problem is given by (see \eqref{formu:SDP_Llarge_identical3}):
\equ{\begin{split}
\min_{\m{X},\m{u}} \tr\sbra{\m{X}}+ \tr\sbra{\m{\Gamma}_{\Omega}^T\widetilde{\m{R}}_{\Omega}^{-1}\m{\Gamma}_{\Omega}\m{T}\sbra{\m{u}}}, \st \begin{bmatrix}\m{X}& \widetilde{\m{R}}_{\Omega}^{\frac{1}{2}} \\ \widetilde{\m{R}}_{\Omega}^{\frac{1}{2}} & \m{\Gamma}_{\Omega}\m{T}\sbra{\m{u}}\m{\Gamma}_{\Omega}^T \end{bmatrix}\geq\m{0}.\end{split} \label{formu:SDP_Llarge_identical5}}
By arguments similar to those in the preceding subsection we have that
\equ{\begin{split}\eqref{formu:SDP_Llarge_identical5} \Leftrightarrow
&\min_{\m{X},\m{u},\m{Z}} \tr\sbra{\m{X}}+ \tr\sbra{\m{\Gamma}_{\Omega}^T\widetilde{\m{R}}_{\Omega}^{-1}\m{\Gamma}_{\Omega}\m{T}\sbra{\m{u}}},\\
&\st \begin{bmatrix}\m{X}& \m{Z}_{\Omega}^H \\ \m{Z}_{\Omega} & \m{T}\sbra{\m{u}} \end{bmatrix}\geq\m{0} \text{ and } \m{Z}_{\Omega} = \widetilde{\m{R}}_{\Omega}^{\frac{1}{2}}\\ \Leftrightarrow
&\min_{\m{Z}} \norm{\m{Z}}_{\cA^w}, \st \m{Z}_{\Omega} = \widetilde{\m{R}}_{\Omega}^{\frac{1}{2}}. \end{split} \label{eq:SDP_Llarge_identical6}}
In \eqref{eq:SDP_Llarge_identical6}, the weighting function of the weighted atomic norm is given by (up to a scaling factor)
\equ{w(f)=\frac{1}{\sqrt{\m{a}^H\sbra{f} \m{\Gamma}_{\Omega}^T\widetilde{\m{R}}_{\Omega}^{-1}\m{\Gamma}_{\Omega} \m{a}\sbra{f}}} = \frac{1}{\sqrt{\m{a}_{\Omega}^H\sbra{f} \widetilde{\m{R}}_{\Omega}^{-1} \m{a}_{\Omega}\sbra{f}}}.}
Therefore, here too $w^2(f)$ is the power spectrum of the Capon's beamformer.

We have shown above that for a sufficiently large number of snapshots ($L\geq M$) GLS corresponds to a weighted atomic norm in which the weighting function is given by the square root of the Capon's spectrum. While the standard atomic norm method suffers from a ``resolution limit'' of $\frac{2.52}{N}$, the above analysis shows that this limit can actually be overcome by using the weighted atomic norm method: indeed GLS is a consistent method and a large-snapshot realization of the MLE. However, it is worth noting that the standard atomic norm method can be applied to general source signals and its ``resolution limit'' is given by a worst case analysis, whereas the statistical properties of GLS are obtained under stronger statistical assumptions on the sources and its performance can degrade if these assumptions are not satisfied. The presented connections between GLS and ANM also imply that GLS is generally robust to source correlations, like ANM, though its power estimates can be biased \cite{yang2016gridless1}.

\subsection{Computational Issues and Solutions}
We have presented several gridless sparse methods for DOA/frequency estimation from multiple snapshots. Typically, these methods require computing the solution of an SDP. While several off-the-shelf solvers exist for solving SDP, they generally have high computational complexity. As an example, SDPT3 is an interior-point based method which has the computational complexity of $O\sbra{n_1^2n_2^{2.5}}$, where $n_1$ denotes the number of variables and $n_2\times n_2$ is the dimension of the PSD matrix of the SDP \cite{toh1999sdpt3,vandenberghe1996semidefinite}. To take a look at the computational complexity of the gridless sparse methods, we consider the atomic norm in \eqref{formu:AN_SDP_MS} as an example. Given $\m{Z}$ it can be seen that $n_1=N+L^2$ and $n_2=N+L$. So the computational complexity of computing the atomic norm can be rather large, viz.~$O\sbra{\sbra{N +L^2}^2\sbra{N+L}^{2.5}}$. In this subsection we present strategies for accelerating the computation.

\subsubsection{Dimensionality Reduction}
We show in this subsection that a similar dimensionality reduction technique as introduced in Section \ref{sec:dimred_D} can be applied to the atomic norm and the weighted atomic norm methods in the case when the number of snapshots $L$ is large. The technique was firstly proposed in \cite{yang2016enhancing} for the gridless setting studied here and it was extended to the discrete case in Section \ref{sec:dimred_D} (note that similar techniques were also reported in \cite{haghighatshoar2017massive,steffens2016compact} later on). It is also worth noting that a similar dimensionality reduction is not required by GLS since GLS is covariance-based and all the information in the data snapshots $\m{Y}_{\Omega}$ is encoded in the sample covariance matrix $\widetilde{\m{R}}_{\Omega}=\frac{1}{L}\m{Y}_{\Omega}^H\m{Y}_{\Omega}$ whose dimension does not increase with $L$. Since it has been shown that GLS and the (weighted) atomic norm are strongly connected, we may naturally wonder if the dimensionality of the atomic norm can be similarly reduced. An affirmative answer is provided in the following result.

\begin{thm}[\cite{yang2016enhancing}] Consider the three ANM problems resulting from \eqref{eq:sparseoptinz_MS}, \eqref{eq:sparseoptinz2_MS} and \eqref{eq:sparseoptinz3_MS} which, respectively, are given by:
{\lentwo\equa{
&&\min_{\m{X},\m{u},\m{Z}} \tr\sbra{\m{X}} + \tr\sbra{\m{T}\sbra{\m{u}}}, \notag \\
&&\st \begin{bmatrix} \m{X} & \m{Z}^H \\ \m{Z} & \m{T}\sbra{\m{u}} \end{bmatrix}\geq \m{0} \text{ and } \frobn{\m{Z}_{\Omega} - \m{Y}_{\Omega}}\leq \eta, \label{eq:ANM_MS1} \\
&&\min_{\m{X},\m{u},\m{Z}} \lambda'\tr\sbra{\m{X}} + \lambda'\tr\sbra{\m{T}\sbra{\m{u}}} + \frobn{\m{Z}_{\Omega} - \m{Y}_{\Omega}}^2, \notag \\
&&\st \begin{bmatrix} \m{X} & \m{Z}^H \\ \m{Z} & \m{T}\sbra{\m{u}} \end{bmatrix}\geq \m{0}, \label{eq:ANM_MS2} \\
&&\min_{\m{X},\m{u},\m{Z}} \tr\sbra{\m{X}} + \tr\sbra{\m{T}\sbra{\m{u}}}, \notag\\
&&\st \begin{bmatrix} \m{X} & \m{Z}^H \\ \m{Z} & \m{T}\sbra{\m{u}} \end{bmatrix}\geq \m{0} \text{ and } \m{Z}_{\Omega} = \m{Y}_{\Omega}. \label{eq:ANM_MS3}}
}Let $\widetilde{\m{Y}_{\Omega}}$ be any matrix satisfying $\widetilde{\m{Y}_{\Omega}}\widetilde{\m{Y}_{\Omega}}^H = \m{Y}_{\Omega}\m{Y}_{\Omega}^H$, such as the $M\times M$ matrix $\sbra{\m{Y}_{\Omega}\m{Y}_{\Omega}^H}^{\frac{1}{2}}$. If we replace $\m{Y}_{\Omega}$ by $\widetilde{\m{Y}_{\Omega}}$ in \eqref{eq:ANM_MS1}-\eqref{eq:ANM_MS3} and correspondingly change the dimensions of $\m{Z}$ and $\m{X}$, then the solution $\m{u}$ before and after the replacement is the same. Moreover, if we can find a matrix $\m{Q}$ satisfying $\m{Q}^H\m{Q}=\m{I}$ and $\widetilde{\m{Y}_{\Omega}} =\m{Y}_{\Omega} \m{Q}$ and if $\sbra{\widehat{\m{X}}, \widehat{\m{Z}},\widehat{\m{u}}}$ is the solution after the replacement, then the solution to the original problems is given by $\sbra{\m{Q}\widehat{\m{X}}\m{Q}^H, \widehat{\m{Z}}\m{Q}^H,\widehat{\m{u}}}$.
\label{thm:dimreduce}
\end{thm}

\begin{cor}[\cite{yang2016enhancing}] Theorem \eqref{thm:dimreduce} also holds true if the atomic norm is replaced by the weighted atomic norm.
\end{cor}

The dimensionality reduction technique provided by Theorem \ref{thm:dimreduce} enables us to reduce the number of snapshots from $L$ to $M$ and yet obtain the same solution $\m{u}$, from which the frequencies and the powers can be retrieved using the Vandermonde decomposition. Therefore, it follows from Theorem \ref{thm:dimreduce} that for ANM, like for GLS, the information in $\m{Y}_{\Omega}$ is preserved in the sample covariance matrix $\widetilde{\m{R}}_{\Omega}$. It is interesting to note that the above property even holds true in the presence of coherent sources, while we might expect that DOA estimation from $\widetilde{\m{R}}_{\Omega}$ could fail in such a case (consider MUSIC as an example).

\subsubsection{Alternating Direction Method of Multipliers (ADMM)}
A reasonably fast algorithm for SDPs is based on the ADMM \cite{boyd2011distributed,bhaskar2013atomic,yang2015gridless,yang2016enhancing}, which is a first-order algorithm that guarantees global optimality. To derive the ADMM algorithm, we consider \eqref{eq:ANM_MS1} as an example. Define
\equ{\cS= \lbra{\m{Z}:\; \frobn{\m{Z}_{\Omega} - \m{Y}_{\Omega}}\leq \eta}.}
Then, \eqref{eq:ANM_MS1} can be re-written as:
\equ{\begin{split}
&\min_{\m{X},\m{u},\m{Z}\in\cS, \m{\cQ}\geq\m{0}} \tr\sbra{\m{X}}+\tr\sbra{\m{T}\sbra{\m{u}}}, \\
&\st \m{\cQ} = \begin{bmatrix}\m{X} & \m{Z}^H \\ \m{Z} & \m{T}\sbra{\m{u}}\end{bmatrix}.
\end{split} \label{eq:ANM_MS_ADMM}}
We will derive the algorithm following the routine of ADMM by taking $\sbra{\m{X},\m{u},\m{Z}}$ and $\m{\cQ}$ as the two variables. We introduce $\m{\Lambda}$ as the Lagrangian multiplier and write the augmented Lagrange function for \eqref{eq:ANM_MS_ADMM} as
\equ{\begin{split}
\cL\sbra{\m{X},\m{u},\m{Z}, \m{\cQ}, \m{\Lambda}}=
& \tr\sbra{\m{X}}+\tr\sbra{\m{T}\sbra{\m{u}}} + \tr\mbra{\sbra{ \m{\cQ} - \begin{bmatrix}\m{X} & \m{Z}^H \\ \m{Z} & \m{T}\sbra{\m{u}}\end{bmatrix} }\m{\Lambda}}\\
&+ \frac{\beta}{2}\frobn{\m{\cQ} - \begin{bmatrix}\m{X} & \m{Z}^H \\ \m{Z} & \m{T}\sbra{\m{u}}\end{bmatrix}}^2,
\end{split}}
where $\beta>0$ is a penalty parameter set according to \cite{boyd2011distributed}. The algorithm is implemented by iteratively updating $\sbra{\m{X},\m{u},\m{Z}}$, $\m{\cQ}$ and $\m{\Lambda}$ as:
{\lentwo\equa{\sbra{\m{X},\m{u},\m{Z}}
&\leftarrow& \arg\min_{\m{X},\m{u},\m{Z}\in\cS} \cL\sbra{\m{X},\m{u},\m{Z}, \m{\cQ}, \m{\Lambda}}, \label{eq:ANM_ADMM_pv1}\\ \m{\cQ}
&\leftarrow& \arg\min_{\m{\cQ}\geq \m{0}} \cL\sbra{\m{X},\m{u},\m{Z}, \m{\cQ}, \m{\Lambda}}, \label{eq:ANM_ADMM_pv2}\\ \m{\Lambda}
&\leftarrow& \m{\Lambda} + \beta\sbra{\m{\cQ} - \begin{bmatrix}\m{X} & \m{Z}^H \\ \m{Z} & \m{T}\sbra{\m{u}}\end{bmatrix}}. \label{eq:ANM_ADMM_dv}
}}Note that $\cL$ in \eqref{eq:ANM_ADMM_pv1} is separable and quadratic in $\m{X}$, $\m{u}$ and $\m{Z}$. Consequently, these variables can be separately solved for in closed form (note that $\m{Z}$ can be obtained by firstly solving for $\m{Z}$ without the set constraint and then projecting the result onto $\cS$). In \eqref{eq:ANM_ADMM_pv2}, $\m{\cQ}$ can be similarly obtained by firstly solving for $\m{\cQ}$ without considering the constraint and then projecting the result onto the semidefinite cone, which can be accomplished by forming the eigen-decomposition and setting the negative eigenvalues to zero. The ADMM algorithm has a per-iteration computational complexity of $O(\sbra{N+L}^3)$ due to the eigen-decomposition. In the case of $L>M$, this complexity can be reduced to $O(\sbra{N+M}^3)$ by applying the dimensionality reduction technique presented in the preceding subsection. Although the ADMM may converge slowly to an extremely accurate solution, moderate accuracy is typically sufficient in practical applications \cite{boyd2011distributed}.


%
%
%
%

\section{Future Research Challenges} \label{sec:futurework}
In this section we highlight several research challenges that should be investigated in future studies.
\begin{itemize}
 \item \textbf{Improving speed and accuracy:} This is a permanent goal of the research on DOA estimation. As compared to most conventional approaches, in general, the sparse methods may have improved DOA estimation accuracy, especially in difficult scenarios, e.g., the cases with no prior knowledge on the source number, few snapshots and coherent sources. But they are more computationally expensive, which is especially true for the recent gridless sparse methods that typically need to solve an SDP. So efficient SDP solvers should be studied in future research, especially when the array size is large. Note that the computation of $\ell_1$ optimization has been greatly accelerated during the past decade with the development of compressed sensing.

     Furthermore, since the convex sparse methods may suffer from certain resolution limits, it will be of interest to study methods that can directly work with $\ell_0$ norms to enhance the resolution. In the case of ULA and SLA, the $\ell_0$ optimization corresponds to matrix rank minimization. Therefore, it is possible to apply the recent nonconvex optimization techniques for rank minimization to DOA estimation (see, e.g., \cite{zachariah2012alternating,bhojanapalli2015dropping, jain2013low, davenport2016overview, tu2015low, park2016finding}).  Recent progresses in this direction have been made in \cite{cho2016fast,cai2016fast}.

 \item \textbf{Automatic model order selection:} Different from the conventional subspace-based methods, the sparse optimization methods usually do not require {\em a prior} knowledge of the source number (a.k.a.~the model order). But this does not mean that the sparse optimization methods are able to accurately estimate the source number. Instead, small spurious sources are usually present in the obtained power spectrum. Therefore, it is of great interest to study how the model order can be automatically estimated within or after the use of the sparse methods. Some results in this direction can be found in \cite{yardibi2010source,tan2014direction,yang2015gridless}. In \cite{yang2015gridless} a parameter refining technique is also introduced once the model order is obtained.


 \item \textbf{Gridless sparse methods for arbitrary array geometry:} The grid-based sparse methods can be applied to any array geometry but they suffer from grid mismatch or other problems. In contrast to this, the gridless sparse methods completely bypass the grid selection problem by utilizing the special Hankel/Toeplitz structure of the sampled data or the data covariance matrix in the case of ULAs and SLAs. For a general array geometry, however, such structures do not exist any more and extension of the gridless sparse methods to general arrays should be studied in future. Recent results in this direction can be found in \cite{mahata2016frequency,mahata2017grid}.

 \item \textbf{Gridless sparse parameter estimation and continuous compressed sensing:} Following the line of the previous discussion, it would be of great interest to extend the existing gridless sparse methods for DOA estimation to general parameter estimation problems. Note that a similar data model, as used in DOA estimation, can be formulated for a rather general parameter estimation problem:
     \equ{\m{y} = \sum_{k=1}^K \m{a}\sbra{\theta_k} x_k + \m{e}, \label{eq:model_genparest}}
     where $\m{y}$ is the data vector, $\m{a}\sbra{\theta}$ is a given function of the continuous parameter $\theta$, $x_k$ are weight coefficients and $\m{e}$ denotes noise. Moreover, to guarantee that the parameters are identifiable as well as to simplify the model in \eqref{eq:model_genparest}, it is natural to assume that ``order'' $K$ is small and thus sparsity concept can be introduced as well. But due to the absence of special Hankel/Toeplitz structures, it would be challenging to develop gridless methods for \eqref{eq:model_genparest}. Note that the estimation of $\theta_k$ and $x_k$, $k=1,\dots,K$ from $\m{y}$ based on the data model in \eqref{eq:model_genparest} is also referred to as continuous or infinite-dimensional compressed sensing, which extends compressed sensing from the discrete to the continuous setting \cite{tang2012compressed,yang2014continuous,adcock2015generalized}.
\end{itemize}


\section{Conclusions} \label{sec:conclusion}
In this article, we provided an overview of the sparse DOA estimation techniques. Two key differences between sparse representation and DOA estimation were pointed out: 1) discrete system versus continuous parameters and 2) single versus multiple snapshots. Based on how the first difference is dealt with, the sparse methods were classified and discussed in three categories, namely on-grid, off-grid and gridless. The second difference can be tackled by exploiting the temporal redundancy of the snapshots. We explained that while the on-grid and off-grid sparse methods can be applied to arbitrary array geometries, they may suffer from grid mismatch, weak theoretical guarantees etc. These drawbacks can be eliminated by using the gridless sparse methods which, however, can only be applied to ULAs and SLAs. We also highlighted some challenging problems that should be studied in future research. Note that these sparse methods have diverse applications to many fields and the future work also includes performance comparisons of these methods for each specific application. Depending on data qualities and quantities, one or more of these methods may be favored in one application but not another.


\end{document}